\documentclass[11pt,a4paper]{article}

\usepackage{amssymb}
\usepackage{amsmath}
\usepackage{epsfig}
\usepackage{graphicx}
\usepackage{color}
\usepackage{bigints}
\usepackage{array}

\graphicspath{{./}{figures/}}

\newcommand{\fins}{\mathit{fins}}
\newcommand{\st}{\mathit{star}}
\newcommand{\shrink}{\mathit{shrink}}

\newcommand{\apex} {\mathit{apex}}
\newcommand{\tail} {\mathit{tail}}

\newcommand{\conv} {\mathrm{conv}}

\newcommand{\cancel}[1] {}

\newenvironment{emromani}
  {
   
   \begin{enumerate}}
  {\end{enumerate}}

\newenvironment{romani}
{
	
	\begin{enumerate}}
	{\end{enumerate}}

\newtheorem{definition}{Definition}[section]
\newtheorem{theorem}{Theorem}[section]
\newtheorem{lemma}{Lemma}[section]

\newbox\ProofSym
\setbox\ProofSym=\hbox{%
\unitlength=0.18ex%
\begin{picture}(10,10)
\put(0,0){\framebox(9,9){}}
\put(0,3){\framebox(6,6){}}
\end{picture}}

\begin{document}

\title{Adaptive Planar Point Location\footnote{A preliminary version appears in Proceedings of International Symposium on Computational Geometry, 2017~\cite{cheng2017}.  Research supported by Research Grants Council, Hong Kong, China (project no.~16201116).}}

\author{Siu-Wing Cheng\footnote{Department of Computer Science and Engineering,
HKUST, Clear Water Bay, Hong Kong. Email:{\tt \{scheng,lmkaa\}@cse.ust.hk}}
\and 
Man-Kit Lau\footnotemark[2]}

\date{}

\maketitle

\begin{abstract}
We present self-adjusting data structures for answering point location queries in convex and connected subdivisions.  Let $n$ be the number of vertices in a convex or connected subdivision.  Our structures use $O(n)$ space.  For any convex subdivision $S$, our method processes any online query sequence $\sigma$ in $O(\mathrm{OPT} + n)$ time, where {\rm OPT} is the minimum time required by any linear decision tree for answering point location queries in $S$ to process $\sigma$.  For connected subdivisions, the processing time is $O(\mathrm{OPT} + n + |\sigma|\log(\log^* n))$.  In both cases, the time bound includes the $O(n)$ preprocessing time.
\end{abstract}

\section{Introduction}

Planar point location is a fundamental problem in computational geometry that
has been studied extensively.  It calls for preprocessing a \emph{planar
subdivision} into a data structure so that for any query point, the region in
the subdivision that contains the query point can be reported.  In this paper, we are concerned with point location using point-line comparisons, which is a very popular model in the literature.  There are several common types of planar subdivisions.  A subdivision is \emph{convex} if every bounded region is a convex polygon and the outer boundary bounds a convex polygon.  A subdivision is \emph{connected} if every bounded region is a simple polygon and the outer boundary bounds a simple polygon.   A subdivision is \emph{general} if every region is a polygon possibly with holes.

Given a general subdivision with $n$ vertices, point location structures
with worst-case $O(\log n)$ query time, $O(n\log n)$ preprocessing time, and
$O(n)$ space are known~\cite{paper:Adamy1998,paper:Edelsbrunner1986,paper:Kirkpatrick1981,paper:Sarnak1986}.  For connected subdivisions, the preprocessing time can be reduced to $O(n)$~\cite{paper:Kirkpatrick1981} after triangulating every region in linear time~\cite{chazelle91}.

When processing a sequence of query points that fall into different regions with vastly different frequencies, one may consider objectives other than minimizing the worst-case query time to answer a single query.  One scenario is that the probability $p_r$ of the query point falling into each region $r \in S$ is given, so one may want to minimize the expected query time.  For subdivisions with regions of constant sizes, Arya, Malamatos, and Mount~\cite{paper:Arya2007a} and Iacono~\cite{paper:Iacono2004} proposed point location structures that use $O(n)$ space and  answer a query in $O(H)$ expected time, where $H$ denotes the entropy $\sum_{r \in S} p_r\log(1/p_r)$.  Later, Arya, Malamatos, Mount,
and Wong~\cite{paper:Arya2007b} improved the expected query time to $H +
O(\sqrt{H})$.  

Although the entropy is a lower bound for the expected query time according to Shannon's theory~\cite{shannon48}, it is a very weak lower bound when some regions have non-constant sizes.  Arya, Malamatos, Mount, and Wong~\cite{paper:Arya2007b} showed a convex polygon of $n$ sides and a query distribution such that a query point lies in the polygon with probability 1/2 and the expected number of point-line comparisons needed to decide whether a query point lies in the polygon is $\Omega(\log n)$.  Note that the entropy is only a constant.  To strengthen the lower bound, subsequent works consider linear decision trees for answering point location queries.  We call them \emph{point location linear decision trees} for convenience.  Each decision tree node contains a linear function which exactly models a point-line comparison.  Given a connected subdivision $S$ with $n$ vertices and the query distribution, Collete et al.~\cite{paper:Collette2012} designed a structure that uses $O(n)$ space and answers a query in $O(H^*)$ expected time, where $H^*$ is the minimum expected time needed by any point location linear decision tree for $S$.  When the query distribution is given, Afshani, Barbay, and Chan~\cite{paper:Afshani2015} and Bose et al.~\cite{paper:Bose2013} solved several geometric query problems, including planar point location in general subdivisions, with optimal expected performance with respect to linear decision trees.

A natural question arises: if no information about the query distribution is given beforehand, is there a way to process an online query sequence in time $o(\log n)$ times the sequence length?  (A query must be answered before the next query is processed.)  Although this is not always feasible by the example of Arya, Malamatos, Mount, and Wong that we mentioned previously, it may still be possible to achieve a running time that adapts to the inherent difficulty of a particular input.

A prime example in one dimension is the splay tree by Sleator and Tarjan~\cite{sleator85} for storing an ordered set of values.  Starting with an empty splay tree, any online sequence $\sigma$ of operations can be processed in $O(\sum_{v} p_v \log (|\sigma|/p_v))$ time, where the sum is over all values in the set and $p_v$ denotes the frequency of $v$ being accessed in $\sigma$.  The operations in $\sigma$ may include query, insert, and delete.  This result is known as the Static Optimality Theorem~\cite{sleator85}.  Note that $\sum_{v} p_v \log (|\sigma|/p_v)$ is the minimum time needed to process $\sigma$ by any static binary search tree.

For planar point location, Iacono and Mulzer~\cite{paper:Iacono2011} proposed a solution for triangulations.  Given a triangulation $T$ with $n$ vertices, their solution uses $O(n)$ space and processes any online query sequence $\sigma$ in $O(n + \sum_{t \in T} p_t \log (|\sigma|/p_t))$ time, where the sum is over all triangles in $T$ and $p_t$ denotes the frequency of a triangle $t$ being hit by a query point in $\sigma$.  The time bound already includes the time for preprocessing before handling the first query in $\sigma$.
Note that $\sum_t p_t \log (|\sigma|/p_t)$ is an information-theoretic lower bound for processing $\sigma$.  Recently, we designed a self-adjusting point location structure for convex subdivisions~\cite{cheng2015adaptive}: given a convex subdivision $S$, our structure uses $O(n)$ space and processes any online query sequence $\sigma$ in
$O(\mathrm{OPT} + n + |\sigma|\log\log n)$ time, where {\rm OPT} is the minimum
time needed by any point location linear decision tree for $S$ to process $\sigma$.

In this paper, we 
develop a self-adjusting point location structure that can process any online query sequence in $O(\mathrm{OPT} + n)$ time, where $n$ is the number of vertices in the input convex subdivision.  The time bound includes the $O(n)$ preprocessing time.  The space usage is $O(n)$.  Our result can be used in nearest neighbor searching as Voronoi diagrams are convex subdivisions.  Our technique extends to connected subdivisions, giving a running time of $O(\mathrm{OPT} + n + |\sigma|\log(\log^* n))$, where $\log^* n$ denotes the number of times the logarithm function is applied iteratively before the ouptut is less than or equal to one.  The function $\log(\log^* n)$ grows extremely slowly and it is no more than a small constant for practical values of $n$.  Again, the $O(n)$ preprocessing time is included and the space usage is $O(n)$.

We first present in Section~\ref{sec:hierarchy} a generic hierarchy of point location structures for processing an online query sequence adaptively.  The construction of the hierarchy depends on some parameters and algorithms to be supplied as input.  We analyze the performance of the hierarchy in terms of these parameters.  Then, we discuss how to design these parameters and algorithms for convex and connected subdivisions in order to obtain the claimed running times.

\section{Hierarchy of point location structures}
\label{sec:hierarchy}

If the query distribution is fixed and given, the best point location linear decision tree for a subdivision has the same asymptotic performance as the best point location linear decision tree for an ``optimal'' triangulation of that subdivision~\cite{paper:Collette2012}.  Therefore, an intuitive solution approach is to estimate the query distribution using the current access frequencies, optimally triangulate the subdivision based on this estimation as in~\cite{paper:Collette2012}, and then apply the result of Iacono and Mulzer~\cite{paper:Iacono2011}.  We state the result in~\cite{paper:Iacono2011} below for future reference. 

\begin{theorem}[\cite{paper:Iacono2011}] \label{thm:iacono}
	For any planar triangulation $T$ with $n$ vertices, there is a point-line comparison based data structure that uses $O(n)$ space and processes any online query sequence $\sigma$ in $O\left( \sum_{t \in T} p_t \log (|\sigma|/p_t) + n \right)$ time, where $p_t$ is the number of query points that fall into the triangle~$t$.  The time bound includes the $O(n)$ preprocessing time, and it matches the information-theoretic lower bound for processing $\sigma$.
\end{theorem}

There are two difficulties with this approach.  First, as the access frequencies evolve, the subdivision has to be retriangulated efficiently.  Second, the access frequencies give imprecise information about the query distribution, so error will be induced by the ``optimal'' triangulation based on the access frequencies.  It is unclear whether the error is small enough for our purposes. 

On the other hand, we showed in~\cite{cheng2015adaptive} that some canonical triangulation method (independent of the access frequencies) can lower the average extra cost per query to $O(\log\log n)$.  Our insight is to \emph{recursively} extract frequently accessed triangles and generate a separate point location structure for them using this canonical triangulation method.  This results in a multi-level structure.  The structure at the highest level is queried first and if that fails, we move down the levels.  We devise an analysis that handles both successful and unsuccessful queries at each level.  Intuitively, the performance improves as the number of levels increases, and thus we circumvent the difficulty of computing an ``optimal'' triangulation.

We describe how the hierarchy is constructed and queried in Section~\ref{sec:alg}.  The construction depends on some algorithm and parameters to be supplied as input.  We analyze the performance of the hierarchy in terms of these algorithm and parameters in Section~\ref{sec:analysis}.  

\subsection{Construction and querying}
\label{sec:alg}

We assume that an algorithm ${\cal A}(f,g)$ is given for constructing the hierarchy according to certain specifications that we will explain later.  The parameter $f$ is a function from $\mathbb{N}$ to $\mathbb{N}$ and the parameter $g$ is a function from $\mathbb{N} \times \mathbb{N}$ to $\mathbb{N}$.  They also need to satisfy some conditions to be explained later.

\paragraph{Framework.}  Let $B_S$ be some triangle that ${\cal A}(f,g)$ computes to enclose $S$.  So $S$ induces a planar subdivision inside $B_S$, and we denote this subdivision by $S_{1,1}$.  Let $R_{1,1}$ denote the set of bounded regions in $S_{1,1}$.  
For convenience, we abuse the notation slightly and use $n$ to denote the number of vertices in $S_{1,1}$, which is three more than the number of vertices in $S$.  

Initially, ${\cal A}(f,g)$ constructs a triangulation $\Delta_{1,1}$ from $S_{1,1}$ in $O(n)$ time, and a point location structure $D_{1,1}$ for $\Delta_{1,1}$.  The structures $R_{1,1}$, $S_{1,1}$, $\Delta_{1,1}$, and $D_{1,1}$ form the bottommost layer of the hierarchy.  The second subscript is unnecessary at this layer because $R_{1,1}$, $S_{1,1}$, $\Delta_{1,1}$ and $D_{1,1}$ will not be replaced, but this is not the case for higher layers.  At the $i$-th layer of the hierarchy for any $i > 1$, we may construct new structures from time to time to replace old ones.  Nonetheless, in the analysis, we need to refer to different versions of these structures, so the second subscript $j$ serves as the version index.  A larger $j$ refers to a more recent version.  We use $*$ in place of $j$ to refer to the most recent version.

The first layer consists of $R_{1,1}$, $S_{1,1}$, $\Delta_{1,1}$ and $D_{1,1}$ as mentioned before.  For $i > 1$, the $j$-th version of structures at the $i$-th layer consist of the following:
\begin{itemize}
	\item $R_{i,j}$: A set of regions, each being a subset of a region in $R_{i-1,l}$ for some $l$.  So each region in $R_{i,j}$ is a subset of some region in $S_{1,1}$.  The boundaries of two adjacent regions are compatible, i.e., a vertex of a region does not lie in the interior of an edge of an adjacent region.
	\item $S_{i,j}$: A planar subdivision that consists of regions in $R_{i,j}$ and additional triangles introduced to connect these regions together.  The vertices of these additional triangles are also vertices in $R_{i,j}$.  The outer boundary of $S_{i,j}$ is $B_S$, and there is no hole in $S_{i,j}$.	
	\item $\Delta_{i,j}$: A triangulation of $S_{i,j}$ that has the same asymptotic complexity as $S_{i,j}$.  For every triangle $t \in \Delta_{i,j}$, if  $t$ lies in some region $r \in R_{i,j}$, then $t$ stores the id of the region in $S_{1,1}$ that contains $r$.
	\item $D_{i,j}$: A point location structure for $\Delta_{i,j}$.
\end{itemize}

Let $P_i$ denote the total time to construct a version of structures at the $i$-th layer.  The algorithm ${\cal A}(f,g)$ is required to satisfy the conditions C1--C5 in Table~\ref{tb:C}.  Since $g(n,1) \geq 1$, C1 and C5 imply that $c_1 \geq c_2+2$.  The gap between $c_1$ and $c_2$ is actually larger when we appy the framework to convex and connected subdivisions.

\begin{table}
	\centering
	\renewcommand{\arraystretch}{1.3}
	\begin{tabular}{|c|p{5in}|}
	\hline 
	C1 & $f(k) = (\log_2 k)^{c_1}$ for some constant $c_1 \geq 5$. \\ 
	\hline
	C2 & $R_{i,j}$ has $O(n_i)$ size, where $n_1 = n$ and $n_i = f(n_{i-1})/\log n_{i-1}$ for $i \geq 2$. \\ 
	\hline
	C3 & $P_1 = O(n_1)$ and for all $i \geq 2$, $P_i = O(n_i\log n_{i-1})$. \\
	\hline
    C4 & There is a function $g(n,i)$ such that for any triangle $t$ that lies inside a region of $S$ (bounded or exterior), $t$ intersects $O(g(n,i))$ triangles in $\Delta_{i,j}$ that are subsets of regions in $R_{i,j}$.  Note $t$ may not lie completely inside any region in $R_{i,j}$. \\
    \hline
    C5 & There is a constant $c_2 \geq 1$ such that $f(n) \geq g(n,1) (\log_2 n)^{c_2+2}$ when $n$ is larger than some constant. \\
    \hline
\end{tabular}
\caption{Conditions C1--C5 that ${\cal A}(f,g)$ has to satisfy.}
\label{tb:C}
\end{table}

\paragraph{Querying.} $D_{i,j}$ consists of two separate point location structures.  The first one, denoted by $D_{i,j}'$, is obtained by invoking Theorem~\ref{thm:iacono} on $\Delta_{i,j}$.   The second one, denoted by $D_{i,j}''$, is a worst-case optimal planar point location structure (e.g.~\cite{paper:Kirkpatrick1981}).    When querying $D_{i,j}$ with a point $q$, we alternate between the search steps in $D_{i,j}'$ and $D''_{i,j}$, and stop as soon as a triangle $t$ in $\Delta_{i,j}$ containing $q$ is found. 

The overall querying procedure works as follows.  Let $m$ denote the index of the highest layer in the hierarchy currently.  When given a query point $q$, we first decide if $q$ lies outside $B_S$.  If so, we declare that $q$ is outside $S$ and stop.  Otherwise, we query $D_{i,*}$ with $q$ for $i = m, m-1, \ldots$ until $D_{i,*}$ reports a triangle in $\Delta_{i,*}$ that stores the id of a region in $S_{1,1}$.  That is, if $D_{m,*}$ returns a triangle that stores no such region id, we query $D_{m-1,*}$ with $q$, and so on.  We say that $q$ is \emph{successfully} located by $D_{i,*}$ if we query $D_{i,*}$ with $q$ and $D_{i,*}$ returns a triangle that stores the id of a region in $S_{1,1}$.  

After $q$ is located successfully by some $D_{i,*}$, we increment the access frequency of the triangle in $\Delta_{i,*}$ that contains $q$.  The access frequencies of triangles in $D_{i',*}$ for all $i' \not= i$ are not affected.  This update is important because the frequencies govern how the method in~\cite{paper:Iacono2011} will adjust $D'_{i,*}$ in order to adapt to incoming queries.

\paragraph{Layer construction.}  The construction of new structures at the $(i+1)$-th layer is triggered as soon as $D_{i,*}$ locates $f(n_i)$ query points successfully, provided that the conditions P1 and P2 in Table~\ref{tb:P} are satisfied.  In other words, if we reach the $i$-th layer such that $n_i < (c_1^2+c_1)^{c_1-1}$ or $f(n_i) < g(n,i) (\log_2 n_i)^{c_2+2}$, the $i$-th layer will remain the highest layer in the hierarchy for the rest of the execution of the algorithm no matter how many query points will be located successfully by $D_{i,*}$ in the future.  By C5 in Table~\ref{tb:C} and P2, the first layer is the highest layer in the hierarchy only if $n = O(1)$, in which case the problem is trivial.

\begin{table}
\centering
\renewcommand{\arraystretch}{1.3}
\begin{tabular}{|c|l|}
\hline
P1 & $f(n_i) \geq g(n,i) (\log_2 n_i)^{c_2+2}$\\
\hline
P2 & $n_i \geq (c_1^2+c_1)^{c_1-1}$\\
\hline
\end{tabular}
\caption{The values $c_1$ and $c_2$ are the constants in C1 and C5 in Table~\ref{tb:C}.  For the $(i+1)$-th layer to be built, both P1 and P2 have to be satisfied at the $i$-th layer.}
\label{tb:P}
\end{table}

Suppose that we are to create new structures at the $(i+1)$-th layer.  Let $j$ be the version index of the most recent structures at the $(i+1)$-th layer.  (If the $(i+1)$-th layer does not exist, we create a new layer and the version index starts from 1.)   The algorithm ${\cal A}(f,g)$ performs the following tasks.
\begin{romani}
	
	\item Select the $f(n_i)/(\log_2 n_i)^{c_2}$ most frequently accessed triangles in $\Delta_{i,*}$, where $c_2 \geq 1$ is the constant in C5 in Table~\ref{tb:C}.
	
	\item Build a set $R_{i+1,j+1}$ of regions such that each region has one or more of the $f(n_i)/(\log_2 n_i)^{c_2}$ selected triangles.  Then, construct $S_{i+1,j+1}$ and $\Delta_{i+1,j+1}$.  The complexities of  $R_{i+1,j+1}$, $S_{i+1,j+1}$ and $\Delta_{i+1,j+1}$ are required to be $O(n_{i+1}) = O(f(n_i)/(\log n_i))$.  This requirement will influence the choices of $c_1$ and $c_2$ in Table~\ref{tb:C}.
	
	\item Build a point location structure $D_{i+1,j+1}$ for answering queries in $\Delta_{i+1,j+1}$.  Recall that $D_{i+1,j+1}$ consists of two parts $D'_{i+1,j+1}$ and $D''_{i+1,j+1}$, where $D'_{i+1,j+1}$ is obtained by applying Theorem~\ref{thm:iacono} to $\Delta_{i+1,j+1}$, and $D''_{i+1,j+1}$ is a worst-case optimal point location structure for $\Delta_{i+1,j+1}$.
	
	\end{romani}

\subsection{Analysis}
\label{sec:analysis}

Let $D$ be a point location structure or a point location linear decision tree.  For any online query sequence $\alpha$, we use $D(\alpha)$ to denote the time taken by $D$ to process $\alpha$, which does not include the preprocessing time to construct $D$.

The next result follows from the fact that, according to the query procedure, access frequencies of triangles in $\Delta_{i,j}$ are only updated when $D_{i,j}$ successfully locates some query points.

\begin{lemma}\label{lem:generic-0}
	Let $\alpha$ be an online sequence of points that $D_{i,j}$ is queried with.  Let $\alpha_{i,j}$ be the maximal subsequence of $\alpha$ that are successfully located by $D_{i,j}$.  Then 
	\[
	D_{i,j}(\alpha) = O\left(D_{i,j}(\alpha_{i,j}) + |\alpha \!\setminus\! \alpha_{i,j}|\log n_i \right).
	\]
\end{lemma}
\begin{proof}
Only query points in $\alpha_{i,j}$ will cause the access frequencies in $\Delta_{i,j}$ to be updated.  Therefore, only $\alpha_{i,j}$ can affect the behavior of $D_{i,j}'$.  The time to process $\alpha \setminus \alpha_{i,j}$ is bounded from above by time taken by $D_{i,j}''$ on $\alpha \setminus \alpha_{i,j}$.  Therefore, $D_{i,j}(\alpha) = O\left(D_{i,j}(\alpha_{i,j}) + |\alpha \!\setminus\! \alpha_{i,j}|\log n_i\right)$.
\end{proof}

The next lemma covers two results.  First, an optimal linear decision tree for $S$ can be expanded into another linear decision tree that has the same asymptotic performance and induces a tiling of $S_{i,j}$ by triangles.  The second result relates the performance of $D_{i,j}$ to that of the optimal linear decision tree for $S$.

\begin{lemma}\label{lem:generic-1}
	Let $\sigma$ be an online query sequence.  Let $D^\sigma$ denote the point location linear decision tree for $S$ that takes the least processing time on $\sigma$.  
	\begin{emromani}
		
		\item There exists a point location linear decision tree $\widehat{D}^\sigma$ such that the leaf nodes of $\widehat{D}^\sigma$ induce a refinement of $S$ into a tiling of triangles, and every query point in $\sigma$ is processed by $\widehat{D}^\sigma$ in the same asymptotic running time as $D^\sigma$.
		
		\item For any online sequence $\alpha$ of points that $D_{i,j}$ is queried with, if $\alpha_{i,j}$ is the maximal subsequence of $\alpha$ that are successfully located by $D_{i,j}$, then 
		\[
		D_{i,j}(\alpha) = O\left(D^\sigma(\alpha_{i,j}) + n_i + |\alpha\!\setminus\!\alpha_{i,j}|\log n_i + |\alpha_{i,j}|\log g(n,i)\right) 
		\]
		
	\end{emromani}
\end{lemma}
\begin{proof}
Each leaf node $v$ of $D^\sigma$ corresponds to a convex polygon $\rho$ in a region of $S$.  If $\rho$ has $k$ sides, then $v$ has depth at least $k$ as each node of $D^\sigma$ applies a cut along a line.  Therefore, we can expand $v$ into a linear decision subtree so that the leaf nodes of this subtree correspond to a triangulation of $\rho$ and the height of this subtree is at most $k-2$.  The linear decision tree obtained by expanding $D^\sigma$ as described above is $\widehat{D}^\sigma$ stated in~(i).  The asymptotic query complexities of $\widehat{D}^\sigma$ and $D^\sigma$ on any point are the same.

Consider $D_{ij}$ for some $i$ and $j$.  Let $t$ be the triangle at a leaf node of $\widehat{D}^\sigma$.  We discuss how to expand this leaf node to a linear decision subtree for reporting a triangle in $\Delta_{i,j}$ that contains the query point.

Note that $t$ lies inside some region of $S$ (bounded or exterior) in order that $\widehat{D}^\sigma$ answers queries in $t$ correctly.  By C4 in Table~\ref{tb:C}, $t$ intersects $O(g(n,i))$ triangles in $\Delta_{i,j}$ that are subsets of regions in $R_{i,j}$.  These triangles refine $t$ into a planar subdivision $S_t$ of size $O(g(n,i))$.  We expand the leaf node of $\widehat{D}^\sigma$ storing $t$ to a linear decision subtree $L_t$ that performs point location in $S_t$ in $O(\log g(n,i))$ worst-case query time.  Some leaf nodes of $L_t$ may correspond to regions in $S_t$ that are disjoint from the regions in $R_{i,j}$.  We expand such leaf nodes further in order to find triangles in $\Delta_{i,j}$ that contain query points in $t \setminus \bigcup_{r \in R_{i,j}} r$.  However, we will not be interested in the query times for such query points, so we can expand such leaf nodes of $L_t$ arbitrarily.

Let $D$ denote the linear decision tree obtained by expanding $\widehat{D}^\sigma$ as described above.  Let $\alpha$ be an online query sequence that $D_{i,j}$ is queried with.  Let $\alpha_{i,j} \subseteq \alpha$ be the maximal subsequence of points that are successfully located by $D_{i,j}$.  For each query point $q \in \alpha_{i,j}$, the search in $D$ traverses a root-to-leaf path in $\widehat{D}^\sigma$ and then
another path of length $O(\log g(n,i))$ to a leaf of $D$.  Therefore, \begin{equation}
D(\alpha_{i,j}) = O(\widehat{D}^\sigma(\alpha_{i,j}) + |\alpha_{i,j}|\log g(n,i)) = O(D^\sigma(\alpha_{i,j}) + |\alpha_{i,j}|\log g(n,i)).  \label{eq:generic-1} 
\end{equation}
Then,
\[
\begin{array}{clcl}
& D_{i,j}(\alpha) \\
= & O\bigl(D_{i,j}(\alpha_{i,j}) + |\alpha\!\setminus\!\alpha_{i,j}|\log n_i \bigr)  && (\because \text{Lemma~\ref{lem:generic-0}}) \\
= & O\bigl(D(\alpha_{i,j}) + n_i + |\alpha\!\setminus\!\alpha_{i,j}|\log n_i\bigr)  && (\because \text{Theorem~\ref{thm:iacono}}) \\
= & O\bigl(D^\sigma(\alpha_{i,j}) + n_i + |\alpha_{i,j}|\log g(n_i) + |\alpha\!\setminus\!\alpha_{i,j}|\log n_i\bigr). && (\because \eqref{eq:generic-1})
\end{array}
\]
\end{proof}

For any online query sequence $\sigma$ processed by our hierarchy of point location structures, define the following subsequences of query points:
\begin{eqnarray*}
\sigma_{i,j} & = & \{q \in \sigma: \mbox{$q$ is located successfully by $D_{i,j}$ }\}, \\
\sigma_{i} & = &
\{ q \in \sigma: \mbox{$q$ is located successfully at the $i$-th layer of the hierarchy}\}, \\
\sigma_{<i} & = &
\{ q \in \sigma: \mbox{$q$ is located successfully at some layer below the $i$-th layer}\}.
\end{eqnarray*}
By definition, $\sigma = \bigcup_{i,j} \sigma_{i,j}$, the $\sigma_{i,j}$'s
are mutually disjoint, $\sigma_i = \bigcup_j \sigma_{i,j}$,
and $\sigma_{<i} = \bigcup_{k=1}^{i-1} \sigma_k$.  Consider the linear decision tree $\widehat{D}^\sigma$ in Lemma~\ref{lem:generic-1}(i).  Define the following subsequences:
\begin{eqnarray*}
\hat{\sigma}_{i,j} & = & \{q \in \sigma_{i,j} : \mbox{$q$ lies in some leaf node of $\widehat{D}^\sigma$ at depth $\log_2\log_2 n_i$ or less} \}, \\
\hat{\sigma}_i & = & \bigcup_j \hat{\sigma}_{i,j}.
\end{eqnarray*}

\begin{lemma}\label{lem:generic-2}
	If $n_i \geq (c_1^2+c_1)^{c_1-1}$ and $f(n_i) \geq g(n,i)(\log n_i)^{c_2+2}$, then 
	\[
	|\hat{\sigma}_{i,j}| = O\left(f(n_i)^2/(\log n_i)^{c_2+1} + |\sigma_{i,j}|/\log n_i\right),
	\] 
	where $c_1$ and $c_2$ are the constants in Tables~\ref{tb:C} and~\ref{tb:P}.
\end{lemma}
\begin{proof}
	We will make use of the following facts:
	\begin{quote}
		\noindent Fact~1:~At most $2^{1 + \log_2\log_2 n_i} - 1 = 2\log_2 n_i - 1$ nodes in $\widehat{D}^\sigma$ have depth at most $\log_2\log_2 n_i$ because $\widehat{D}^\sigma$ is a binary tree.  (The root has depth 0.)
		
		\vspace{8pt}
		
		\noindent Fact~2:~Suppose that $\Delta_{i,j}$ is the most recent triangulation at the $i$-th layer.  For each triangle $t \in \Delta_{i,j}$, if the current access frequency of $t$ is at least $|\sigma_{i,j}|(\log_2 n_i)^{c_2}/f(n_i)$, then $t$ must be selected for the next construction of stuctures at the $(i+1)$-th layer.  The reason is that the sum of frequencies in $\Delta_{i,j}$ is at most $|\sigma_{i,j}|$, so the $(f(n_i)/(\log_2 n_i)^{c_2} + 1)$-th highest frequency is less than $|\sigma_{i,j}|(\log_2 n_i)^{c_2}/f(n_i)$, implying that the frequency of $t$ must be among the top $f(n_i)/(\log_2 n_i)^{c_2}$ frequencies in $\Delta_{i,j}$.
		
	\end{quote}
	
	Let $\Delta'_{i,j}$ be the subset of triangles in $\Delta_{i,j}$ that store the ids of some regions of $S_{1,1}$.  Note that all query points in $\sigma_{i,j}$ lie in triangles in $\Delta'_{i,j}$.  Let $Z$ be the subset of triangles in $\Delta'_{i,j}$ that overlap with the regions associated with leaf nodes in $\widehat{D}^\sigma$ at depth $\log_2\log_2 n_i$ or less.
		
	Let $\tau$ be a triangle at a leaf node of $\widehat{D}^\sigma$ at depth $\log_2\log_2 n_i$ or less.  The triangle $\tau$ must lie inside a region (bounded or exterior) of $S$ in order that $\widehat{D}^\sigma$ answers queries correctly. Therefore, $\tau$ intersects $O(g(n,i))$ triangles in $\Delta'_{i,j}$ by condition C4 in Table~\ref{tb:C}.   By Fact~1 and the assumption that $f(n_i) \geq g(n,i)(\log_2 n_i)^{c_2+2}$, we obtain 
	\begin{equation}
	|Z| = O(g(n,i) \log n_i) = O\left(f(n_i)/(\log n_i)^{c_2+1}\right). \label{eq:a-0}
	\end{equation}
	
	Consider a triangle $t \in \Delta'_{i,j}$ that contains a query point in
	$\hat{\sigma}_{i,j}$.  Thus, $t \in Z$ because $t$ must overlap with some
	triangle (leaf node) in $\widehat{D}_\sigma$ at depth $\log_2\log_2 n_i$ or less.  If the frequency of $t$ in $\Delta_{i,j}$ is smaller than $|\sigma_{i,j}|(\log_2 n_i)^{c_2}/f(n_i)$, then fewer than $|\sigma_{i,j}|(\log_2 n_i)^{c_2}/f(n_i)$ query points in $t$ are from $\hat{\sigma}_{i,j}$.  Suppose that the frequency of $t$ in $\Delta_{i,j}$ reaches $|\sigma_{i,j}|(\log_2 n_i)^{c_2}/f(n_i)$, say after the construction of $\Delta_{i+1,l}$ and before the construction of $\Delta_{i+1,l+1}$.  At most $f(n_i)$ queries can be answered by $D_{i,j}$ during this period.\footnote{The assumptions of $n_i \geq (c_1^2+c_1)^{c_1-1}$ and $f(n_i) \geq g(n,i)(\log_2 n_i)^{c_2+2}$ are implicitly used here in order that the construction of $\Delta_{i+1,l+1}$ will be triggered.}  It means that the frequency of $t$ in $\Delta_{i,j}$ is at most $f(n_i) + |\sigma_{i,j}|(\log_2 n_i)^{c_2} /f(n_i)$ before the construction of $\Delta_{i+1,l+1}$.  By Fact~2, $t$ will be included in forming $\Delta_{i+1,l'}$ for all $l' > l$ until $\Delta_{i,j}$ is replaced by $\Delta_{i,j+1}$, at which point $\sigma_{i,j}$ has been exhausted.  While $\Delta_{i,j}$ is the most recent structure at the $i$-th layer, every query point that falls in $t$ after the construction of $\Delta_{i+1,l+1}$ will be located successfully in $S$ at level $i+1$ or higher.  Thus, the frequency of $t$ in $\Delta_{i,j}$ will not be increased further and at most $f(n_i) + |\sigma_{i,j}|(\log_2 n_i)^{c_2}/f(n_i)$ query points in $t$ are from $\hat{\sigma}_{i,j}$.  Hence, $|\hat{\sigma}_{i,j}| \leq (f(n_i) + |\sigma_{i,j}|(\log_2 n_i)^{c_2}/f(n_i)) \cdot |Z|$.  Since $|Z| = O\left(f(n_i)/(\log n_i)^{c_2+1}\right)$ by \eqref{eq:a-0}, we obtain $|\hat{\sigma}_{i,j}| = O\left(f(n_i)^2/(\log n_i)^{c_2+1} + |\sigma_{i,j}|/\log n_i\right)$.  
\end{proof}	

We need one more technical result before analyzing the performance of the hierarchy of point location structures.

\begin{lemma}\label{lem:claim}
	Let $m \geq 2$ be the index of some layer in the hierarchy in processing some online query sequence.  
	\begin{emromani}
		\item For all $i \in [1,m-1]$, $m-i \leq \log_2^* n_i$, where $\log_2^* n_i$ is the number of times that the logarithm function is applied to $n_i$ so that the output is less than or equal to 1.
		\item For all $i \in [1,m-1]$, $\sum_{l=i+1}^m \log_2 n_l < (3c_1^2+2c_1)\log_2\log_2 n_i$, where $c_1$ is the constant in Table~\ref{tb:C}.
	\end{emromani}
\end{lemma}
\begin{proof}
	Consider any index $l \in [i+1,m]$.  Since $f(n_{l-1}) = (\log_2 n_{l-1})^{c_1}$ and $n_l = f(n_{l-1})/(\log_2 n_{l-1})$ by C1 and C2 in Table~\ref{tb:C}, we get
	\begin{equation*}
	\log_2 n_l = \log_2 \frac{f(n_{l-1})}{\log_2 n_{l-1}} <  c_1\log_2\log_2 n_{l-1}.
	\end{equation*}
	If $l \geq i+2$, then we can similarly expand the equation and obtain 
	\begin{equation*}
	\log_2 n_l < c_1\log_2(c_1\log_2\log_2 n_{l-2}) = c_1\log_2 c_1 + c_1\log_2\log_2\log_2 n_{l-2}.  
	\end{equation*}
	If $l \geq i+3$, then we can expand further and get 
	\begin{equation*}
	\log_2 n_l < c_1\log_2 c_1 + c_1\log_2\log_2 c_1 + c_1\log_2\log_2\log_2\log_2 n_{l-3}.
	\end{equation*}
	
	Use $\log_2^{(k)} a$ to denote $k$ successive applications of the logarithm function to $a$.  Then, we can inductively obtain
	\[
	\log_2 n_l < \left(\sum_{k=1}^{l-i-1} c_1\log_2^{(k)} c_1\right) + c_1\log_2^{(l-i+1)} n_i.
	\]
	It is clear that $\log_2^{(k)} a \leq a/2^k$ for $a \geq 4$.  As $c_1 \geq 5$ in Table~\ref{tb:C}, we obtain $\sum_{k=1}^{l-i-1} c_1\log_2^{(k)} c_1 < \sum_{k=1}^\infty c_1^2/2^k = c_1^2$, which implies that
	\[
	\log_2 n_l < c_1^2 + c_1\log_2^{(l-i+1)} n_i.
	\]
	When $l-i+1 = \log_2^* n_i$, we have $\log_2 n_l < c_1^2+c_1$ and hence $n_{l+1} = (\log_2 n_l)^{c_1-1} < (c_1^2+c_1)^{c_1-1}$.  Since we would not have built the $m$-th level if $n_{m-1} < (c_1^2+c_1)^{c_1-1}$, we conclude that $m-i \leq \log_2^* n_i$.  Note that this establishes the correctness of (i).  Therefore,
	\begin{equation*}
	\sum_{l=i+1}^m \log_2 n_l < c_1^2(m-i) + \sum_{l=i+1}^m c_1\log_2^{(l-i+1)} n_i \\
	\leq c_1^2\log^* n_i + \sum_{l=i+1}^m c_1\log_2^{(l-i+1)} n_i. 
	\end{equation*}
	Since $i \leq m-1$, P2 in Table~\ref{tb:P} must hold for the $(i+1)$-th layer to be built.  This implies that $n_i \geq (c_1^2 + c_1)^{c_1-1} \geq 30^4$ as $c_1 \geq 5$.  Thus, $\log_2\log_2 n_i > 4$ which makes $\log_2^{(k)} (\log_2\log_2 n_i) \leq (\log_2\log_2 n_i)/2^k$.  Then, 
	\[
	\log_2^{(l-i+1)} n_i = \log_2^{(l-i-1)} (\log_2\log_2 n_i) \leq (\log_2\log_2 n_i)/2^{l-i-1}. 
	\]
	Therefore, 
	\[
	\sum_{l=i+1}^m c_1\log_2^{(l-i+1)} n_i \leq \sum_{l=i+1}^m c_1(\log_2\log_2 n_i)/2^{l-i-1} < 2c_1\log_2 \log_2 n_i.
	\]
	Consequently,
	\[
	\sum_{l=i+1}^m \log_2 n_l < c_1^2\log_2^* n_i + 2c_1\log_2\log_2 n_i \leq (3c_1^2+2c_1)\log_2\log_2 n_i,
	\]
	establishing the correctness of (ii).
\end{proof}
	
Finally, we are ready to analyze the performance of the hierarchy of point location structures.

\begin{theorem}\label{thm:generic-3}
	Given a subdivision $S$ with $n$ vertices and an algorithm ${\cal A}(f,g)$, we can process any online query sequence $\sigma$ with a hierarchy of point-line comparison based point location structures.  Let $m$ be the index of the highest layer in the hierarchy in processing $\sigma$.  The hierarchy uses $O(n)$ space and processes $\sigma$ in $T(n)$ time, where
	\[
	T(n) = \left\{\begin{array}{lcl}
	O(\mathrm{OPT} + n) & & \mbox{if $n_m \geq (c_1^2+c_1)^{c_1-1}$ and} \\
	                    & & \;\;\;\mbox{$f(n_m) \geq g(n,m)(\log_2 n_m)^{c_2+2}$}, \\
	                    \\
	O(\mathrm{OPT} + n + |\sigma_m|\log g(n,m)) & & \mbox{otherwise}.
	\end{array}\right.
	\]
\end{theorem}	
\begin{proof}
	Let $D^\sigma$ be the point location linear decision tree for $S$ that takes the least time to process $\sigma$.  Let $\Gamma_i$ denote the total processing time required by $D_{i,j}$ over all $j$, including the preprocessing time $P_i$ for each $D_{i,j}$.  Recall that $P_i = O(n_i\log n_{i-1})$ for $i > 1$ and $P_1 = O(n_1)$.  For $i > 1$, $\Gamma_i$ includes the time spent on unsuccessfully locating some query points in $\sigma_{<i}$.  By Lemma~\ref{lem:generic-1}(ii), since all queries in $\sigma_1$ are successfully located by $D_{1,1}$, we get
	\[
	\Gamma_1 = D_{1,1}(\sigma_1) + P_1 = O({D}^\sigma(\sigma_1) + n + |\sigma_1|\log g(n,1))，
	\]  	
	and for $i \in [2,m]$,
	\begin{eqnarray*}
		\Gamma_i & = & O\left(\sum_j {D}^\sigma(\sigma_{i,j}) + \sum_j n_i + |\sigma_{<i}|\log n_i + \sum_j |\sigma_{i,j}|\log g(n,i) \right) + \sum_j P_i \\
		& = & O\left({D}^\sigma(\sigma_i) + \sum_j n_i\log n_{i-1}
		+ |\sigma_{<i}|\log n_i + |\sigma_i|\log g(n,i) \right).
		\end{eqnarray*}
		Therefore,
		\[
		\sum_{i=1}^m \Gamma_i =
		O\left(\sum_{i=1}^m {D}^\sigma(\sigma_i) + n + \sum_{i=2}^m \sum_j n_i\log n_i + \sum_{i=1}^m |\sigma_i|\log g(n,i) + \sum_{i=2}^m |\sigma_{<i}|\log n_i\right).
		\]
		Since $\Delta_{i,j}$ is constructed after answering $f(n_{i-1})$ new queries using $D_{i-1,*}$, the preprocessing time of $P_i = O(n_i\log n_{i-1}) = O(f(n_{i-1}))$ can be charged to these new queries.  So $\sum_{i=2}^m \sum_j n_i\log n_{i-1}$ can be charged to the queries in
		$\sigma$, i.e., $\sum_{i=2}^m \sum_j n_i\log n_{i-1} = O(|\sigma|)$.  We rewrite the term $\sum_{i=2}^m |\sigma_{<i}|\log n_i = \sum_{i=1}^{m-1}
		(|\sigma_i|\sum_{l=i+1}^m \log n_l)$.  By Lemma~\ref{lem:claim}(ii), for $i \in [1,m-1]$, $\sum_{l=i+1}^m \log_2 n_l = O(\log\log n_i)$.  Consequently,
		\begin{equation}
		\sum_{i=1}^m \Gamma_i = 
		O\left(\sum_{i=1}^m {D}^\sigma(\sigma_i) + n + |\sigma| +
		\sum_{i=1}^m |\sigma_i|\log g(n,i) + \sum_{i=1}^{m-1} |\sigma_i|\log\log n_i\right).  \label{eq:3-1}
		\end{equation}
		We stop building any new layer once P1 in Table~\ref{tb:P} is violated, it follows that for $i \in [1,m-1]$, $f(n_i) \geq g(n,i)(\log_2 n_i)^{c_2+2}$ which gives $\log g(n,i) = O(\log f(n_i)) = O(\log\log n_i)$.  So we can write
		\begin{equation*}
		\sum_{i=1}^m |\sigma_i|\log g(n,i) + \sum_{i=1}^{m-1} |\sigma_i|\log\log n_i = O\left(|\sigma_m|\log g(n,m) + \sum_{i=1}^{m-1} |\sigma_i|\log \log n_i\right).
		\end{equation*}
		Moreover, if $f(n_m) \geq g(n,m)(\log_2 n_m)^{c_2+2}$, we can simplify further by the same reasoning:
		\begin{equation*}
		\sum_{i=1}^m |\sigma_i|\log g(n,i) + \sum_{i=1}^{m-1} |\sigma_i|\log\log n_i = O\left(\sum_{i=1}^m |\sigma_i|\log \log n_i\right).
		\end{equation*}
		Let $m' = m$ if $n_m \geq (c_1^2+c_1)^{c_1-1}$ and $f(n_m) \geq g(n,m)(\log_2 n_m)^{c_2+2}$, and let $m' = m - 1$ otherwise.  We can rewrite \eqref{eq:3-1} as
		\begin{equation}
		\sum_{i=1}^m \Gamma_i = O\left({D}^\sigma(\sigma) + n + |\sigma| + 
		\Lambda + \sum_{i=1}^{m'} |\sigma_i|\log\log n_i\right),  \label{eq:3}
		\end{equation}
		where
		\[
		\Lambda = \left\{\begin{array}{lcl}
		0 & & \mbox{if $n_m \geq (c_1^2+c_1)^{c_1-1}$ and } \\
		  & & \;\;\;\mbox{$f(n_m) \geq g(n,m)(\log_2 n_m)^{c_2+2}$}, \\
		  \\
		|\sigma_m|\log g(n,m) & & \mbox{otherwise}.
		\end{array}\right.
		\]
		Recall that $\widehat{D}^\sigma$ is the point location decision tree in Lemma~\ref{lem:generic-1}(i), $\hat{\sigma}_{i,j}$ consists of query points in $\sigma_{i,j}$ that are located successfully at leaves of $\widehat{D}^\sigma$ at depth $\log_2 \log_2 n_i$ or less, and $\hat{\sigma}_i = \bigcup_j \hat{\sigma}_{i,j}$.
		By Lemma~\ref{lem:generic-2}, for $i \in [1,m']$,
		\begin{equation}
		|\hat{\sigma}_i| = \sum_j |\hat{\sigma}_{i,j}| =
		O\left(\sum_j f(n_i)^2/(\log n_i)^{c_2+1} + |\sigma_i|/\log n_i \right).  \label{eq:3-0-0}
		\end{equation}
		If $|\sigma_i\! \setminus \! \hat{\sigma}_i| \geq |\hat{\sigma}_i|$, then
		\begin{eqnarray*}
		|\sigma_i|\log\log n_i & = & |\sigma_i \! \setminus \! \hat{\sigma}_i|\log\log n_i + |\hat{\sigma}_i| \log \log n_i \\
		& = & O(|\sigma_i \! \setminus \! \hat{\sigma}_i| \log\log n_i ) \\
		& = & O({D}^{\sigma}(\sigma_i)).
		\end{eqnarray*}
		In the last step above, we use two facts.  First, $\widehat{D}^{\sigma}(\sigma_i) =
		\Omega(|\sigma_i\! \setminus \! \hat{\sigma}_i|\log\log n_i)$, which is true because each query point in $\sigma_i \! \setminus \! \hat{\sigma}_i$ lies in a triangle (leaf node) in $\widehat{D}^\sigma$ at depth greater than $\log_2\log_2 n_i$.  Second, $\widehat{D}^\sigma(\sigma_i) = \Theta(D^\sigma(\sigma_i))$ by Lemma~\ref{lem:generic-1}(i).  If $|\sigma_i \! \setminus \! \hat{\sigma}_i| < |\hat{\sigma}_i|$, then
		\begin{eqnarray*}
		|\sigma_i|\log\log n_i & = & |\sigma_i \! \setminus \! \hat{\sigma}_i|\log\log n_i + |\hat{\sigma}_i|\log \log n_i \\
		& = & O(|\hat{\sigma}_i|\log \log n_i) \\
		& \stackrel{\eqref{eq:3-0-0}}{=} & O\left(\sum_j f(n_i)^2/(\log n_i)^{c_2} + |\sigma_i|\right) \\
		& = & O\left(\sum_j n_i + |\sigma_i| \right).
		\end{eqnarray*}
		In the last step above, we use the fact that $c_1 \geq c_2+2$ and $f(n_i) = (\log_2 n_i)^{c_1}$.  So $f(n_i)^2/(\log n_i)^{c_2} = O(n_i)$.  Combining the two cases above and the fact that ${D}^\sigma(\sigma_i) = \Omega(|\sigma_i|)$,
		we obtain $|\sigma_i|\log\log n_i = O\bigl({D}^\sigma(\sigma_i) + \sum_j n_i\bigr)$.
		Substituting this equation into \eqref{eq:3} gives
		\[
		\sum_{i=1}^m \Gamma_i = O\left({D}^\sigma(\sigma) + n + |\sigma| + \Lambda +
		\sum_{i=1}^{m'} {D}^\sigma(\sigma_i) + \sum_{i=1}^{m'} \sum_j n_i \right).
		\]
		We have shown previously that $\sum_{i=2}^m \sum_j n_i\log n_{i-1} = O(|\sigma|)$.  Note that $\sum_j n_1 = n_1 = n$ as the structures at the first layer are never replaced.  Also, $\sum_{i=1}^{m'} {D}^\sigma(\sigma_i) \leq {D}^\sigma(\sigma)$ and ${D}^\sigma(\sigma) = \Omega(|\sigma|)$.  Therefore,
		\[
		\sum_{i=1}^m \Gamma_i = O(D^\sigma(\sigma) + n + \Lambda) = O(\mathrm{OPT} + n + \Lambda).
		\]
		
		The total size of our data structure is $O(\sum_{i=1}^m n_i)$.  Recall that $n_i = O(f(n_{i-1})/\log n_{i-1}) = O((\log n_{i-1})^{c_1-1})$.  It is clear that $n_i = O(n/2^{i-1})$ by an inductive argument.  As a result, the total size is $O(\sum_{i=1}^m n_i) = O(\sum_{i=1}^m n/2^{i-1}) = O(n)$.
\end{proof}

\section{Convex subdivision}
\label{sec:convex}

We present an algorithm ${\cal A}(f,g)$ for convex subdivisions so that Theorem~\ref{thm:generic-3} can be invoked to give an optimal processing time of $O(\mathrm{OPT} + n)$.  We first review the canonical triangulation in~\cite{cheng2015adaptive} which will be used later.  

\subsection{Canonical triangulation}

Let $S$ denote a convex subdivision with $n$ vertices.  For each bounded region
$r$ in $S$, the procedure {\sf TriReg} is called to triangulate $r$.
Figure~\ref{fg:region} gives an example.  Let $\mathit{size}(r)$ denote the complexity of $r$.  {\sf TriReg} runs in $O(\mathit{size}(r))$ time
and produces a triangulation of $O(\mathit{size}(r))$ size.  This triangulation method was
first introduced by Dobkin and Kirkpatrick for convex polygon intersection
detection~\cite{Dobkin1990}.   Every line segment in $r$ intersects $O(\log
\mathit{size}(r))$ triangles.

\begin{quote}
	\noindent {\sf TriReg}$(r)$
	\begin{enumerate}
		\item If $r$ is a line segment or triangle, then return.
		\item Take any maximum subsequence $\alpha$ of vertices of $r$ such that no two
		vertices in $\alpha$ are adjacent along the boundary of $r$, except possibly
		the first and the last ones.
		\item Connect the vertices in $\alpha$ to form a convex polygon $r'$.
		\item Call {\sf TriReg}$(r')$.
	\end{enumerate}
\end{quote}

Second, we triangulate the exterior region of $S$.  We pick three boundary
edges of $S$ such that removing them gives three boundary chains of $S$ of
roughly equal sizes.  The support lines of these three edges bound a triangle,
denoted by $B_S$, that contains $S$.\footnote{It is possible that $B_S$ is unbounded, but we will assume that $B_S$ is bounded for simplicity.} We call the three interior-disjoint regions between the boundaries
of $B_S$ and $S$ \emph{fins}.\footnote{We call them fins instead of boomerangs as in~\cite{cheng2015adaptive}.}   

\begin{definition}
	The boundary of each fin has three distinguished parts: the head, which consists of two straight sides, the apex, which is the convex vertex incident to the two head sides, and the tail, which is a reflex chain attached to the non-apex vertices of the head.  Figure~\ref{fg:boomerang}(a) gives an example.  For
	each fin~$\alpha$, we denote its apex and tail by $\apex(\alpha)$ and $\tail(\alpha)$, respectively.
\end{definition}

For each fin $\alpha$ between the boundaries of $B_S$ and $S$, we call the procedure {\sf SplitFin} to partition $\alpha$ hierarchically into triangular regions as well as to construct a binary tree that represents this hierarchy.  

\begin{definition}
	Given a fin $\alpha$, we use $T_\alpha$ to denote the binary tree representing the hierarchy of triangular regions produced by {\sf SplitFin}.\footnote{The binary tree $T_\alpha$ is not constructed in~\cite{cheng2015adaptive}, but we will need it.} Figure~\ref{fg:split} gives an example.   Let $\mathit{size}(\alpha)$ denote the complexity of $\alpha$.  {\sf SplitFin} runs in $O(\mathit{size}(\alpha))$ time. $T_\alpha$ has $O(\mathit{size}(\alpha))$ size and $O(\log \mathit{size}(\alpha))$ height.
\end{definition}

\begin{quote}
	\noindent {\sf SplitFin}$(\alpha)$
	\begin{enumerate}
		\item If $\alpha$ is a triangle, then return.	
		\item Take the middle edge $e$ of $\tail(\alpha)$.
		\item Divide $\alpha$ with the support line of $e$ into a triangle $t$ and two smaller fins
		$\alpha_1$ and $\alpha_2$.
		\item Call {\sf SplitFin}$(\alpha_1)$ and {\sf SplitFin}$(\alpha_2)$ to obtain $T_{\alpha_1}$ and $T_{\alpha_2}$.
		\item Create the binary tree $T_\alpha$ with root $v$ containing $t$.  Make $T_{\alpha_1}$ and $T_{\alpha_2}$
		left and right subtrees of $v$.
		\item Return $T_\alpha$.
	\end{enumerate}
\end{quote}

\begin{figure}
	\centerline{\includegraphics[scale=0.6]{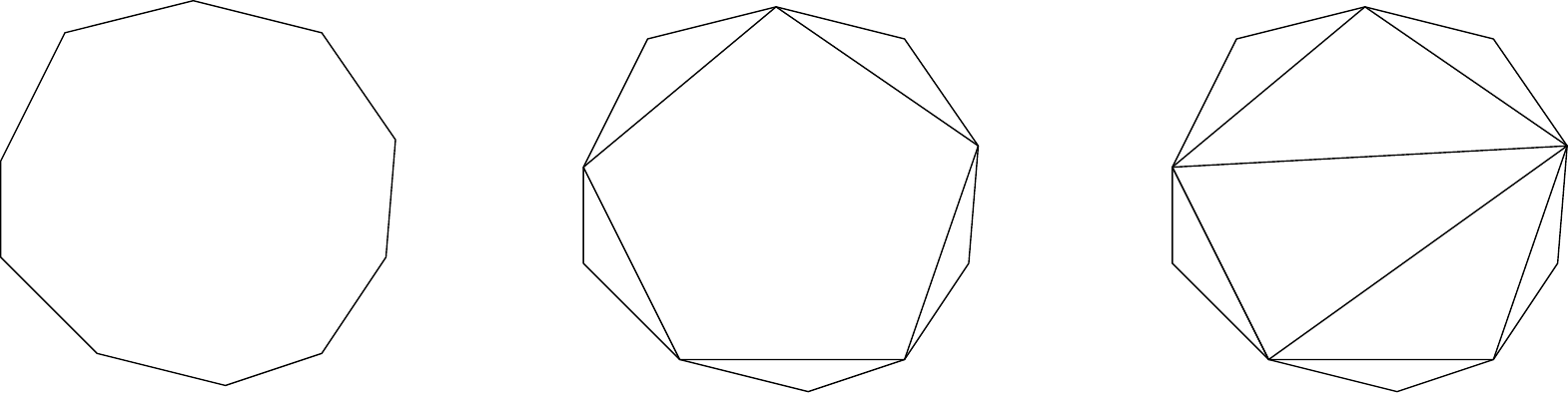}}
	\caption{Triangulation of a bounded region in $S$.}
	\label{fg:region}
\end{figure}

\begin{figure}
	\centering
	\begin{tabular}{ccc}
		\includegraphics[scale=0.4]{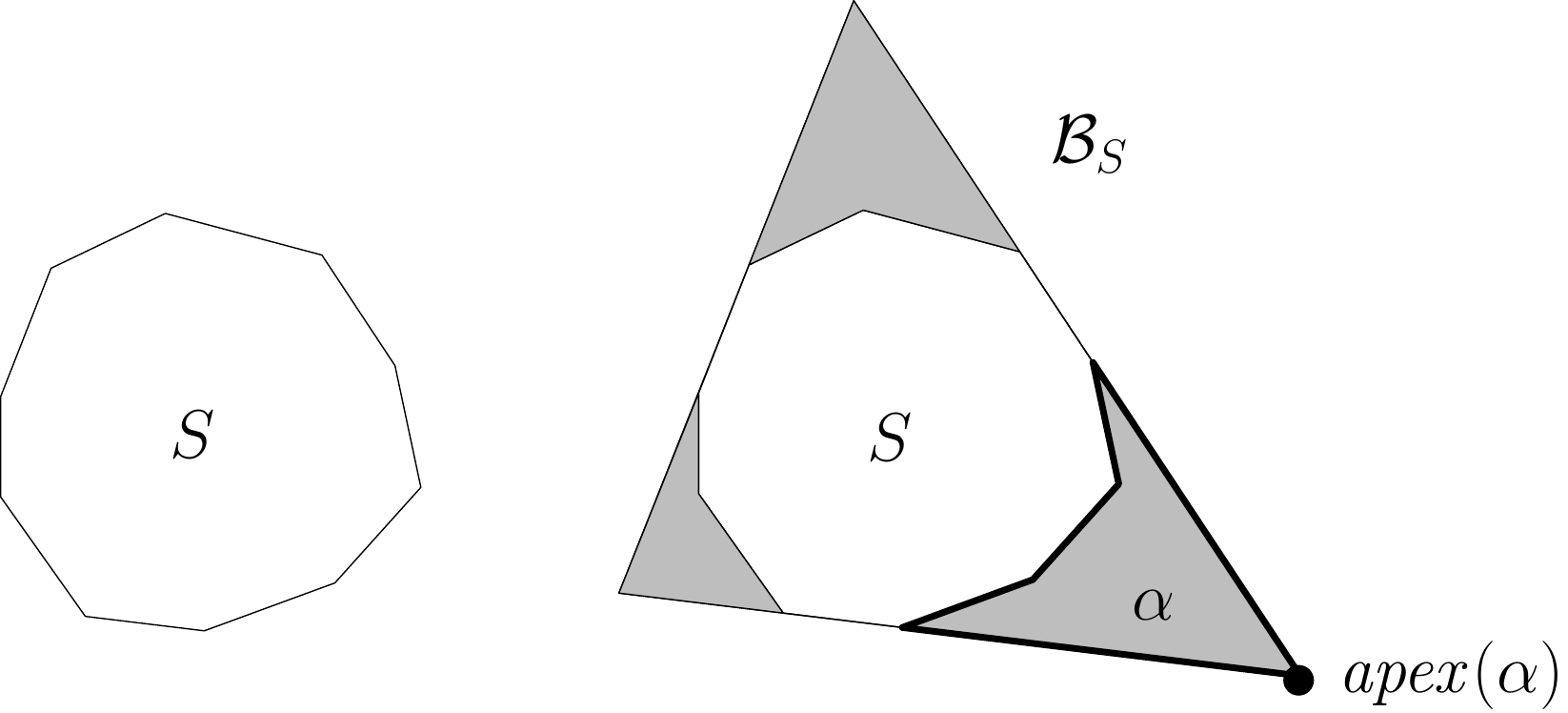} & &
		\includegraphics[scale=0.45]{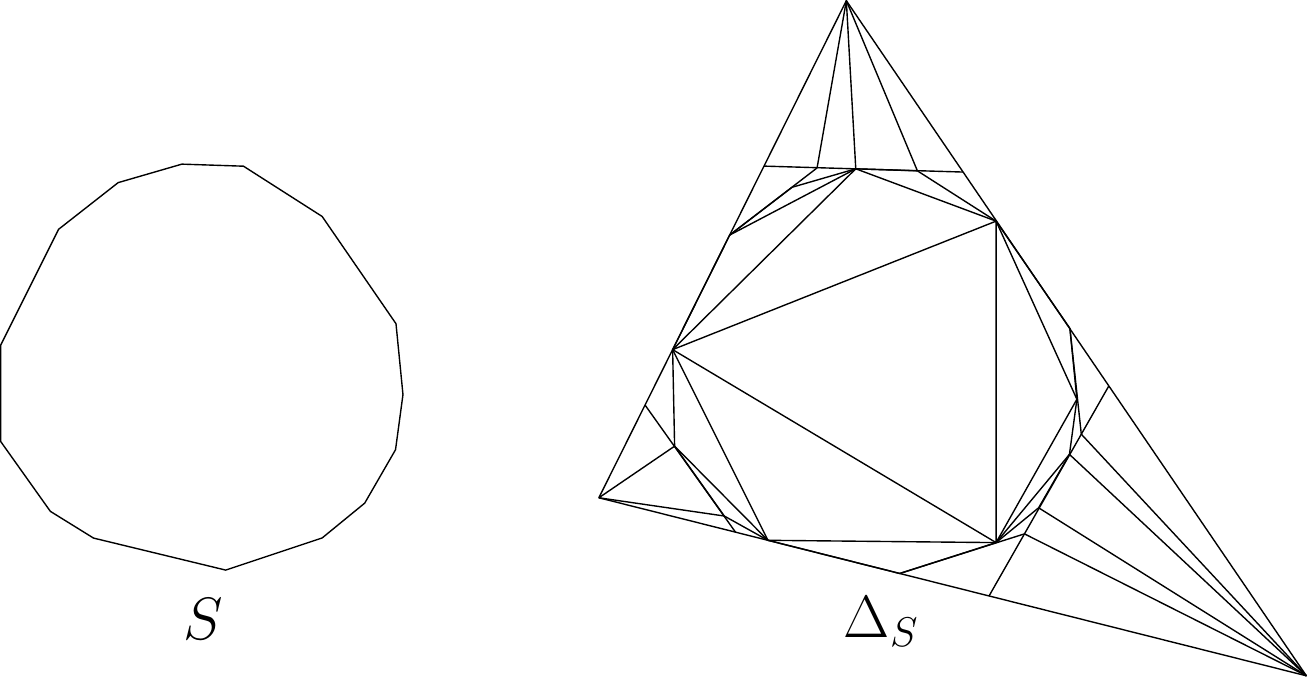} \\
		(a) & \hspace*{.2in} & (b)
	\end{tabular}
	\caption{(a)~The three fins are shown shaded.  For fin $\alpha$, $\apex(\alpha)$ is marked with a black dot, the head sides of $\alpha$ are the two segments incident to $\apex(\alpha)$, and $\tail(\alpha)$ is the reflex chain opposite $\apex(\alpha)$.  (b)~An example in which
		$S$ is just one convex polygon.}
	\label{fg:boomerang}
\end{figure}

\cancel{
\begin{figure}
	\centerline{\includegraphics[scale=.8]{figures/split}}
	\caption{The nodes of $T_\alpha$ are given the same colors as the corresponding regions in $\tilde{\alpha}$.}
	\label{fg:split}
\end{figure}
}

\begin{figure}
	\centering
	\begin{tabular}{ccccc}
		\includegraphics[scale=0.9]{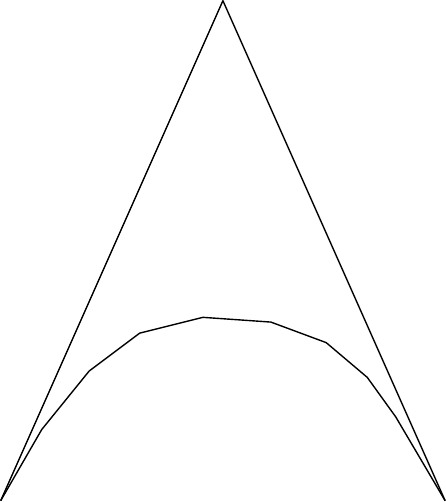} & &
		\includegraphics[scale=0.9]{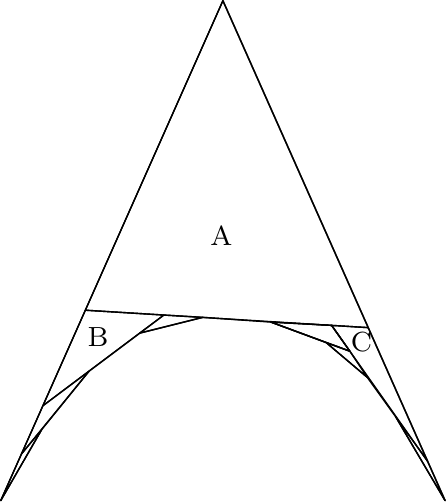} & &
		\includegraphics[scale=1.2]{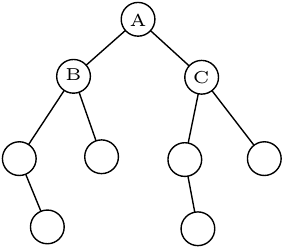} \\ \\
		(a) & \hspace*{.2in} & (b) & \hspace*{.2in} & (c)
	\end{tabular}
\caption{The fin $\alpha$ in (a) is partitioned into triangular regions in (b).  In (c), the rooted tree $T_\alpha$ of the triangular regions in (b) is shown.  The top three nodes in $T_\alpha$ are the triangular regions A, B and C in (b).  The other nodes of $T_\alpha$ correspond in a similar way to other triangular regions in (b).}
\label{fg:split}
\end{figure}

For every triangular region $r$ in $T_{\alpha}$, there is
exactly one side $e$ of $r$ that bounds $S$, and $e$ contains $O(\log
\mathit{size}(\alpha))$ vertices.  Finally, we call {\sf TriFin}$(\alpha)$ to obtain a
triangulation of $\alpha$.

\parbox[t]{\textwidth}{
	\begin{quote}
		\noindent {\sf TriFin}$(\alpha)$
		\begin{enumerate}
			\item For each triangular region $r$ in $T_{\alpha}$, do
			\begin{enumerate}
				\item take the side $e$ of $r$ that contains vertices in its interior,
				\item add edges to connect the vertices in $e$ to the vertex of $r$ opposite $e$.
			\end{enumerate}
		\end{enumerate}
\end{quote}}

The triangulations of the bounded regions and the three fins form the \emph{canonical triangulation} of $S$.

\begin{lemma}\label{lem:basic}
	For every convex subdivsion $S$ with $n$ vertices, a canonical triangulation of $S$ can be constructed in $O(n)$ time and it satsfies the following properties. 
	\begin{emromani}
		\item The canonical triangulation has $O(n)$ vertices, a triangular outer boundary, and no additional vertex in any bounded region of $S$.
		\item Any triangle that lies inside a bounded region of $S$ intersects $O(\log n)$ triangles in the canonical triangulation.
		\item Any triangle that lies inside a fin intersects $O(\log^2 n)$ triangles in the canonical triangulation.
	\end{emromani}
\end{lemma}
\begin{proof}
	The correctness of (i) is clear from the construction.  Let $t$ be any triangle.  
	
	Suppose that $t$ lies inside a bounded region $r \in S$.  As there is no additonal vertex in any bounded region of $S$, a triangle $\tau$ in the canonical triangulation intersects $t$ if and only if an edge of $t$ intersects $\tau$.  $\mathsf{TriReg}(r)$ generates a sequence of convex polygons successively.  The difference between any two adjacent polygons in the sequence is a set of triangles.  The set of all such triangles form the triangulation of $r$.  Every edge of $t$ intersects at most two triangles in the difference of two adjacent polygons in the sequence.  The $O(\log n)$ bound thus follows.  
	
	Suppose that $t$ lies inside a fin $\alpha$.  Observe that a triangular region in $T_\alpha$ intersects $t$ if and only if it intersects an edge of $t$.  Also, the triangular regions in $T_{\alpha}$ intersected by any edge  lie on an ancestor-to-descendant path in $T_\alpha$.  Thus, $t$ intersects $O(\log \mathit{size}(\alpha))$ triangular regions in $T_\alpha$.  Every triangular region is partitioned into $O(\log \mathit{size}(\alpha))$ triangles, which implies that $t$ intersects $O(\log^2 \mathit{size}(\alpha))$ triangles.
\end{proof}

Fins and their triangulations were first described by Hershberger and Suri~\cite{paper:Hershberger1993} in the context of solving the ray-shooting problem in a simple polygon.\footnote{Fins are called boomerangs in~\cite{paper:Hershberger1993}.}  They give a tighter but more involved analysis that yields an $O(\log |\alpha|)$ bound on the triangles intersected by a segment.  For our purposes, a polylogarithmic bound suffices.

\subsection{Algorithm for hierarchy construction}
\label{sec:first}

We take the canonical triangulation in Lemma~\ref{lem:basic} as $\Delta_{1,1}$.  We define the parameters $c_1$, $c_2$, $f$, and $g$ required in Table~\ref{tb:C}:
\[
c_1 = 6, \quad
c_2 = 2, \quad
f(k) = (\log_2 k)^6, \quad
g(n,i) = \log^2_2 n_i.
\]

It remains to describe how to construct new structures at the $(i+1)$-th level, including $R_{i+1,j+1}$, $S_{i+1,j+1}$, and $\Delta_{i+1,j+1}$.  As described in Section~\ref{sec:hierarchy}, $D_{i+1,j+1}$ consists of one point location structure obtained by applying Theorem~\ref{thm:iacono} to $\Delta_{i+1,j+1}$ and another worst-case optimal point location structure for $\Delta_{i+1,j+1}$.

\paragraph{Second layer.}  To start, we discuss the construction of the $j$-th version of structures at the second level.  We extract the $f(n_1)/(\log_2 n_1)^{c_2} = f(n_1)/(\log_2 n_1)^{2}$ triangles in $\Delta_{1,1}$ that have the highest $f(n_1)/(\log_2 n_1)^{2}$ access frequencies in $\Delta_{1,1}$.  Let $X$ denote the set of triangle extracted.  To extract $X$ quickly, we maintain a doubly linked list $A$ such that the $i$-th entry of $A$ stores a doubly linked list of triangles with the $i$-th highest frequency.  Whenever we need to output $X$, we scan the lists in the entries of $A$ in order until we have collected $f(n_1)/(\log_2 n_1)^{2}$ triangles.  Whenever the frequency of a triangle $t \in \Delta_{1,1}$ is incremented, we need to relocate $t$ within $A$.  Suppose that $t$ is currently stored in the list at the $k$-th entry of $A$.  If the triangles in the list at the $(k-1)$-th entry of $A$ have the same frequency as $t$, then we move $t$ to the end of the list at the $(k-1)$-th entry.  Otherwise, the triangles in the list at the $(k-1)$-th entry have a higher frequency than $t$, and so we insert a new entry of $A$ between the $(k-1)$-th and the $k$-th entries and make $t$ a singleton list at this new entry of $A$.  If the list at the $k$-th entry of $A$ becomes empty after moving $t$, we delete this entry of $A$.  So each update of $A$ takes $O(1)$ time.

For each triangle $t$ in $\Delta_{1,1}$, if $t$ lies in a region $r$ in $S_{1,1}$, we store the region id $r$ at $t$.  Moreover, if $r$ is the exterior region of $S$, then $t$ lies inside a triangular region $\tau$ in $T_\alpha$ for some fin $\alpha$, and we store the id of $\tau$ at $t$.  By sorting the triangles in $X$ with respect to the ids stored at them, we can find the triangles in $r \cap X$ for every region $r$ in $S_{1,1}$.

For each bounded region $r$ in $S$, let $X_r$ denote the subset of triangles $\{ t \in X : t \subseteq r \}$, let $\conv(X_r)$ denote the convex
hull of the triangles in $X_r$, and we can compute $\conv(X_r)$ in
$O(|X_r|\log |X_r|)$ time.  The size of $\conv(X_r)$ is $O(|X_r|)$.  For every fin $\alpha$ and every triangle $t \in X$ inside $\alpha$, we mark the triangular region $\tau$ in $T_\alpha$ that contains $t$, as well as all ancestors of $\tau$ in $T_\alpha$.  We form the union of the marked triangular regions, which is a fin shape.  Denote the union by $\shrink(\alpha)$.  Every edge in the reflex chain of $\shrink(\alpha)$ supports an outer boundary edge of $S$.  Figure~\ref{fg:split2} gives an example.  The size of $\mathit{shrink}(\alpha)$ is $|\{ t \in X : t \subseteq \alpha \}| \cdot O(\log \mathit{size}(\alpha))$, and $\mathit{shrink}(\alpha)$ can be computed in time linear in its size.

The set of convex hulls and shrunk fins obtained in the above form the new set of regions, $R_{2,j}$, at the second level.  The number of vertices in $R_{2,j}$ is $O(\log n_1 \cdot f(n_1)/(\log n_1)^2) = O(f(n_1)/\log n_1) = O(n_2)$.

The triangle $B_S$ encloses $R_{2,j}$.  We apply a plane sweep to fill the space among the regions in $R_{2,j}$ inside $B_S$ with additional triangles.  No additional vertex is introduced.  The resulting subdivision is $S_{2,j}$.

We construct $\Delta_{2,j}$ as follows.  For each convex hull $\conv(X_r)$ in $R_{2,j}$, we call {\sf TriReg}$(\conv(X_r))$ to triangulate $\conv(X_r)$.  Each resulting triangle stores the region id $r$.  For every fin $\shrink(\alpha)$ in $R_{2,j}$,
we call {\sf SplitBR}$(\shrink(\alpha))$ and then {\sf TriBR}$(\shrink(\alpha))$ to obtain a triangulation of $\shrink(\alpha)$.
At each resulting triangle $t$, we store the id of the exterior region of $S$ and the id of the triangular region in $T_{\shrink(\alpha)}$ that contains $t$.  The triangulation $\Delta_{2,j}$ has $O(n_2)$ vertices.

Constructing $R_{2,j}$ takes $O(n_2)$ time, constructing $S_{2,j}$ takes $O(n_2\log n_2)$ time, and constructing $\Delta_{2,j}$ takes $O(n_2)$ time.  Hence, $P_2 = O(n_2\log n_2)$.

\cancel{
\begin{figure}
	\centerline{\includegraphics[scale=0.7]{figures/split2}}
	\caption{The two triangular regions in $\tilde{\alpha}$ with white dots
		contain some triangles in $X$.  Corresponding nodes in $T_\alpha$ are also marked with white dots.  Then, all ancestors of these nodes in $T_\alpha$ are marked, and the union of the corresponding triangular regions in $\tilde{\alpha}$ is a fin $\shrink(\alpha)$ (shown shaded).}
	\label{fg:split2}
\end{figure}
}

\begin{figure}
	\centering
	\begin{tabular}{ccc}
		\includegraphics[scale=0.9]{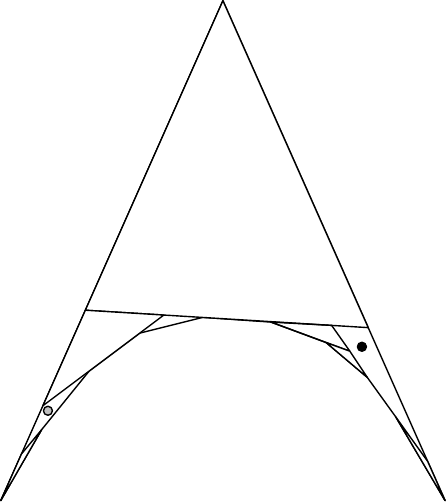} & & 
		\includegraphics[scale=0.9]{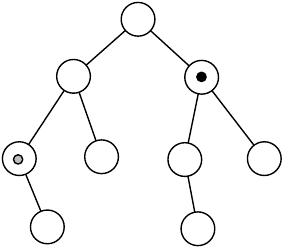} \\
		(a) & \hspace*{.4in} & (b) \\ \\ \\
		\includegraphics[scale=0.9]{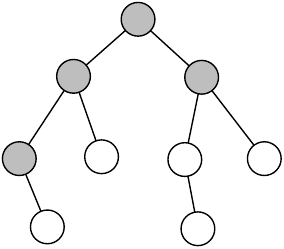} & & 
		\includegraphics[scale=0.9]{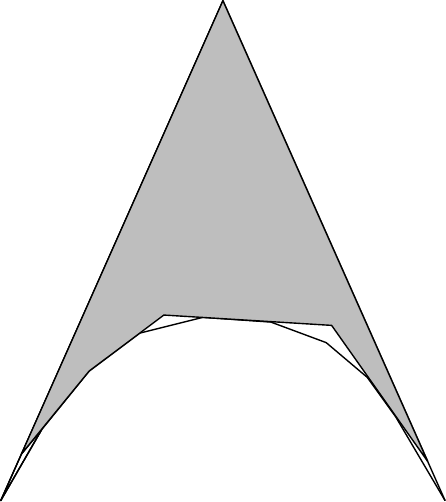} \\
		(c) & \hspace*{.4in} & (d)	
	\end{tabular}
\caption{In (a), we have a fin $\alpha$ and the two triangular regions labelled with grey and black dots contain some triangles in $X$.  The corresponding nodes in $T_\alpha$ in (b) are also labelled with grey and black dots.  These nodes and their ancestors are marked in (c).  The union of the corresponding triangular regions is $\mathit{shrink}(\alpha)$, which is shown shaded in (d).}
\label{fg:split2}
\end{figure}

\paragraph{Higher layers.}  To construct the $(j+1)$-th version of structures at the $(i+1)$-th layer for $i \geq 2$, we extract the $f(n_i)/(\log_2 n_i)^{2}$ triangles in $\Delta_{i,*}$ with the highest $f(n_i)/(\log_2 n_i)^{2}$ access frequencies.  Let $X$ denote the set of triangles extracted.  For each convex hull $C \in R_{i,*}$, let $X_C = \{t \in X : t \subseteq C\}$, and we compute $\conv(X_C)$.  For each fin $\alpha \in R_{i,*}$ and every triangle $t$ in $X$ inside $\alpha$, we mark the triangular region $\tau$ in $T_\alpha$ that contains $t$, as well as all ancestors of $\tau$ in $T_\alpha$.  We form the union, $\shrink(\alpha)$, of the marked triangular regions.  The collection of convex hulls and shrunk fins form the new set of regions $R_{i+1,j+1}$.  

The construction of $S_{i+1,j+1}$ and $\Delta_{i+1,j+1}$ is the same as described for the second layer.  At each triangle $t \in \Delta_{i+1,j+1}$ that lies inside a region in $R_{i+1,j+1}$, we store the id of the region in $S$ that contains $t$.  Moreover, if $t$ is contained in a shrunk fin $\mathit{shrink}(\alpha)$, we also store at $t$ the id of the triangular region in $T_{\mathit{shrink}(\alpha)}$ that contains $t$.

The number of vertices in $R_{i+1,j}$ is $O(f(n_i)/\log n_i) = O(n_{i+1})$. The total time need to construct these new structures is $P_{i+1} = O(n_{i+1}\log n_{i+1})$.

\paragraph{Meeting the conditions.}  Conditions C1--C3 are satisfied by the setting of $c_1$, $c_2$, and $f$ as well as the working of the construction as described above.  

Consider Condition C4.  Consider $R_{i,j}$ for some $i$ and $j$.  Let $t$ be any triangle.  If $t$ lies inside a bounded region $r$ of $S$, then $t$ can only intersect the convex hull $C$ in $R_{i,j}$ that lies inside $r$.  The intersection $t \cap C$ has $O(1)$ size, and so $t \cap C$ can be split into $O(1)$ triangles that lie completely inside or outside $C$.  Each such triangle intersects $O(\log n_i)$ triangles in $\{ t' \in \Delta_{i,j} : t' \subseteq C\}$.  Suppose that $t$ lies inside the exterior face of $S$.  Then, $t$ intersects at most one fin in $R_{i,j}$, say $\alpha$.  The triangle $t$ contains at most one vertex of $\tail(\alpha)$; otherwise, $t$ would cross the outer boundary of $S$, a contradiction to the assumption that $t$ lies inside the exterior region of $S$.  Therefore, we can split $t \cap \alpha$ into $O(1)$ triangles that lie completely inside or outside $\alpha$.  By the same reasoning as in the proof of Lemma~\ref{lem:basic}(iii), each such triangle intersects $O(\log^2 n_i)$ triangles in $\{ t' \in \Delta_{i,j} : t' \subseteq \alpha\}$.  This proves that condition C4 is satisfied by the setting $g(n,i) = \log^2_2 n_i$.

Condition C5 in Table~\ref{tb:C} holds by our setting of $c_2$, $f$ and $g$.  In fact, we have $f(n_i) \geq g(n,i)(\log_2 n_i)^{c_2+2}$ for all $n$ and $i$.

Let $m$ be the index of the highest layer in processing an online query sequence.  Consider the application of Theorem~\ref{thm:generic-3}.  Condition P1 in Table~\ref{tb:P} is satisfied as $f(n_i) \geq g(n,i)(\log_2 n_i)^{c_2+2}$ for all $n$ and $i$.  If P2 in Table~\ref{tb:P} holds, the total processing time is $O(\mathrm{OPT} + n)$.  If P2 does not hold, then $n_m < (c_1^2+c_1)^{c_1-1}$, which implies that $\log g(n,m) = O(1)$.  Hence, the total processing time is still $O(\mathrm{OPT} + n)$.

\begin{theorem}
\label{thm:convex}
Let $S$ be a planar convex subdivision with $n$ vertices. There is a point-line
comparison based structure that processes any online
sequence $\sigma$ of point location queries in $S$ in $O(\mathrm{OPT} + n)$
time, where {\rm OPT} is the minimum time required by any point location linear
decision tree for $S$ to process $\sigma$.  The time bound includes the
$O(n)$ preprocessing time.  The space usage is $O(n)$.
\end{theorem}

\section{Connected subdivision}
\label{sec:connected}

We describe an algorithm ${\cal A}(f,g)$ for connected subdivisions so that a processing time of $O(\mathrm{OPT} + n + |\sigma|\log(\log^*n))$ can be obtained by applying Theorem~\ref{thm:generic-3}.  We describe how to compute the first layer in $O(n)$ time in Section~\ref{sec:connectedfirstsol}.  It is a simulation of the procedure in Section~\ref{sec:first}.  The construction of the other layers is much more involved.  We need several new ideas in order to obtain a small $g(n,i)$.  

\subsection{First layer}
\label{sec:connectedfirstsol}

We compute the triangulation $\Delta_{1,1}$ as in~\cite{paper:Hershberger1993} for solving the ray-shooting problem in a simple polygon.  Other auxiliary structures are constructed along the way. The details are given below.

\paragraph{Balanced geodesic triangulation.}  Let $r$ be a simple polygon.  Let $k$ be the number of vertices of $r$.  Pick three vertices $v_1, v_{k/3}, v_{2k/3}$ of $r$ (which divide the boundary of $r$ into chains of roughly equal sizes).  The three geodesic paths inside $r$ among $v_1, v_{k/3}, v_{2k/3}$ bound a \emph{kite}.  Refer to Figure~\ref{fg:geo}(a) for an example.  

The part of the kite with a non-empty interior is a \emph{geodesic triangle} $\tau$, whose boundary consists of three reflex chains.  
The kite splits the boundary of $r$ into chains delimited by $v_1$, $v_{k/3}$ and $v_{2k/3}$. Next, compute the geodesic paths from $v_1$ and $v_{k/3}$ to $v_{k/6}$, the middle vertex in the boundary chain between $v_1$ and $v_{k/3}$.  This creates another kite joining $v_1$, $v_{k/6}$ and $v_{k/3}$ and hence another geodesic triangle $\tau'$ inside this kite.  We call $\tau$ the \emph{parent} of $\tau'$ and $\tau'$ a \emph{child} of $\tau$.  The same process is repeated to other chains recursively.  

In the end, we obtain a \emph{balanced geodesic triangulation} that partitions $r$ into geodesic triangles.  The parent-child relations among these geodesic triangles are explicitly represented in a rooted tree that stores them.  We denote this rooted tree by $G_r$.  The construction of the balanced geodesic triangulation and $G_r$ takes $O(\mathit{size}(r))$ time~\cite{chazelle94}, where $\mathit{size}(r)$ denotes the complexity of $r$.

\begin{figure}
	\centering
	\begin{tabular}{ccccc}
		\includegraphics[scale=0.4]{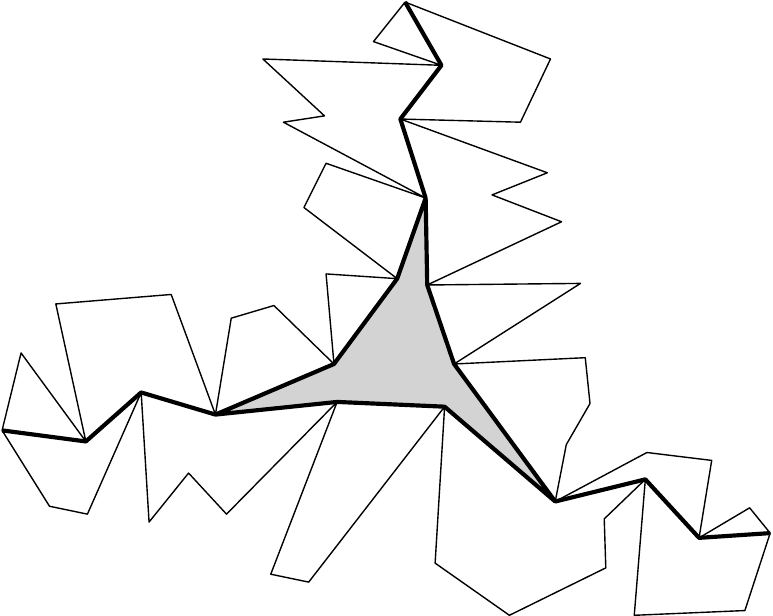} & \hspace*{.2in} &
		\includegraphics[scale=0.4]{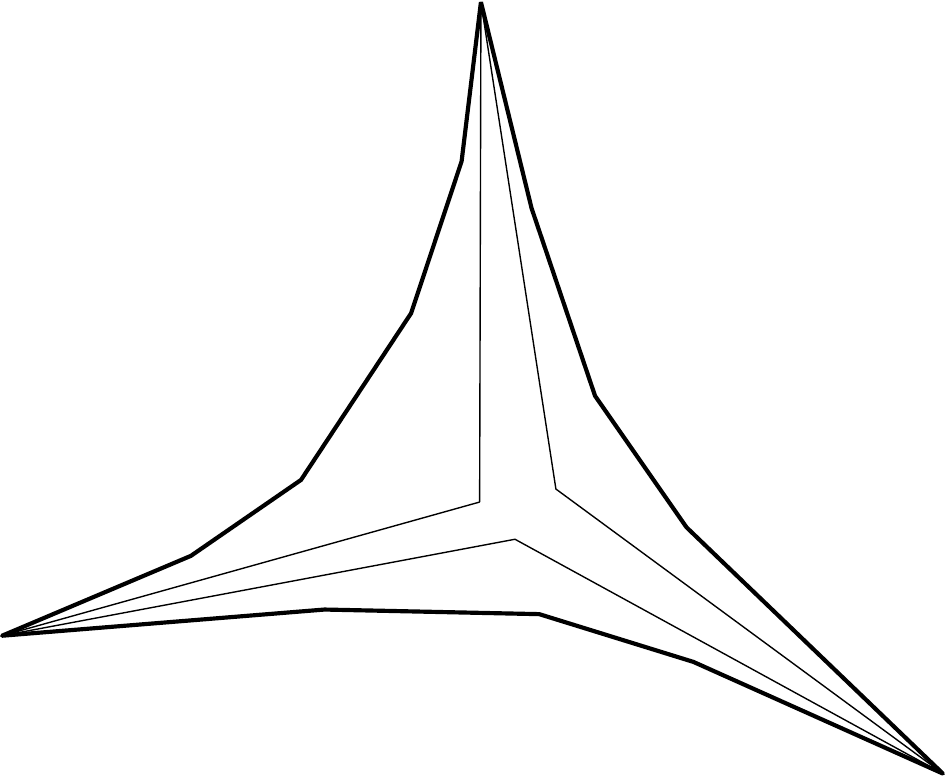} & \hspace*{.2in} &
		\includegraphics[scale=0.4]{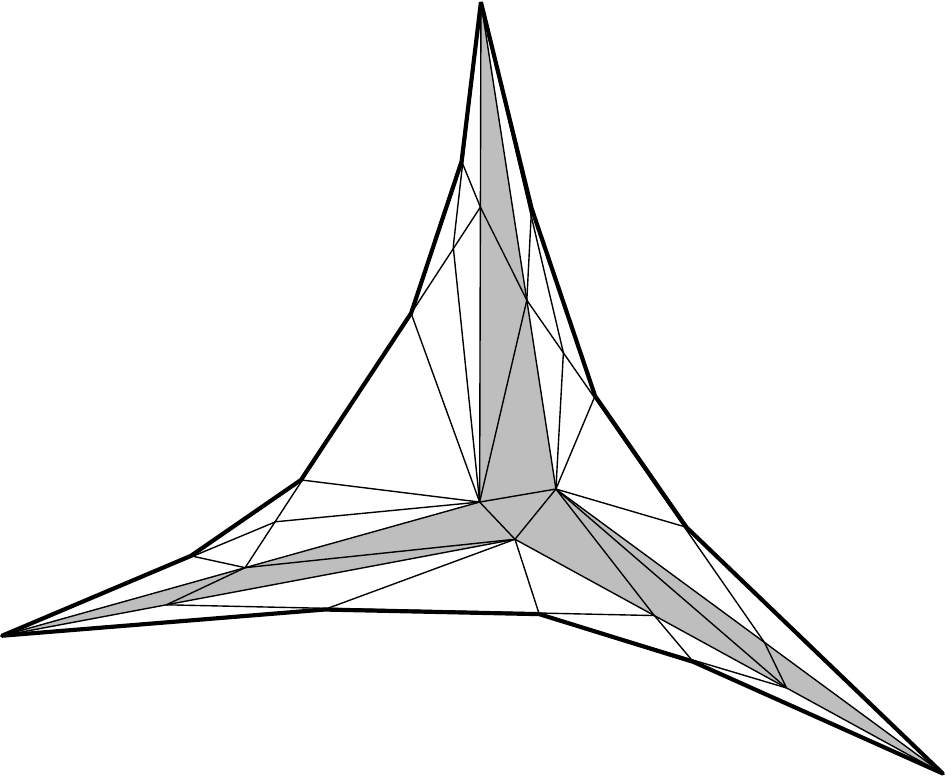} \\
		(a) & & (b) & & (c)
	\end{tabular}
	\caption{(a) Kite and the geodesic triangle inside (shown shaded).  (b) Divide a geodesic triangle into one star and three fins.  (c) Triangulation of the star and fins.}
	\label{fg:geo}
\end{figure}

\paragraph{Refinement.}  Let $\tau$ be a geodesic triangle in $G_r$ for some simple polygon $r$.  We shoot two rays inward from each vertex of $\tau$.  These six rays meet inside $\tau$ to form a star with three spikes.  Denote this star by $\st(\tau)$.    The difference $\tau \setminus \st(\tau)$ is a set of three fins, denoted by $\fins(\tau)$.  Figure~\ref{fg:geo}(b) shows an example.  Next, split $\st(\tau)$ into four triangles by connecting its three reflex vertices.  The fins are triangulated as described in Section~\ref{sec:first}.  For each fin $\alpha \in \mathit{fins}(\tau)$, we also build a binary tree $T_\alpha$ corresponding to the hierarchy of triangular regions in the decomposition of a fin as in Section~\ref{sec:first}.  Let $\mathit{size}(\tau)$ denote the complexity of $\tau$.  The triangulation of the fins places $O(\log \mathit{size}(\tau))$ vertices on the boundary of $\st(\tau)$.  We add segments arbitrarily to connect these vertices to triangulate $\st(\tau)$.  Figure~\ref{fg:geo}(c) shows an example.  This gives the triangulation of $\tau$ in $O(\mathit{size}(\tau))$ time.  The desired triangulation of $r$ is obtained by triangulating every geodesic triangle in $G_r$ as described above.  The refinement of the geodesic triangles in $r$ into stars and fins takes $O(\mathit{size}(r))$ time.  So does the subsequent triangulation of these stars and fins.

\paragraph{Structures at the first layer.} We compute the convex hull $\conv(S)$ of the outer boundary of $S$ in $O(|S|)$ time~\cite{melkman87}.   We compute a triangle $B_S$ to enclose $\conv(S)$ as described in Section~\ref{sec:first}.  Define:
\begin{quote}
\begin{tabular}{rp{5in}}
	$S_{1,1}$~~= & The overlay of $B_S$, $\conv(S)$ and $S$. \\
	$R_{1,1}$~~= & The set of bounded regions in $S_{1,1}$. \\
	$\tilde{\Delta}_{1,1}$~~= & The balanced geodesic triangulations of regions in $R_{1,1}$. \\
	$\hat{\Delta}_{1,1}$~~= & 
	$\displaystyle \bigl\{ \st(\tau) : \tau \in G_r, \text{region $r \in R_{1,1}$} \bigr\} \cup \bigcup_{r \in R_{1,1}, \tau \in G_r} \fins(\tau)$. \\
	$\Delta_{1,1}$~~= & The triangulation of all stars and fins in $\hat{\Delta}_{1,1}$ as described earlier. \\
	$D_{1,1}$~~= & A point location structure for $\Delta_{1,1}$.
\end{tabular}
\end{quote}

The structures $S_{1,1}$, $R_{1,1}$, $\Delta_{1,1}$, and $D_{1,1}$ are already introduced in Section~\ref{sec:alg}.  The additional structures $\tilde{\Delta}_{1,1}$ and $\hat{\Delta}_{1,1}$ are needed for handling connected subdivisions.

There are three fins between $B_S$ and $\conv(S)$, and they appear as ``degenerate'' geodesic triangles in $\tilde{\Delta}_{1,1}$.  In forming $\Delta_{1,1}$, it is unnecessary to decompose them further into fins and stars.  It suffices to split the three fins between $B_S$ and $\conv(S)$ as described earlier into triangular regions and subsequently into triangles. Later on, for the sake of simplicity,  we will not distinguish degenerate geodesic triangles when we discuss the construction of structures at higher layers.

\begin{lemma}\label{lem:1}
	Let $S$ be a connected subdivision of $n$ vertices.  The construction of $S_{1,1}$, $R_{1,1}$, $\tilde{\Delta}_{1,1}$, $\hat{\Delta}_{1,1}$,  $\Delta_{1,1}$, and $D_{1,1}$ takes $O(n)$ time.  
\end{lemma}
\cancel{
\begin{proof}
	The $O(n)$ running time bound follows from previous discussion.  As before $D_{1,1}$ consists of two structures, one obtained by applying Theorem~\ref{thm:iacono} and an worst-case optimal planar point location structure.
	
	Let $t$ be any triangle that lies inside a region (bounded or exterior) of $S$.  If $t$ lies outside $\conv(S)$, then $t$ intersects only one geodesic triangle in $\tilde{\Delta}_{1,1}$ outside $\conv(S)$ and $O(\log^2 n)$ triangles in $\Delta_{1,1}$ as we argued in the proof of Lemma~\ref{lem:basic}(iii). If $t$ crosses an edge of $\conv(S)$, $t$ can be split into $O(1)$ triangles such that each lies outside $\conv(S)$ or within a region inside $\conv(S)$.  So it suffices to analyze the latter case.  Suppose that $t$ lies in a region $r$ inside $\conv(S)$.  Each edge of $t$ intersects $O(\log n)$ kites in the balanced geodesic triangulation of $r$ and hence $O(\log n)$ geodesic triangles in $\tilde{\Delta}_{1,1}$.  Let $\tau$ be a geodesic triangle intersected by $t$.  In the triangulation of $\tau$, there are only $O(\log |\tau|)$ triangles in $\st(\tau)$.  Consider a fin $\alpha \in \fins(\tau)$.  The triangle $t$ cannot contain any vertex of $\tail(\alpha)$.  Then, we can argue as in the analysis of Lemma~\ref{lem:basic}(iii) that $t$ intersects $O(\log^2 |\alpha|)$ triangles of $\Delta_{1,1}$ inside $\alpha$.  As a result, $t$ intersects $O(\log^2 n)$ triangles in $\tau$ and hence $O(\log^3 n)$ triangles in $\Delta_{1,1}$.
\end{proof}
}

\subsection{Layer construction}

We have the following settings:
\[
c_1 = 31, \quad c_2 = 25, \quad f(k) = (\log_2 k)^{31}, \quad g(n,i) = i\log^3 n_i.
\]

Inductively, for any level $i \geq 2$ and any version index $j \geq 1$, regions in $R_{i,j}$ is formed by merging some fins and stars.  Each star taken up by $R_{i,j}$ is a star in $\hat{\Delta}_{i-1,*}$.  Each fin taken up by $R_{i,j}$ is a subset of a fin in $\hat{\Delta}_{i-1,*}$.  However, $R_{i,j}$ may not use all stars and fins in $\hat{\Delta}_{i-1,*}$.
\begin{quote}
	\begin{tabular}{rp{4.5in}}
		$R_{i,j}$~~= & A set of bounded regions, each being a subset of some bounded region in $S_{i-1,*}$.  The boundaries of regions in $R_{i,j}$ are compatible, i.e., a vertex of a region does not appear in interior of an edge of an adjacent region. \\
		$S_{i,j}$~~= & The overlay of $B_S$, $R_{i,j}$, and dummy triangles introduced to fill the space between the boundary of $B_S$ and the regions in $R_{i,j}$.  No additional vertex is introduced for adding the dummy triangles. \\
		$\tilde{\Delta}_{i,j}$~~= & The balanced geodesic triangulations of regions in $R_{i,j}$. \\
		$\hat{\Delta}_{i,j}$~~= & 
		$\displaystyle \bigl\{ \st(\tau) : \tau \in G_r, \text{region $r \in R_{i,j}$} \bigr\} \cup \bigcup_{r \in R_{i,j}, \tau \in G_r} \fins(\tau)$. \\
		$\Delta_{i,j}$~~= & The triangulation of $S_{i,j}$ that consists of the triangulations of stars and fins in $\hat{\Delta}_{i,j}$ (as illustrated in Figure~\ref{fg:geo}(c)) and the dummy triangles in $S_{i,j} \setminus R_{i,j}$. \\
		$D_{i,j}$~~= & A point location structure for $\Delta_{i,j}$.
	\end{tabular}
\end{quote}



\begin{table}
\centering
\renewcommand{\arraystretch}{1.3}
\begin{tabular}{|c|p{5in}|}
		\hline
		Q1 & For all $k$ and $l$, every edge in the balanced geodesic triangulations of the regions in $R_{k,l}$ is incident to a vertex of $S_{1,1}$. \\
		\hline
		Q2 & For all $k$ and $l$, no two regions in $R_{k,l}$ that lie in the same region in $S_{1,1}$ share an edge although they may share a vertex. \\
		\hline
		Q3 & For all $k$ and $l$, every region in $R_{k,l}$ is a simple polygon. \\
		\hline
\end{tabular}
\caption{Invariants satisfied by $R_{k,l}$.}
\label{tb:Q}
\end{table}

We require invariants Q1, Q2 and Q3 in Table~\ref{tb:Q} to be satisfied inductively.  These invariants will help us to meet conditions C1--C5 in Table~\ref{tb:C}, which will be discussed later in Section~\ref{sec:conditions}.  In the base step, Q1--Q3 are satisfied by $R_{1,1}$ trivially.  

$R_{i,j}$ will be constructed from the most recent structures at the $(i-1)$-th layer.  In order to preserve Q1 inductively, we need to adjust the construction of the balanced geodesic triangulation of a simple polygon slightly.  Recall that we pick three vertices that divide the polygon boundary into three chains of roughly equal lengths.  When we are constructing structures at some $i$-th level for some $i \geq 2$, it is possible that one or more of the division vertices picked from a region $r$ in $R_{i-1,*}$ is not a vertex of $S_{1,1}$.  If we pick a vertex $v$ that is not a vertex of $S_{1,1}$, Q1 ensures that the two vertices of $r$ adjacent to $v$ belong to $S_{1,1}$.  In this case, we pick a neighboring vertex of $v$ to be the division vertex instead.  The same shifting strategy applies when we recursively build the rest of $G_r$.  

\begin{lemma}\label{lem:shift}
	Let $r$ be a simple polygon.  Assume that every edge of $r$ is incident to a vertex of $S_{1,1}$.  In constructing a balanced geodesic triangulation of $r$, the shifting strategy ensures that the boundary of every kite produced consists of shortest paths among three vertices of $S_{1,1}$.
\end{lemma}

Lemma~\ref{lem:shift} will be helpful in preserving Q1.  We will need some ``seed fins'' to start the construction of $R_{i,j}$.  We explain what they are in the next section.


\subsubsection{Seed fins}
\label{sec:seed-fins}

Suppose that we are to construct the $j$-th version of structures at the $i$-th layer for some $i \geq 2$.  We discuss how to identify some \emph{seed fins}.  They will be used to form $R_{i,j}$.

Mark the $f(n_{i-1})/(\log_2 n_{i-1})^{c_2} = (\log_2 n_{i-1})^{c_1-c_2}$ most frequently accessed triangles in $\Delta_{i-1,*}$.   Select the geodesic triangles in $\tilde{\Delta}_{i-1,*}$ that contain these marked triangles.  For each such geodesic triangle $\tau$, $\tau$ belongs to $G_r$ for some region $r$ in $R_{i-1,*}$, and we also select all ancestors of $\tau$ in $G_r$.  Let $\tilde{X}_0$ be the resulting set of selected geodesic triangles in $\tilde{\Delta}_{i-1,*}$.

For every $\tau \in \tilde{X}_0$ and every $\alpha \in \fins(\tau)$, we select the triangular regions in $T_\alpha$ that contain a marked triangle or are incident to a convex vertex of $\tau$. Let ${X}_1$ be the set of selected triangular regions.  Figure~\ref{fg:example}(a) shows an example.

We select some extra triangular regions to deal with complications caused by non-convexity.  
For every pair of distinct geodesic triangles $\tau$ and $\tau'$ in $\tilde{X}_0$, if they share an edge $e$, then $e$ is shared by the tails of two fins $\alpha$ and $\alpha'$, where $\alpha \in \fins(\tau)$ and $\alpha' \in \fins(\tau')$, and we select the two triangular regions in $T_\alpha$ and $T_{\alpha'}$ that are incident on $e$.  Figure~\ref{fg:example}(b) shows an example.

\begin{figure}
	\centering
	\begin{tabular}{cc}
		\includegraphics[scale=0.7]{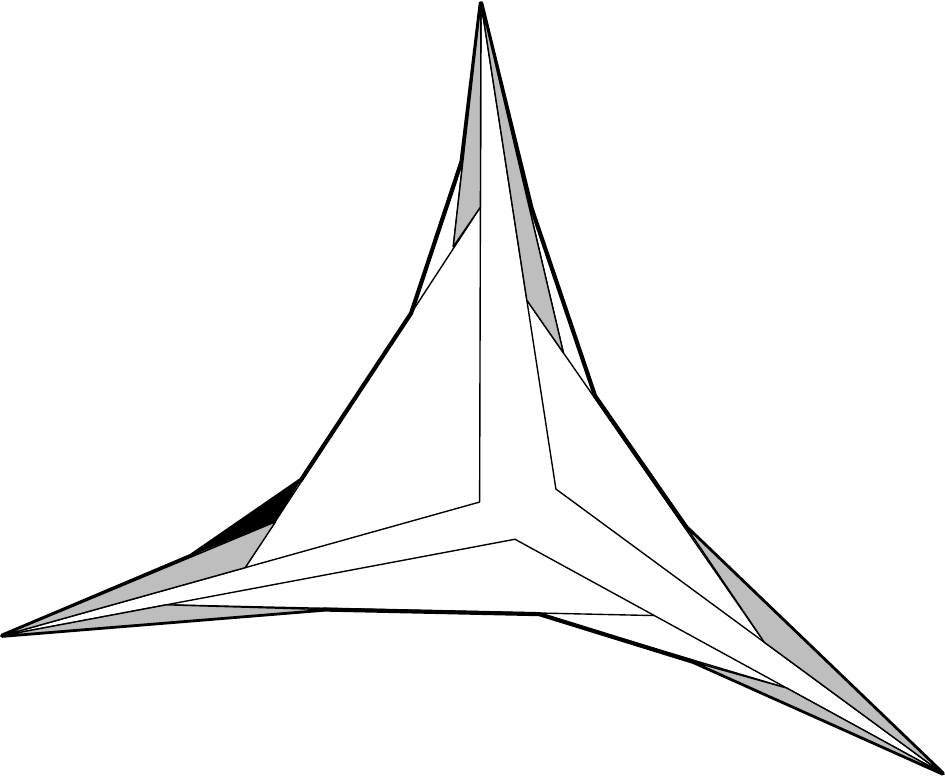} &
		\includegraphics[scale=0.7]{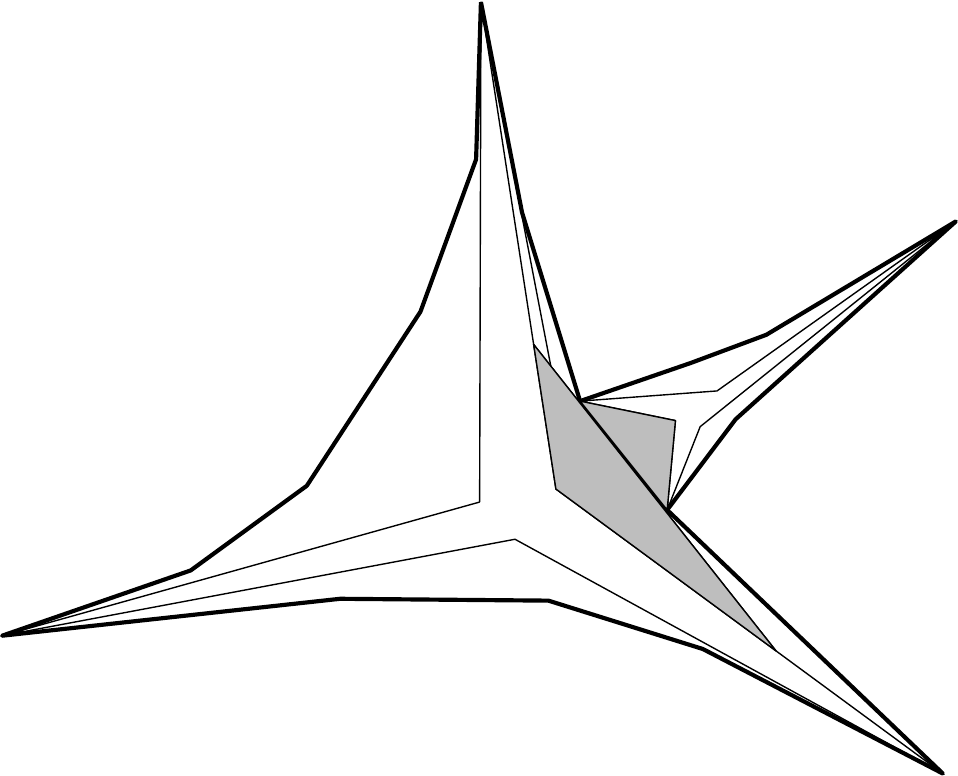} \\
		(a) &  (b) \\ \\
	\end{tabular}
	\begin{tabular}{c}
		\includegraphics[scale=0.7]{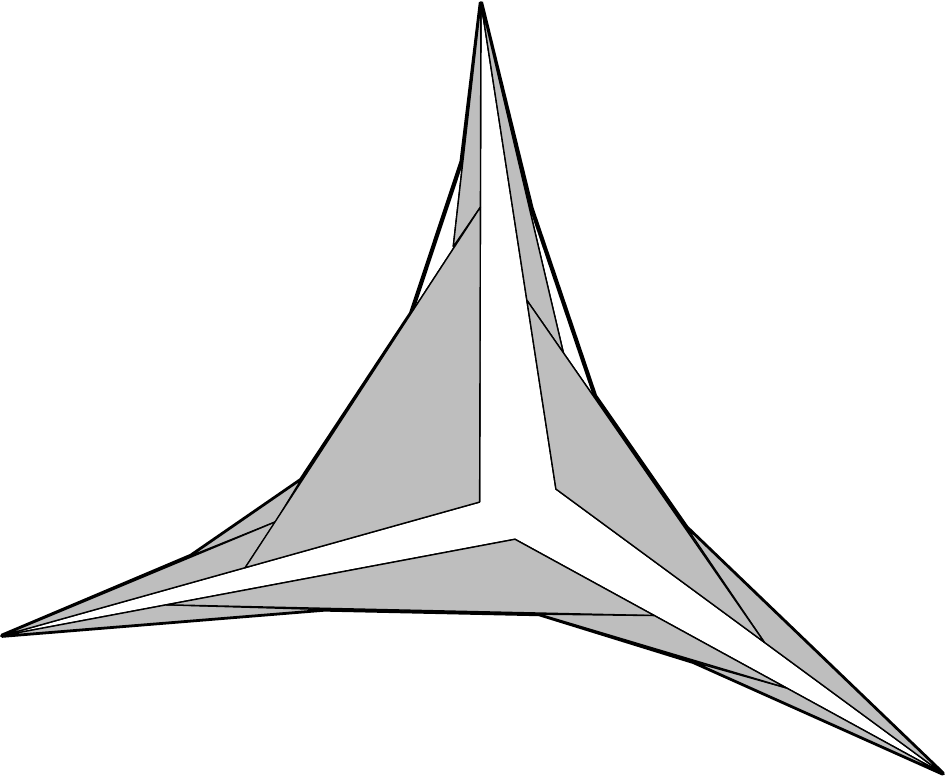} \\
		(c)
	\end{tabular}
	\caption{In (a), a geodesic triangle $\tau$ is partitioned into a star and three fins, and the fins are partitioned into triangular regions.  The black triangle is marked as it has been frequently accessed by queries.  The gray triangles are those incident to the convex vertices of $\tau$.  Both the black and gray triangles are selected and included into $X_1$.  In (b), two selected geodesic triangles are shown, and the two gray triangular regions incident on their common edge are selected and included in $X_2$.  In (c), the grey parts are the seed fins formed by the union of the selected triangular regions in (a) and their ancestors.}
	\label{fg:example}
\end{figure}

\begin{figure}
	\centering
	\begin{tabular}{cc}
		\includegraphics[scale=0.45]{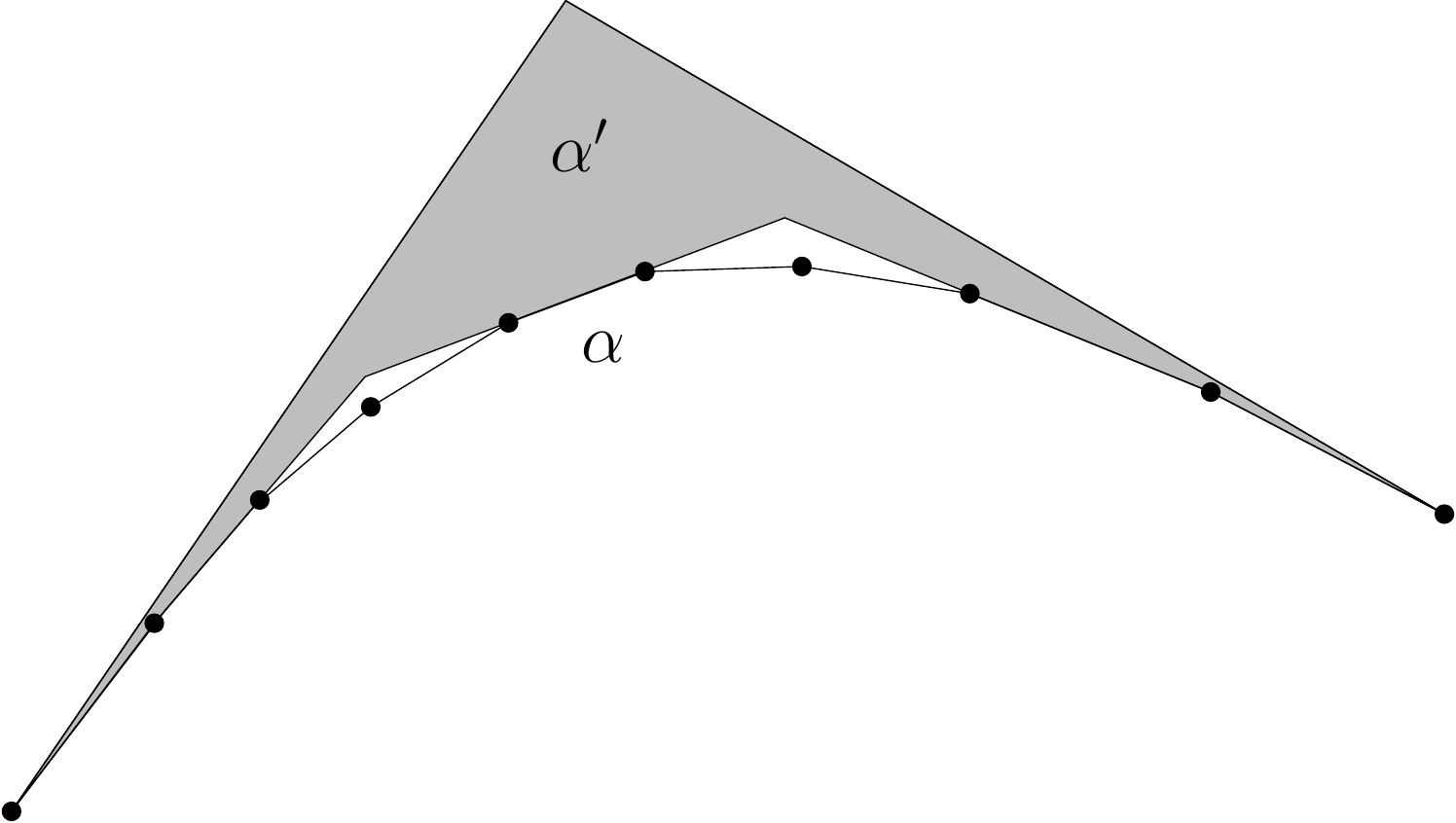} &
		\includegraphics[scale=0.45]{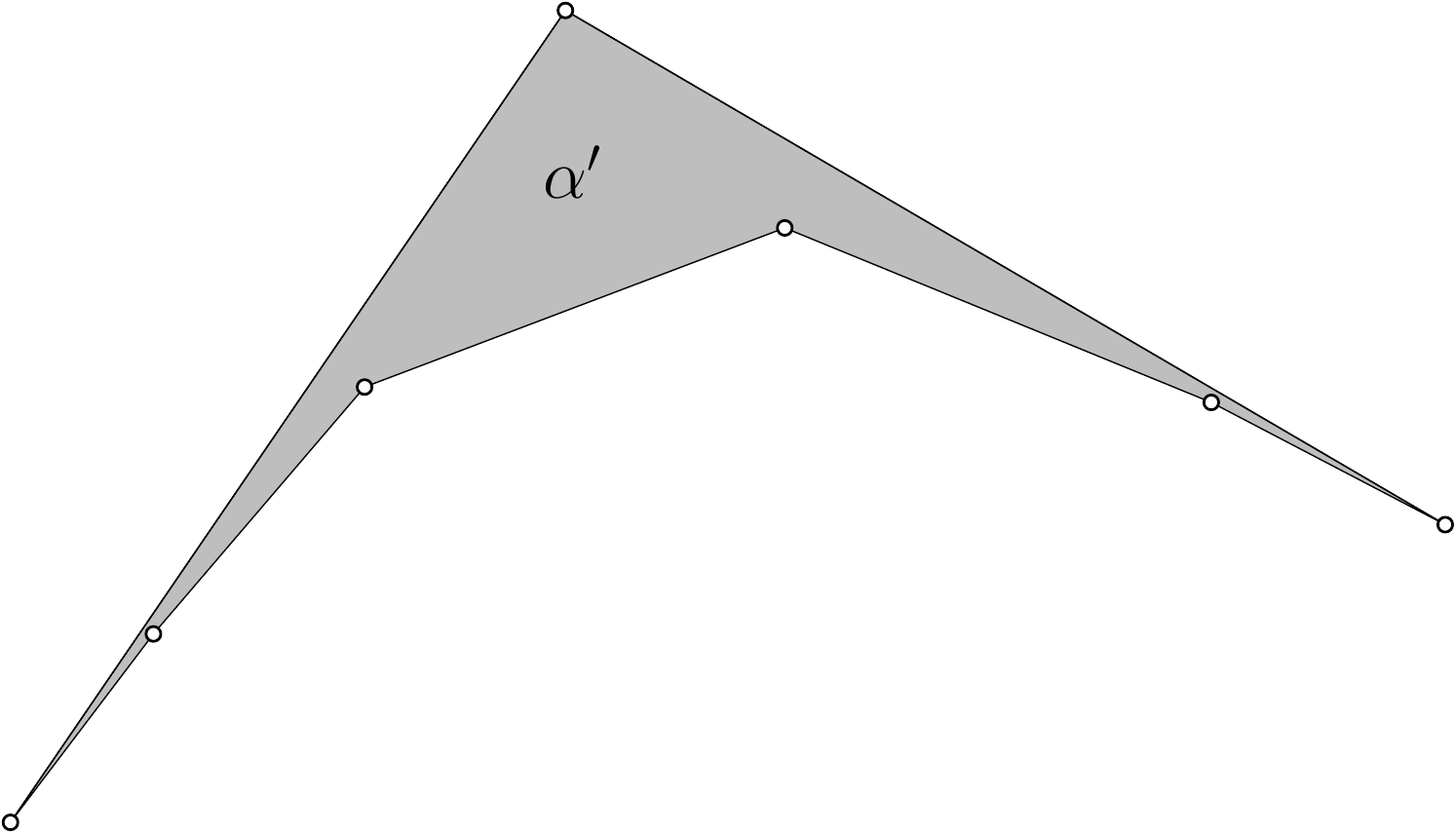} \\
		(a) & (b) \\ \\ \\
	\end{tabular}
	\caption{In (a), a seed fin $\alpha'$ is shown in gray inside a fin $\alpha \in \tilde{\Delta}_{i-1,*}$.  In (b), the vertices of $\alpha'$ are shown as white dots.  Some vertices of $\tail(\alpha)$ lie in the interior of edges in $\tail(\alpha')$.}
	\label{fg:seed-fin}
\end{figure}

Let ${X}_2$ be the union of ${X}_1$ and the set of extra triangular regions selected according to the description in the previous paragraph.  Every triangular region $t \in {X}_2$ corresponds to a node in $T_\alpha$ for some fin $\alpha \in \hat{\Delta}_{i-1,*}$.  We collect the ancestors of $t$ in $T_\alpha$.  Let ${X}_3$ be the set of triangular regions that consist of those in ${X}_2$ as well as their ancestors.  We are ready to define the set of seed fins.

\begin{quote}
	For every fin $\alpha \in \hat{\Delta}_{i-1,*}$, we merge all triangular regions in $X_3$ that lie inside $\alpha$.  The result is empty or a \emph{seed fin}.  There is at most one seed fin in $\alpha$.  Let $\mathit{seed\_fins}[i,j]$ denote the set of all seed fins obtained.    The parameters $i$ and $j$ signify that this is a structure of the $j$-th version at level $i$.  Figure~\ref{fg:example}(c) shows an example of the merging to form a seed fin.  Figure~\ref{fg:seed-fin} shows a seed fin as subset of a fin in $\hat{\Delta}_{i-1,*}$.  
 
\end{quote}

How long does it take to construct $\mathit{seed\_fins}[i,j]$?  

There are $O(\log n_{i-1} \cdot f(n_{i-1})/(\log_2 n_{i-1})^{c_2}) = O((\log n_{i-1})^{c_1-c_2+1})$ geodesic triangles in $\tilde{X}_0$.  The same bound applies to the number of triangular regions in ${X}_1$.  There are $O(|\tilde{X}_0|^2) = O((\log n_{i-1})^{2c_1-2c_2+2})$ pairs of geodesic triangles to check in identifying the extra triangular regions to be included in ${X}_2$.  Given a pair of geodesic triangles $\tau$ and $\tau'$ in $\tilde{X}_0$, 
we check if they share an edge in $O(\log n_{i-1})$ time as follows.  
There are nine pairs of reflex chains, one from $\tau$ and another from $\tau'$.  We check every pair of reflex chains, say $\{\gamma_1,\gamma_2\}$, to determine whether $\tau$ and $\tau'$ share a common edge.  If $\gamma_1$ or $\gamma_2$ is a single edge, say $\gamma_1$, then we can decide in $O(\log n_{i-1})$ time whether $\gamma_2$ contains $\gamma_1$ as an edge, assuming that the edges in $\gamma_2$ are stored in a dictionary structure.  If neither $\gamma_1$ nor $\gamma_2$ is a single edge, then any common edge must be the first or last edges of $\gamma_1$ and $\gamma_2$.  So this can be decided in $O(1)$ time.\footnote{In the degenerate case, two adjacent geodesic triangles $\tau_1$ and $\tau_2$ may share a collinear sequence $\xi$ of edges.  In this case, for $s \in [1,2]$, there is a single triangular region incident to the edges in $\xi$ in the refinement of fins in $\fins(\tau_s)$.  If $\xi$ is a whole reflex chain of $\tau_1$, it suffices to take any edge $e$ in $\xi$ and search for $e$ in a reflex chain of $\tau_2$ in logarithm time.  If $\xi$ is not a whole reflex chain of $\tau_1$ or $\tau_2$, then $\xi$ must be incident to convex vertices of $\tau_1$ and $\tau_2$ for $\xi$ to be shared by $\tau_1$ and $\tau_2$.  Therefore, it suffices to compare the edges incident to convex vertices of $\tau_1$ and $\tau_2$ in $O(1)$ time.}


In all, there are $O((\log n_{i-1})^{2c_1-2c_2+2})$ triangular regions in ${X}_2$, and $X_2$ can be constructed in $O((\log n_{i-1})^{2c_1-2c_2+3})$ time.  It follows that there are $O((\log n_{i-1})^{2c_1-2c_2+2})$ seed fins.

For every triangular region $t \in {X}_2$, it takes $O(\log n_{i-1})$ time to identify the ancestors of $t$ in $T_\alpha$, where $\alpha$ is the fin in $\hat{\Delta}_{i-1,*}$ that contains $t$.  The total complexity of $\mathit{seed\_fins}[i,j]$ is $O(|{X}_2|\log n_{i-1}) = O((\log n_{i-1})^{2c_1-2c_2+3})$, and the construction time of $\mathit{seed\_fins}[i,j]$ is $O((\log n_{i-1})^{2c_1-2c_2+3})$.

\newcounter{pagecount2}
\setcounter{pagecount2}{\thepage}

\begin{lemma}
	\label{lem:seed-fins}
	The cardinality of $\mathit{seed\_fins}[i,j]$ is $O((\log n_{i-1})^{2c_1-2c_2+2})$, its total complexity is $O((\log n_{i-1})^{2c_1-2c_2+3})$, and its construction time is $O((\log n_{i-1})^{2c_1-2c_2+3})$.  For every $\alpha \in \mathit{seed\_fins}[i,j]$, the endpoints of $\tail(\alpha)$ are convex vertices of the geodesic triangle in $\tilde{\Delta}_{i-1,*}$ that contains $\alpha$.
	\cancel{
	\begin{emromani}
		\item the endpoints of $\tail(\alpha)$ are convex vertices of the geodesic triangle in $\tilde{\Delta}_{i-1,*}$ that contains $\alpha$, and
		\item every edge in $\tail(\alpha)$ is incident to a vertex of $S_{1,1}$.
	\end{emromani}
}
\end{lemma}

\subsubsection{Augment the seed fins}
\label{sec:augment-seed-fins}

We cannot use the seed fins only to define $R_{i,j}$ while keeping a small $g(n,i)$.  A segment that lies inside a region of $S_{1,1}$ may generate many intersection points with the boundaries of the seed fins.  To remedy this, we augment (grow) the seed fins appropriately so that the number of such intersections decreases to an acceptable level.  

Every seed fin $\alpha$ is a subset of some fin $\alpha_0 \in \hat{\Delta}_{i-1,*}$.  The augmentation of $\alpha$ is done by merging $\alpha$ with suitable subsets of $\alpha_0 \setminus \alpha$.  We first introduce a notation for $\alpha_0 \setminus \alpha$.

\begin{figure}
	\centerline{\includegraphics[scale=0.5]{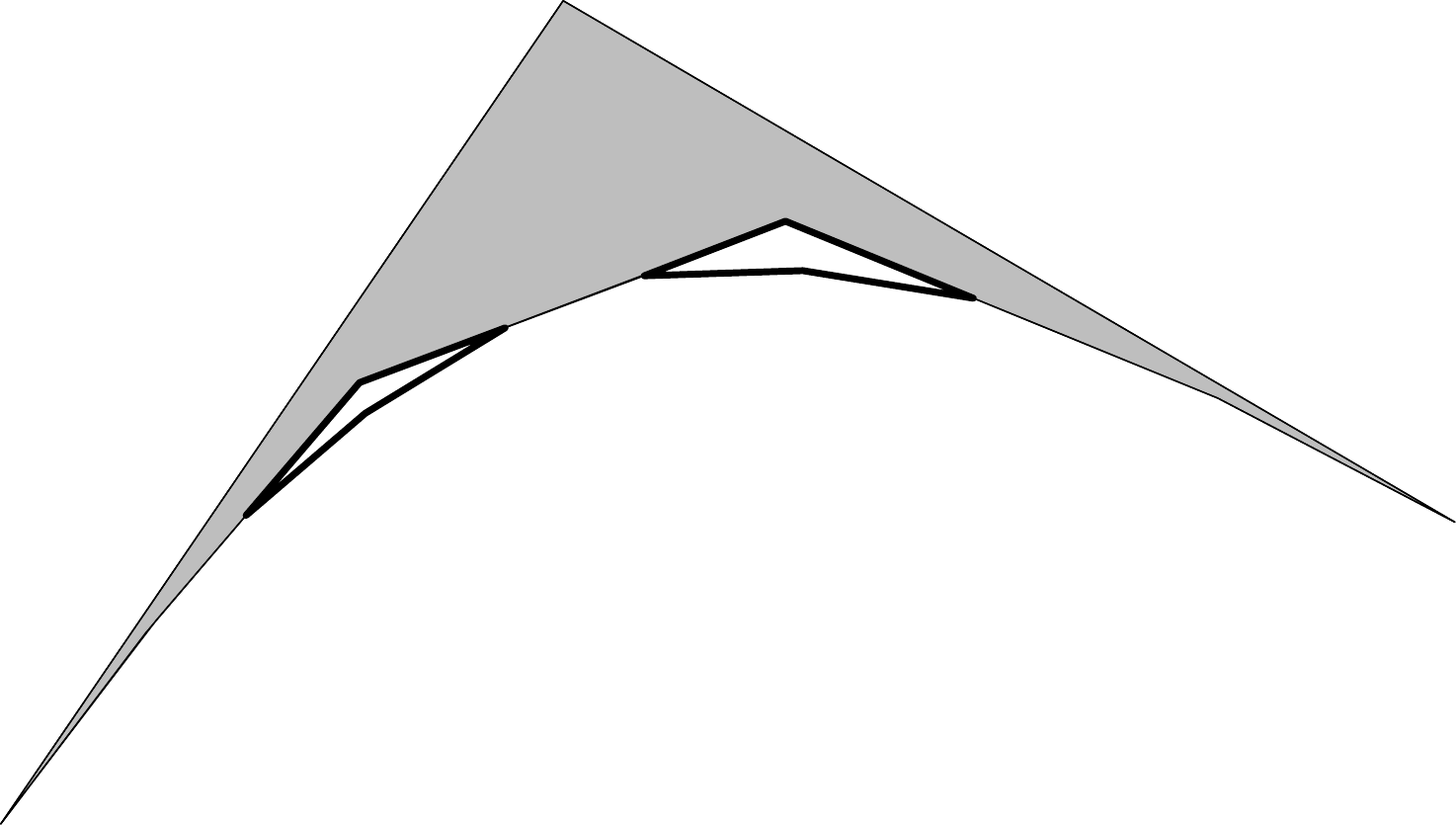}}
	\caption{The grey part is the seed fin $\alpha$ that lies inside the bigger fin shape $\alpha_0 \in \hat{\Delta}_{i-1,*}$.  The difference $\alpha_0 \setminus \alpha$ consists of two subfins shown with bold boundaries.}
	\label{fg:subfin}
\end{figure}

Observe that $\alpha_0 \setminus \alpha$ is a collection of smaller disjoint fin shapes, each having a vertex of $\tail(\alpha)$ as its apex. Figure~\ref{fg:subfin} shows an example.  We call each shape in $\alpha_0 \setminus \alpha$ a \emph{subfin}.  Let \emph{subfins}$[i,j]$ denote the set of all subfins.  That is,
\[
\mathit{subfins}[i,j] = \left\{\alpha_0 \setminus \alpha : \alpha \in \mathit{seed\_fins}[i,j], \alpha_0 \in \hat{\Delta}_{i-1,*}, \alpha \subseteq \alpha_0 \right\}.
\]
Note that a subfin tail endpoint may lie in the interior of a tail edge of a seed fin.

To extract appropriate subsets from $\mathit{subfins}[i,j]$ for the merge with $\mathit{seed\_fins}[i,j]$, we compute a set of pivots and tangent lines, denoted by $\mathit{pivots}[i,j]$ and \emph{tangents}$[i,j]$, respectively.  
\begin{itemize}
	
	\item Step~1: Let $\mathit{pivots}[i,j]$ be the set of tail vertices in $\mathit{seed\_fins}[i,j]$ and tail endpoints in $\mathit{subfins}[i,j]$.  For every subfin $\beta \in \mathit{subfins}[i,j]$ and every vertex $v \in \mathit{pivots}[i,j] \setminus \tail(\beta)$, we compute the tangents from $v$ to $\conv(\tail(\beta))$.  There are at most two such tangents and we insert them into \emph{tangents}$[i,j]$.  
	
	\item Step~2: For every pair $\beta, \gamma \in \mathit{subfins}[i,j]$, we compute all common tangents between $\conv(\tail(\beta))$ and $\conv(\tail(\gamma))$.  Let $s_\beta$ and $s_\gamma$ be the segments that connect the endpoints of $\tail(\beta)$ and $\tail(\gamma)$, respectively.  Adding $s_\beta$ to $\tail(\beta)$ and $s_\gamma$ to $\tail(\gamma)$ gives $\conv(\tail(\beta))$ and $\conv(\tail(\gamma))$.  Since $\tail(\beta)$ and $\tail(\gamma)$ do not cross, the boundaries of $\conv(\tail(\beta))$ and $\conv(\tail(\gamma))$ can intersect only at  $s_\beta \cap \tail(\gamma)$, $s_\gamma \cap \tail(\beta)$, and $s_\beta \cap s_\gamma$.  As $\tail(\gamma)$ is reflex, $s_\beta$ intersects $\tail(\gamma)$ at most twice, and if $s_\beta$ intersects $\tail(\gamma)$ twice, then $s_\beta \cap s_\gamma = \emptyset$ and $s_\gamma \cap \tail(\beta) = \emptyset$.  The same can be said about $\tail(\beta)$ and $s_\gamma$.  Therefore, the boundaries of $\conv(\tail(\beta))$ and $\conv(\tail(\gamma))$ intersect at most twice, implying that there are at most four common tangents to $\conv(\tail(\beta))$ and $\conv(\tail(\gamma))$.  We compute these common tangents and insert them into \emph{tangents}$[i,j]$.
	
\end{itemize}

We use the lines in $\mathit{tangents}[i,j]$ to identify subsets of $\mathit{subfins}[i,j]$ that will be used to augment $\mathit{seed\_fins}[i,j]$.
Specifically, for every $\beta \in \mathit{subfins}[i,j]$ and every $\ell \in \mathit{tangents}[i,j]$, if $\ell$ is tangent to a vertex or some collinear edges of $\tail(\beta)$, then: 
\begin{quote}
	Case 1: $\ell \cap \tail(\beta)$ is a single vertex $v$. 
	
	\begin{quote}
		
	Case 1.1: $v$ is not a vertex of $S_{1,1}$.  By Q1 in Table~\ref{tb:Q}, the two vertices of $\tail(\beta)$ adjacent to $v$, say $w_1$ and $w_2$, are vertices of $S_{1,1}$.  For $i \in [1,2]$, shoot a ray from $w_i$ in the opposite direction of $v$.  Figure~\ref{fg:trim}(a) gives an illustration. 

	\vspace{6pt}	
		
	Case 1.2: $v$ is a vertex of $S_{1,1}$.  We shoot two rays from $v$ through $w_1$ and $w_2$.  Then, we rotate $\overrightarrow{vw_i}$ around $v$ slightly so that the rotated ray shoots into the interior of $\beta$ in a direction that lies between $vw_i$ and $\ell$.
	Figure~\ref{fg:trim}(b) gives an illustration.  
	
	\end{quote}

	In Case~1.1, cutting along the two rays and the edges $vw_1$ and $vw_2$ extracts a wedge out of $\beta$ that contains $\apex(\beta)$.  Similarly, in case~1.2, a wedge that contains $\apex(\beta)$ is extracted by cutting along the two rays from $v$.  We include this wedge in the set $\mathit{wedges}(\beta)$.  The vertex $v$ and, if applicable, $w_1$ and $w_2$ as well, are retained as vertices on the boundary of the wedge.  If $v$ is an endpoint of $\tail(\beta)$, either $w_1$ or $w_2$ is undefined and the wedge is cut out using only one ray.
	
	\vspace{6pt}

	\newcounter{pagecount}
	\setcounter{pagecount}{\thepage}
		
	Case 2: $\ell \cap \tail(\beta)$ consist of some collinear edges. Let $v_1$ and $v_2$ be the endpoints of this collinear sequence of edges.  Let $w_i$ be the vertex of $\tail(\beta)$ that is adjacent to $v_i$ and outside $\ell \cap \tail(\beta)$.  Depending on whether $v_i$ is a vertex of $S_{1,1}$, we shoot a ray from $v_i$ or $w_i$ as in Cases~1.1 and~1.2.  
	Figures~\ref{fg:trim}(c) and (d) show the situation.  We extract a wedge out of $\beta$ that contains $\apex(\beta)$ by cutting along the edges $v_1 \cdots v_2$ and the two rays obtained.  We include this wedge in $\mathit{wedges}(\beta)$.  The vertices $v_1$, $v_2$, and if applicable, $w_1$ and $w_2$ as well, are retained as vertices on the boundary of the wedge.  The vertices in $v_1 \cdots v_2$ other than $v_1$ and $v_2$ do not appear as boundary vertices of the wedge.  If $v_i$ is an endpoint of $\tail(\beta)$, then $w_i$ is not defined, and neither is the ray from $v_i$ or $w_i$.

\end{quote}

\begin{figure}
	\centering
	\begin{tabular}{ccc}
		\includegraphics[scale=0.9]{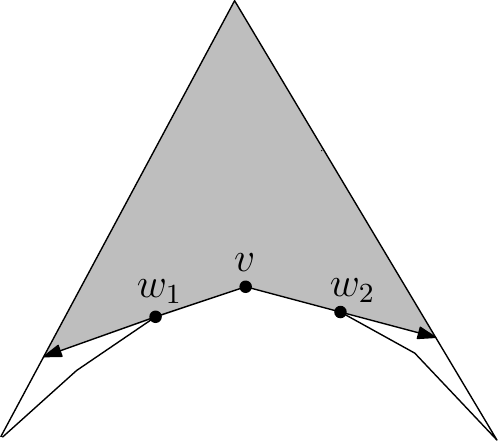} & &
		\includegraphics[scale=0.9]{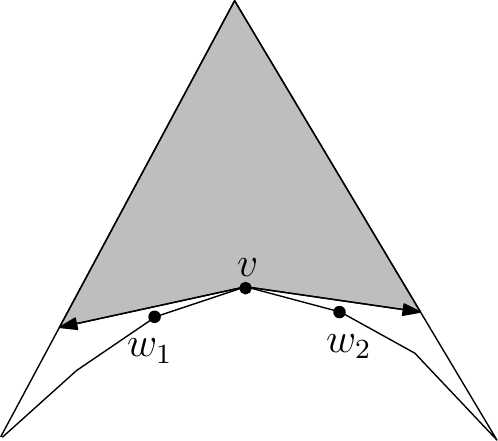} \\ \\
		(a) & \hspace*{.1in} & (b) \\ \\
		\includegraphics[scale=0.9]{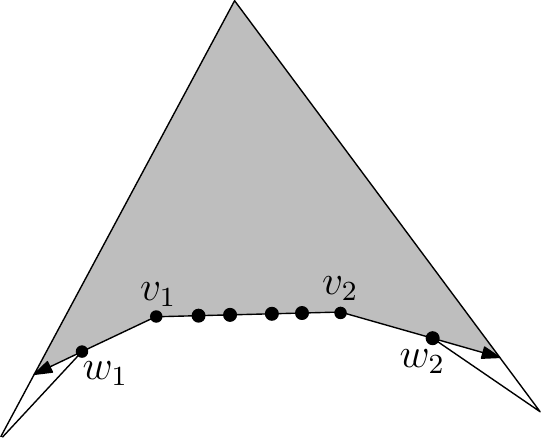} & &
		\includegraphics[scale=0.9]{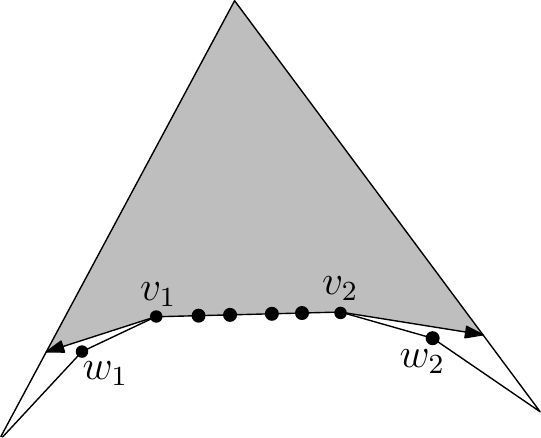} \\ \\
		(c) & \hspace*{.1in} & (d) 
	\end{tabular}
	\caption{The four figures show a subfin $\beta$.  The grey parts are wedges to be included in $\mathit{wedges}(\beta)$.}
	\label{fg:trim}
\end{figure}

Let $\bigcup \mathit{wedges}(\beta)$ denote the union of all wedges collected inside $\beta$.  We use this union of wedges to augment $\mathit{seed\_fins}[i,j]$ as defined below.
\begin{quote}
Every $\alpha \in \mathit{seed\_fins}[i,j]$ induces an augmented seed fin, $\mathit{augment}(\alpha)$, which is the union of $\alpha$ and $\bigcup \mathit{wedges}(\beta)$ for all $\beta \in \mathit{subfins}[i,j]$ such that $\apex(\beta) \in \tail(\alpha)$.
\end{quote}
We also define a set of smaller fin shapes that will be useful in the subsequent analysis.
\begin{quote}
	The set of \emph{trimmed subfins} of $\beta$, denoted by $\mathit{trimmed\_subfins}(\beta)$, is the set of fin shapes in $\beta \setminus \bigcup \mathit{wedges}(\beta)$.
\end{quote} 

For a subfin $\beta$, the support line $\ell$ of any head side of $\beta$ belongs to $\mathit{tangents}[i,j]$ because $\apex(\beta) \in \mathit{pivots}[i,j] \setminus \tail(\beta)$.  So $\ell$ triggers a trimming of $\beta$.  As a result, every head side of every trimmed subfin is generated as part of a ray in Case~1 or~2 above.  Therefore, no trimmed subfin head side lies on the boundary of any seed fin.

Figure~\ref{fg:pre-edge-1} shows an example of trimming a subfin to augment a seed fin.  The next result follows from our previous discussion.  It characterizes the tails of augmented seed fins.

\begin{lemma}
	\label{lem:augmented}
	Let $\alpha$ be a seed fin in $\mathit{seed\_fins}[i,j]$.  Let $\alpha_0$ be the fin in $\hat{\Delta}_{i-1,*}$ that contains $\alpha$.  Every tail edge $e$ of $\mathit{augment}(\alpha)$ satisfies one of the following conditions:
	\begin{emromani}
		\item $e$ is the union of some tail edges of $\alpha_0$, or
		\item $e$ is a head side of a trimmed subfin, 
		and the interior of $e$ lies strictly inside $\alpha_0 \setminus \alpha$.
	\end{emromani}
\end{lemma}
\cancel{
\begin{proof}
	The correctness of (i) follows immediately from our construction of augmented seed fins.  Consider (ii). Let $e$ be any tail edge of $\mathit{augment}(\alpha)$.  If $e$ is not produced during the augmentation of $\alpha$, the endpoints of $e$ are tail vertices of $\alpha$, and therefore, they belong to $\mathit{pivots}[i,j]$.  Suppose that $e$ is produced during the augmentation of $\alpha$.  If $e$ is a head side of a trimmed subfin, its endpoints include the apex of that trimmed subfin and a vertex of $S_{1,1}$ (the source of the ray-shooting that makes $e$ a head side). In the remaining case, either $e$ is $vw_1$ or $vw_2$ in Case~1 on page~\arabic{pagecount}, or $e$ is $v_1w_1$, $v_2w_2$, or $v_1v_2$ in Case~2 on page~\arabic{pagecount}.  In the first subcase, $v$ is a contact vertex in $\mathit{pivots}[i,j]$.  In the second subcase, $v_1$ and $v_2$ are contact vertices in $\mathit{pivots}[i,j]$.  In both subcases, $w_1$ and $w_2$ are vertices of $S_{1,1}$.  
\end{proof}
}

We prove some properties of the trimmed subfins below.

\begin{lemma}
	\label{lem:suboutreg}
	For every $\beta \in \mathit{subfins}[i,j]$ and every $\gamma \in \mathit{trimmed\_subfins}(\beta)$, 
	\begin{emromani}
		\item $\tail(\gamma)$ is a contiguous subsequence of edges in $\tail(\beta)$; 
		\item for every head side $xu$ of $\gamma$, where $x = \apex(\gamma)$, $u$ is a vertex of $S_{1,1}$; and	
		\item for every head side of $\gamma$, its support line $\ell$ does not contain any vertex in $\mathit{pivots}[i,j]$ that lies outside $\tail(\beta)$, and $\ell$ does not contain any head side of another trimmed subfin in $\bigcup_{\beta' \in \mathit{subfins}[i,j]} \mathit{trimmed\_subfins}(\beta')$.
	\end{emromani}
\end{lemma}
\begin{proof}
	The correctness of (i) is immediate by our construction.  
	
	Consider (ii).  A head side $xu$ of $\gamma$, where $x = \apex(\gamma)$, is created as part of a ray generated in Case~1 or~2 on page~\arabic{pagecount}.  In this case, $u$ is the source of this ray which is guaranteed to a vertex of $S_{1,1}$.  
	
	Consider (iii).  Let $xu$ be a head side of $\gamma$, where $x = \apex(\gamma)$.  Let $\ell$ be the support line of $xu$.  Let $\beta$ be the subfin such that $\gamma \in \mathit{trimmed\_subfins}(\beta)$.  Since $xu$ is a head side of $\gamma$, $\ell$ is tangent to $\tail(\beta)$ at $u$. 
	
	If $\ell$ contains a vertex in $\mathit{pivots}[i,j] \setminus \tail(\beta)$, then $\ell \in \mathit{tangents}[i,j]$.  But then $\gamma$ should have been trimmed further by a ray produced in Case~1 or~2 on page~\arabic{pagecount} that shoots from $u$ or its neighboring vertex in $\tail(\gamma)$, a contradiction.
	
	Next, we assume for the sake of contradiction that $\ell$ contains the head side $yv$ of another trimmed subfin $\gamma'$ where $y = \apex(\gamma')$.  By our construction, no line can contain the head sides of two trimmed subfins that are contained in the same subfin.  Thus, $\gamma$ and $\gamma'$ are contained in two distinct subfins, say $\beta$ and $\beta'$, respectively.  Since $\ell$ contains $xu$ and $yv$, $\ell$ is tangent to $\tail(\beta)$ at $u$ and $\tail(\beta')$ at $v$. Thus, $\ell \in \mathit{tangents}[i,j]$.  But then $\gamma$ would have been trimmed further by a ray that shoots from $u$ or its neighboring vertex in $\tail(\gamma)$, a contradiction.
\end{proof}

We bound the total complexity of augmented seed fins and their construction time.

\begin{lemma}
	\label{lem:augment}
	The total complexity of $\bigl\{\mathit{augment}(\alpha) : \alpha \in \mathit{seed\_fins}[i,j]\bigr\}$ is $O((\log n_{i-1})^{4c_1-4c_2+6})$.  The augmented seed fins can be constructed in $O((\log n_{i-1})^{4c_1-4c_2+7})$ time.
\end{lemma}
\begin{proof}
The complexity of $\mathit{seed\_fins}[i,j]$ is $O((\log n_{i-1})^{2c_1-2c_2+3})$ by Lemma~\ref{lem:seed-fins}, so the number of subfins is $O((\log n_{i-1})^{2c_1-2c_2+3})$.  Each pair of subfins may generate up to four common tangents and hence $O(1)$ corresponding wedges.  Similarly, each vertex in $\mathit{pivots}[i,j]$ (i.e., tail vertex in $\mathit{seed\_fins}[i,j]$ or tail endpoint in $\mathit{subfins}[i,j]$) may generate up to two tangents with a subfin.  Each wedge causes an $O(1)$ increase in complexity.  Therefore, augmenting $\mathit{seed\_fins}[i,j]$ with the wedges may increase the complexity to $O((\log n_{i-1})^{4c_1-4c_2+6})$.

A common tangent can be found in $O(\log n_{i-1})$ time, resulting in $O((\log n_{i-1})^{4c_1-4c_2+7})$ time overall.  
Gluing the wedges to the seed fins can be done in time linear in the sum of complexities of the seed fins and wedges.
\end{proof}

\begin{figure}
	\centerline{\includegraphics[scale=.45]{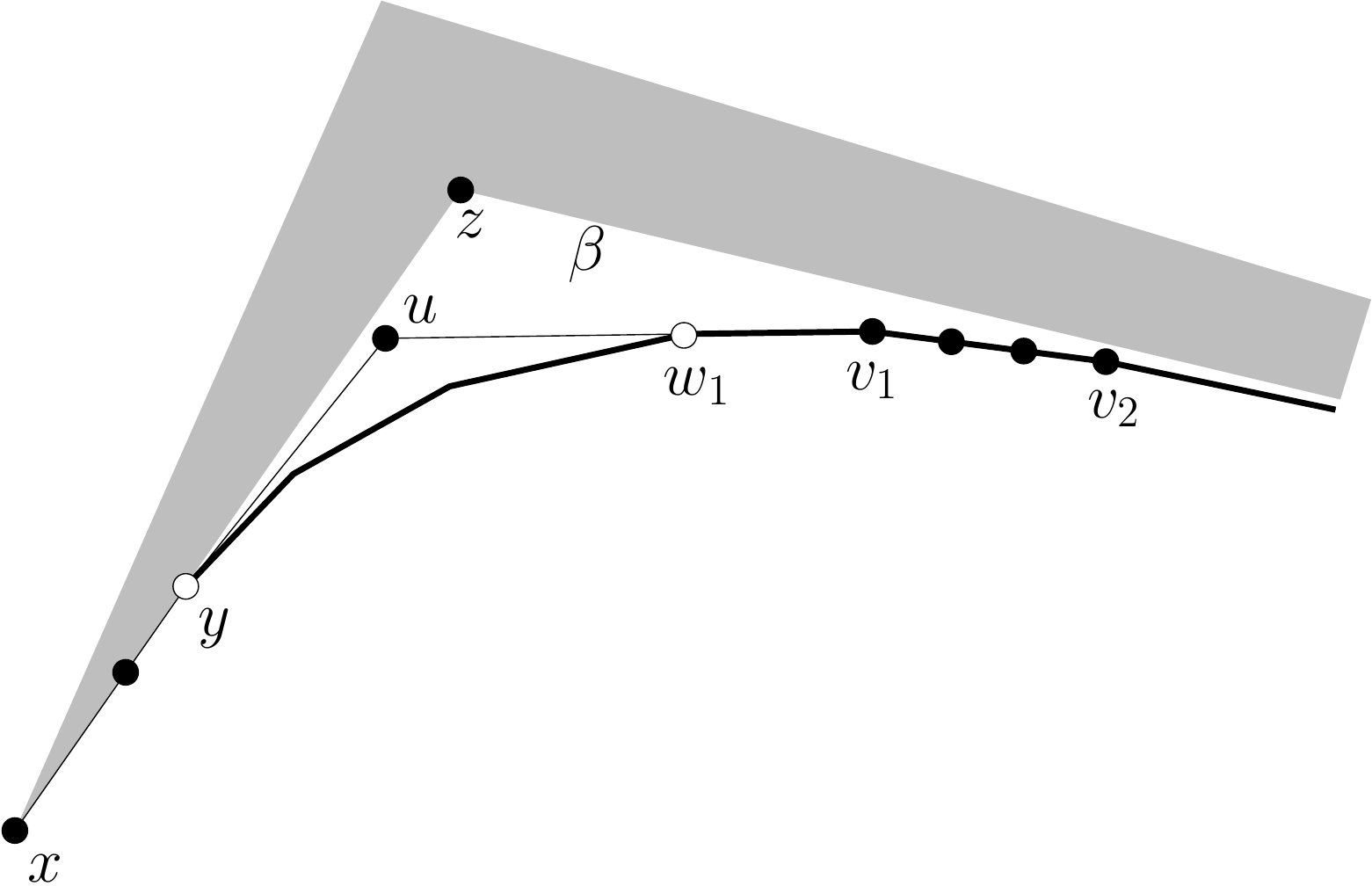}}
	\caption{The grey area shows a part of a seed fin.  A subfin $\beta$ is shown, including its apex $z$, one head size $yz$, part of the other head side, and part of its tail---the bold reflex chain.  The tail edge $xz$ of the seed fin contains the head side $yz$ of $\beta$. Suppose that $y$ is a vertex of $S_{1,1}$ and $v_1$ is not a vertex of $S_{1,1}$.  So the neighbor $w_1$ of $v_1$ is a vertex of $S_{1,1}$.   Suppose that a ray is shot from $y$ towards $u$ to prune $\beta$. Suppose that a line in $\mathit{tangents}[i,j]$ passes through the collinear edges $v_1 \cdots v_2$, which generates a ray from $w_1$ towards $u$ to prune $\beta$.  Assume that there is no more pruning of $\beta$.  
		A trimmed subfin in $\mathit{trimmed\_subfins}(\beta)$ has apex $u$ and tail endpoints $y$ and $w_1$.  The segments $xy$, $yu$, $uw_1$, $w_1v_1$, and $v_1v_2$ become some of the tail edges of the augmented seed fin.}
	\label{fg:pre-edge-1}
\end{figure}

\subsubsection{The $j$-th structures at the $i$-th layer}

We are ready to describe the construction of $R_{i,j}$, $S_{i,j}$, $\tilde{\Delta}_{i,j}$, $\hat{\Delta}_{i,j}$, and $\Delta_{i,j}$.  As before, $D_{i,j}$ consists of two data structures, one obtained by applying Theorem~\ref{thm:iacono} to $\Delta_{i,j}$ and a worst-case optimal planar point location structure.

Let $\cal T$ denote the subset of geodesic triangles in $\tilde{\Delta}_{i-1,*}$ that contain seed fins in \emph{seed\_fins}$[i,j]$.    For every geodesic triangle $\tau \in {\cal T}$, define:
\[
\mathit{shrink}(\tau) = \mathit{star}(\tau) \cup \bigcup_{\substack{\alpha \in \mathit{seed\_fins}[i,j] \\ \alpha \subseteq \tau}} \mathit{augment}(\alpha).
\]
Note that $\mathit{shrink}(\tau) \subseteq \tau$.  The set $R_{i,j}$ of regions is built as follows.

\begin{enumerate}
	\item For every region $r$ in $S_{1,1}$ that contains some geodesic triangles in $\cal T$, we use $U_r$ to denote $\bigcup_{\tau \in {\cal T},\,\tau \subseteq r } \mathit{shrink}(\tau)$.  The interior of $U_r$ consists of one or more connected components, i.e., $U_r$ may be pinched at one or more vertices.  Divide $U_r$ into pieces by cutting at the pinched vertices, if there is any, and let $C_r$ be the set of these pieces.  The set $C_r$ can be constructed by a plane sweep over $\mathit{shrink}(\tau)$ for all $\tau \in {\cal T}$ that lie inside $r$.  

	\item Let ${\cal C} := \bigcup_{\text{region $r \in S_{1,1}$}} C_r$, which is a set of regions.  
	
	\item For every pair of adjacent regions in $\cal C$, ensure that their boundaries are compatible, that is, if a vertex $v$ of a region appears in the interior of an edge $e$ of an adjacent region, then insert $v$ into $e$. 
	
	\item Afterwards, for every resulting region boundary edge $e'$, if neither endpoint of $e'$ is a vertex of $S_{1,1}$, we find a vertex $w$ of $S_{1,1}$ in the interior of $e'$ and insert $w$ into $e'$.  The final set of regions produced is $R_{i,j}$.
	
\end{enumerate}
The subdivision $S_{i,j}$ is defined as follows.
\begin{quote}
Fill the space among the regions in $R_{i,j}$ inside the enclosing triangle $B_S$ with dummy triangles.  No additional vertex is introduced.  This can be done using a plane sweep.  The resulting subdivision is $S_{i,j}$.
\end{quote}

We prove a property of $R_{i,j}$ below.  It shows that step~4 of the construction of $R_{i,j}$ always succeeds.

\begin{lemma}
	\label{lem:pre-edge}
	Every boundary edge of every region in $R_{i,j}$ is incident to a vertex of $S_{1,1}$.
\end{lemma}
\begin{proof}
	This property is guaranteed by step~4 in the construction of $R_{i,j}$ above.  It suffices to show that, in step~4, if neither endpoint of an edge $e'$ is a vertex of $S_{1,1}$, we can find a vertex $w$ of $S_{1,1}$ in the interior of $e'$.  
	
	Let $\cal T$ denote the set of geodesic triangles in $\tilde{\Delta}_{i-1,*}$ that contain seed fins in $\mathit{seed\_fins}[i,j]$.  For every geodesic triangle $\tau \in {\cal T}$ and every $\alpha \in \mathit{fins}(\tau)$, we include the triangular regions incident to convex vertices of $\tau$ as well as their ancestors in $T_\alpha$ in forming the seed fins.  Therefore, $\mathit{star}(\tau)$ does not contribute to the boundary of $\mathit{shrink}(\tau)$.  As a result, only tail edges of augmented seed fins can contribute to the boundary of $R_{i,j}$.  
	
	If $e'$ is a head side of a trimmed subfin, then by Lemma~\ref{lem:augmented}(ii), $e'$ does not merge with any other edge in the union of augmented seed fins.  Therefore, $e'$ will become a boundary edge of $R_{i,j}$, and it is incident to a vertex of $S_{1,1}$.  
	
	If $e'$ is not a head side of any trimmed subfin, it may be the whole or some part of a tail edge of an augmented seed fin.  Lemma~\ref{lem:augmented}(i) implies that $e'$ is a union of tail edges of some fin in $\hat{\Delta}_{i-1,*}$.  As $R_{i-1,*}$ satisfies Q1 in Table~\ref{tb:Q} inductively, every edge in $\tilde{\Delta}_{i-1,*}$ is incident to a vertex of $S_{1,1}$.  If neither endpoints of $e'$ is a vertex of $S_{1,1}$, then $e'$ consists of at least two tail edges in $\hat{\Delta}_{i-1,*}$.  The tail vertices in $\hat{\Delta}_{i-1,*}$ that lie in $e'$ next to its two endpoints are vertices of $S_{1,1}$.
\end{proof}

We bound the total complexities and construction times of $R_{i,j}$, $S_{i,j}$, $\tilde{\Delta}_{i,j}$, $\hat{\Delta}_{i,j}$, and $\Delta_{i,j}$,

\begin{lemma}
	\label{lem:Delta}
	The complexities of $R_{i,j}$, $S_{i,j}$, $\tilde{\Delta}_{i,j}$, $\hat{\Delta}_{i,j}$, and $\Delta_{i,j}$ are $O((\log n_{i-1})^{4c_1-4c_2+6})$.  They can be constructed in $O((\log n_{i-1})^{4c_1-4c_2+7})$ time.
\end{lemma}
\begin{proof}
	Let $\cal T$ be the set of geodesic triangles in $\tilde{\Delta}_{i-1,*}$ that contain seed fins in $\mathit{seed\_fins}[i,j]$.  Since there are $O((\log n_{i-1})^{2c_1-2c_2+2})$ seed fins by Lemma~\ref{lem:seed-fins}, $|{\cal T}| = O((\log n_{i-1})^{2c_1-2c_2+2})$.  For every $\tau \in {\cal T}$ and every $\alpha \in \fins(\tau)$, since the triangular regions in $T_\alpha$ incident to convex vertices of $\tau$ as well as their ancestors in $T_\alpha$ are used in forming seed fins, $\st(\tau)$ does not contribute to the boundary of $\mathit{shrink}(\tau)$.  By Lemma~\ref{lem:augment}, the total complexity of all augmented seed fins is $O((\log n_{i-1})^{4c_1-4c_2+6})$.  Therefore, the total complexity of all $\mathit{shink}(\tau)$'s is $O((\log n_{i-1})^{4c_1-4c_2+6})$.
	
	The complexities of $R_{i,j}$ and $S_{i,j}$ are thus $O((\log n_{i-1})^{4c_1-4c_2+6})$.  The balanced geodesic triangulations of the regions in $R_{i,j}$ have a total complexity linear in the complexity of $R_{i,j}$.  So $\tilde{\Delta}_{i,j}$ has $O((\log n_{i-1})^{4c_1-4c_2+6})$ complexity.  Partitioning each geodesic triangle in $\tilde{\Delta}_{i,j}$ into fins and stars does not increase the asymptotic complexity.  Neither does the triangulation of $\hat{\Delta}_{i,j}$ to form $\Delta_{i,j}$.  So both $\hat{\Delta}_{i,j}$ and $\Delta_{i,j}$ have $O((\log n_{i-1})^{4c_1-4c_2+6})$ complexities.
	
	By Lemma~\ref{lem:augment}, the augmented seed fins can be constructed in $O((\log n_{i-1})^{4c_1-4c_2+7})$ time.  Since the construction of $R_{i,j}$ involves plane sweeps, it takes $O((\log n_{i-1})^{4c_1-4c_2+6}\log\log n_{i-1})$ time.  The plane sweep to form $S_{i,j}$ also takes $O((\log n_{i-1})^{4c_1-4c_2+6} \log\log n_{i-1})$ time.  The balanced geodesic triangulations of the regions in $R_{i,j}$ and hence $\tilde{\Delta}_{i,j}$ can be constructed in $O((\log n_{i-1})^{4c_1-4c_2+6})$ time.  It takes only $O(1)$ time per geodesic triangle to form $\hat{\Delta}_{i,j}$.  Finally, triangulating $\hat{\Delta}_{i,j}$ to form $\Delta_{i,j}$ takes $O(|\hat{\Delta}_{i,j}|) = O((\log n_{i-1})^{4c_1-4c_2+6})$ time.
\end{proof}

We have discussed  the construction of the structures, their complexities, and the processing time needed.  We now prove the invariants Q1--Q3 in Table~\ref{tb:Q}.  We first show that $R_{i,j}$ satisfies Q1 in Table~\ref{tb:Q}.

\begin{lemma}
	\label{lem:edgeTouchS}
	$R_{i,j}$ satisfies Q1 in Table~\ref{tb:Q}.
\end{lemma}
\begin{proof}
	The lemma is clearly true for $R_{1,1}$.  Assume inductively that the lemma is true for all level index $< i$ and for all version index $\geq 1$.  Lemma~\ref{lem:pre-edge} states that every boundary edge of $R_{i,j}$ is incident to a vertex of $S_{1,1}$.  The only worry is whether some edge in the balanced geodesic triangulations of regions in $R_{i,j}$ may connect two vertices that are not vertices of $S_{1,1}$.
	
	Assume to the contrary that such an edge $x_1x_2$ exists.  Let $r_i$ denote the region in $R_{i,j}$ that contains $x_1x_2$.    Let $\partial$ denote the boundary operator that returns the boundary of a region.
	Since $x_1$ and $x_2$ are not vertices of $S_{1,1}$, they must be created at level $i$ or earlier.  By our construction method, a newly created vertex must be the apex of a trimmed subfin.  Therefore, for $s \in [1,2]$, there exists $i_s \leq i$, $j_s \geq 1$, $\beta_s \in \mathit{subfins}[i_s,j_s]$, and $\gamma_s \in \mathit{trimmed\_subfins}(\beta_s)$ such that $x_s = \apex(\gamma_s)$.  (It is possible that a trimmed subfin is actually a subfin.)   
	
	The edge $x_1x_2$ lies on a shortest path between two of the three defining vertices of a kite.  Let $\rho$ denote this shortest path.  Let $\ell$ denote the support line of $x_1x_2$.  By Lemma~\ref{lem:shift}, the endpoints of $\rho$ are vertices of $S_{1,1}$, which means that $x_1$ and $x_2$ are not endpoints of $\rho$.  Then, as $\rho$ is a shortest path, $\ell$ is tangent to $\partial r_i$ locally at $x_1$ and $x_2$.  
	
	For $s \in [1,2]$, let $r_{i_s}$ be the region in $R_{i_s,j_s}$ that contains $r_i$.  Since $r_i \subseteq r_{i_s}$, $\ell$ is also tangent to $\partial r_{i_s}$ locally at $x_s$.  Therefore, $\gamma_s$ lies on one side of $\ell$.  Starting from $r_i$, we can trace a sequence of regions from level $i$ to level $\min(i_1,i_2)$ such that any region in the sequence is contained in the next region in the sequence.  For $s \in [1,2]$, $x_s$ appears as a region vertex from level $i_s$ to level $i$ in this sequence.  Also, $\ell$ is tangent to the region boundary locally at $x_s$ from level $i_s$ to level $i$ in this sequence.  Without loss of generality, assume that $i_1 \leq i_2$.  
	
	We prove that $x_1 \not\in \tail(\beta_2)$.  Assume to the contrary that $x_1$ is a vertex of $\tail(\beta_2)$.  The first possibility is that $x_1x_2$ does not contain any endpoint of $\tail(\gamma_2)$.  In this case, since $\tail(\beta_2)$ is reflex and the head sides of $\gamma_2$ are tangent to $\tail(\beta_2)$, the part of $x_1x_2$ near $x_2$ lies inside $\gamma_2$. Note that $\gamma_2$ lies outside $r_{i_2}$.  But this contradicts the fact that the line $\ell$ through $x_1x_2$ is tangent to $\partial r_{i_2}$ at $x_2$.  Figure~\ref{fg:precase}(a) gives an illustration.  The second possibility is that $x_1x_2$ contains an endpoint $v$ of $\tail(\gamma_2)$.  Refer to Figure~\ref{fg:precase}(b).  Thus, $x_2v$ is a head side of $\gamma_2$.  Recall that every head side of $\gamma_2$ is a boundary edge of $r_{i_2}$.  So $x_2v \subseteq \partial r_{i_2}$.  As $r_i \subseteq r_{i_2}$ and $x_1x_2 \subseteq r_i$, we have $x_2v \subseteq \partial r_i$.  If $x_1x_2 \subseteq \partial r_i$, then $x_1x_2$ is a boundary edge of $r_i$ as it is a single edge of the balanced geodesic triangulation of $r_i$.  But this contradicts Lemma~\ref{lem:pre-edge} because neither $x_1$ nor $x_2$ is a vertex of $S_{1,1}$.  If $x_1x_2 \not\subseteq \partial r_i$, then as we walk from $v$ towards $x_1$, we must exit $\partial r_i$ at a boundary vertex of $r_i$ before reaching $x_1$.  But this contradicts the fact that $x_1x_2$ is a single edge of the balanced geodesic triangulation of $r_i$.  This proves that $x_1 \not\in \tail(\beta_2)$.
	
	\begin{figure}
		\centering
		\begin{tabular}{ccc}
			\includegraphics[scale=0.55]{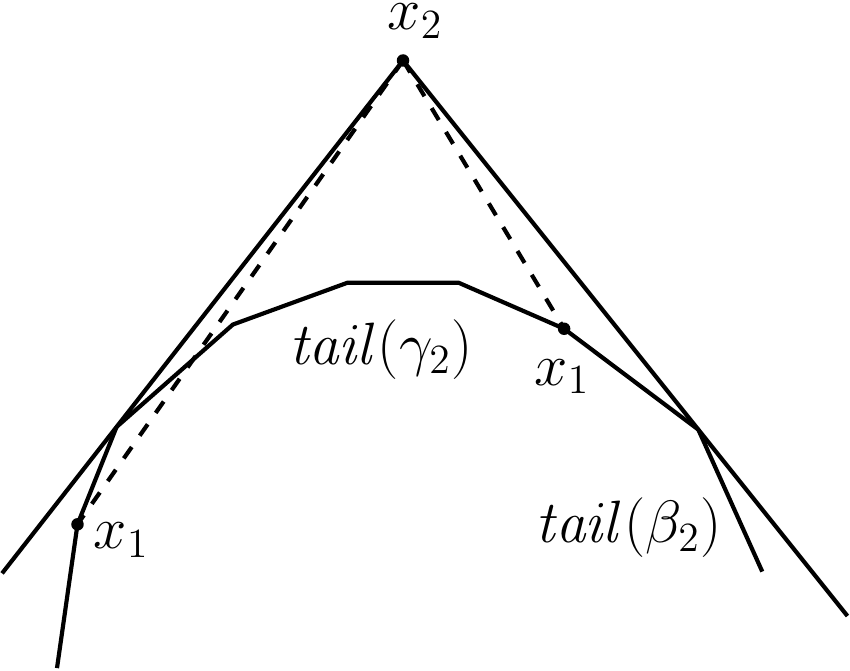} & & 
			\includegraphics[scale=0.55]{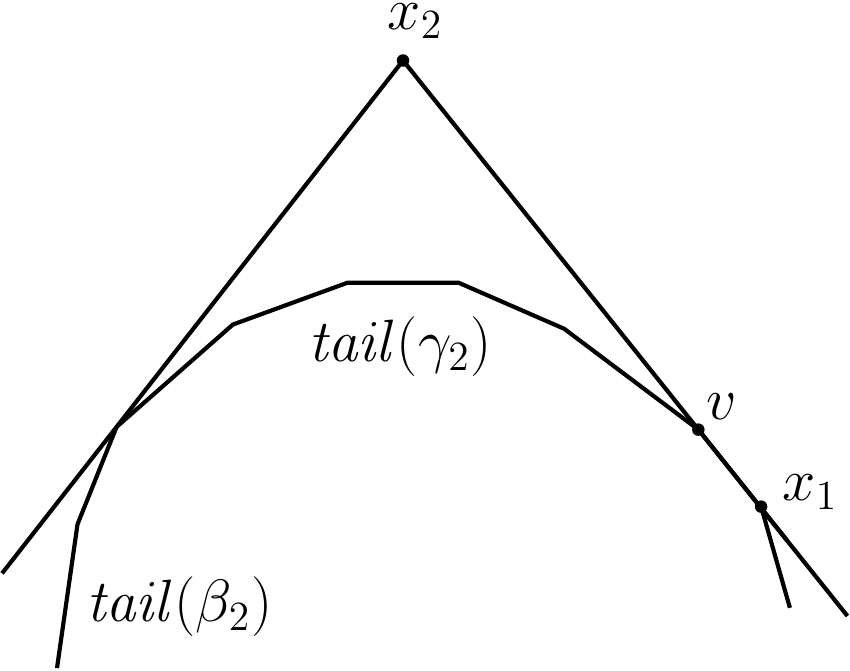} \\ \\
			(a) & \hspace*{.2in} & (b) \\
		\end{tabular}
		\caption{Impossible configurations of $x_1 \in \tail(\beta_2)$.}
		\label{fg:precase}
	\end{figure}
		  
	There are two cases to analyze depending on $i_1 < i_2$ or not.
	
	\vspace{4pt}
	
	Case 1: $i_1 < i_2$.  So $x_1$ is inherited from level $i_1$ to level $i_2$. 
	
	\vspace{4pt}
	
	\begin{figure}
		\centerline{\includegraphics[scale=0.6]{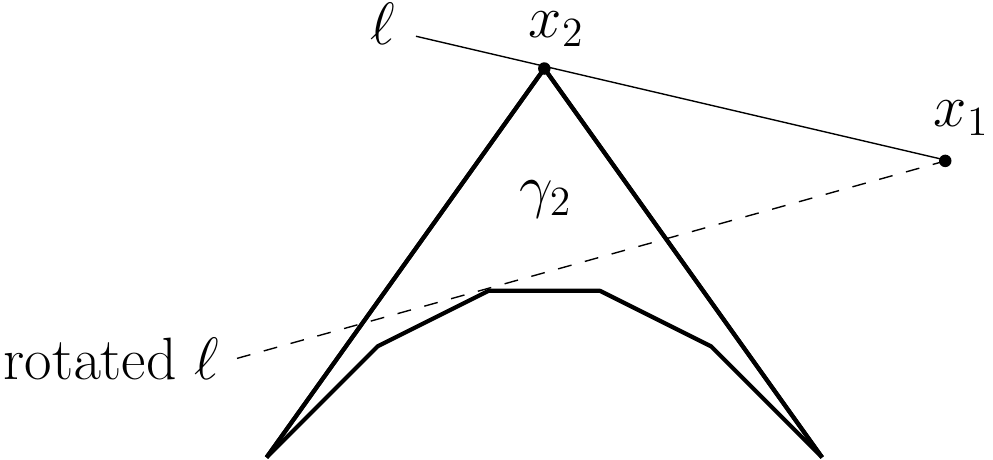}}
		\caption{Rotate $\ell$ to obtain a tangent to $\tail(\gamma_2)$ that intersects the interior of $\gamma_2$.}
		\label{fg:tangent1}
	\end{figure}
	
	Case 1.1: $x_1 \in \mathit{pivots}[i_2,j_2]$.  We have shown earlier that $x_1 \not\in \tail(\beta_2)$.  By Lemma~\ref{lem:suboutreg}(iii), $\ell$ cannot contain any head side of $\gamma_2$.  Figure~\ref{fg:tangent1} shows a possible configuration.  One can turn $\ell$ around $x_1$ in a direction to intersect the interior of $\gamma_2$.  Rotate $\ell$ around $x_1$ further in this direction to obtain a tangent from $x_1$ to $\tail(\gamma_2)$ that intersects the interior of $\gamma_2$.  This tangent belongs to $\mathit{tangents}[i_2,j_2]$ because $\tail(\gamma_2)$ is a contiguous subsequence of $\tail(\beta_2)$ by Lemma~\ref{lem:suboutreg}(i).  But then this tangent should have caused $\gamma_2$ to be trimmed further in Case~1 or~2 on page~\arabic{pagecount}, a contradiction.
	
	\vspace{4pt}
	
	Case 1.2: $x_1 \not\in \mathit{pivots}[i_2,j_2]$.  Then, there is a line $\ell' \in \mathit{tangents}[i_2,j_2]$ and a subfin $\beta' \in \mathit{subfins}[i_2,j_2]$ such that $x_1$ is an interior vertex of $\tail(\beta')$, $x_1 \in \ell' \cap \tail(\beta')$, and $x_1$ is added to some augmented seed fin boundary when we trim $\beta'$ with $\ell'$.\footnote{The trimming caused by $\ell'$ may add vertices that do not lie on $\ell'$, including $w_1$ and $w_2$ in Cases 1 and 2 on page \arabic{pagecount}.  The vertices $w_1$ and $w_2$ are vertices of $S_{1,1}$, so $x_1$ is not one of them.}  We have shown before that $x_1 \not\in \tail(\beta_2)$.  It follows that $\beta' \not= \beta_2$ and so $\beta'$ and $\beta_2$ are disjoint.  There are two subcases.
	
	\vspace{4pt}
	
	\begin{figure}
		\centering
		\begin{tabular}{ccc}
			\includegraphics[scale=0.55]{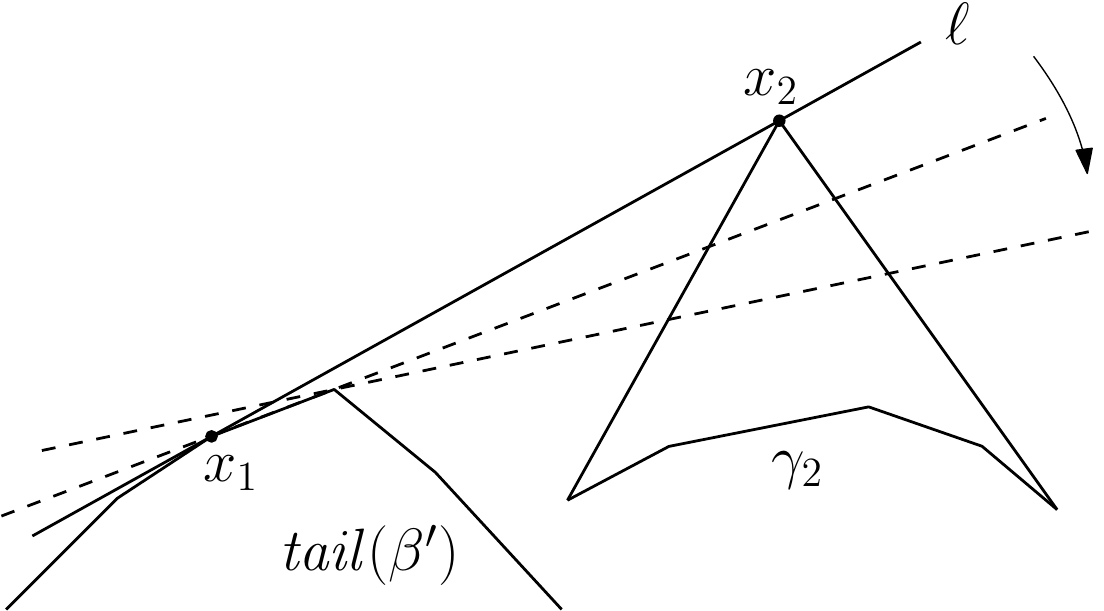} & & 
			\includegraphics[scale=0.55]{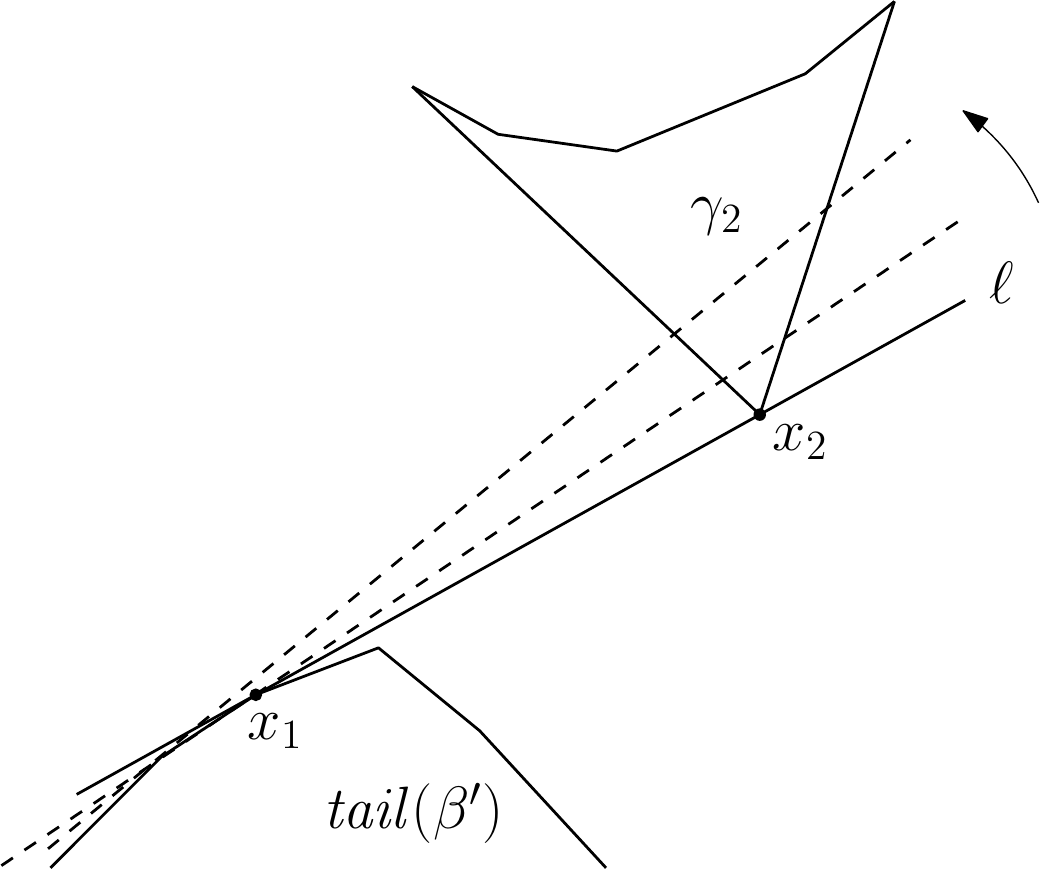} \\ \\
			(a) & \hspace{.15in} & (b) \\ \\
			\includegraphics[scale=0.55]{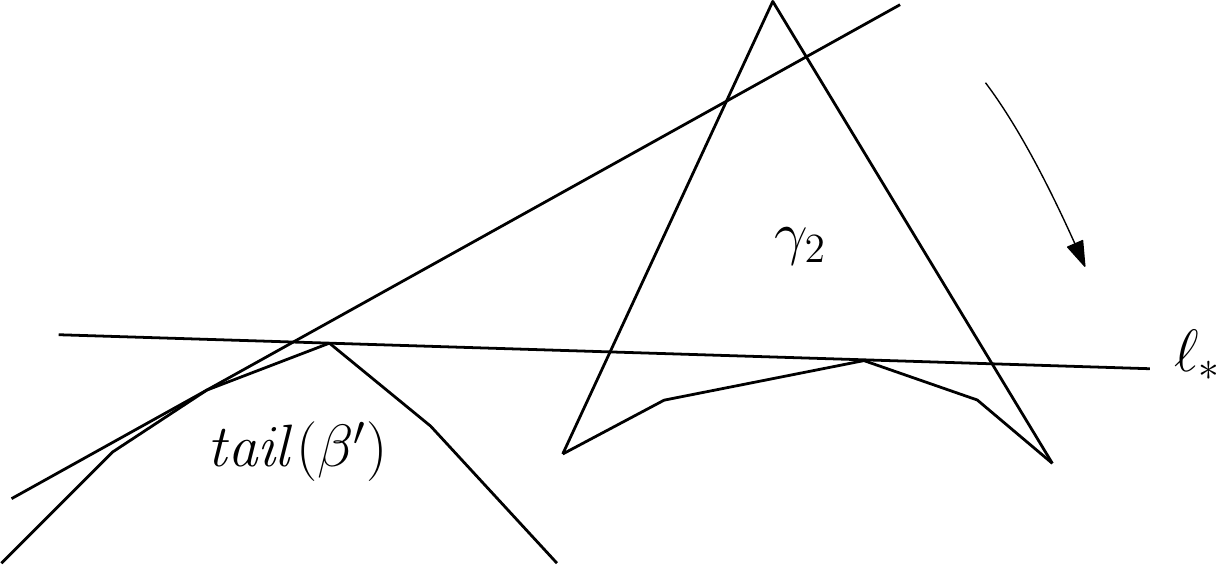} & & 
			\includegraphics[scale=0.55]{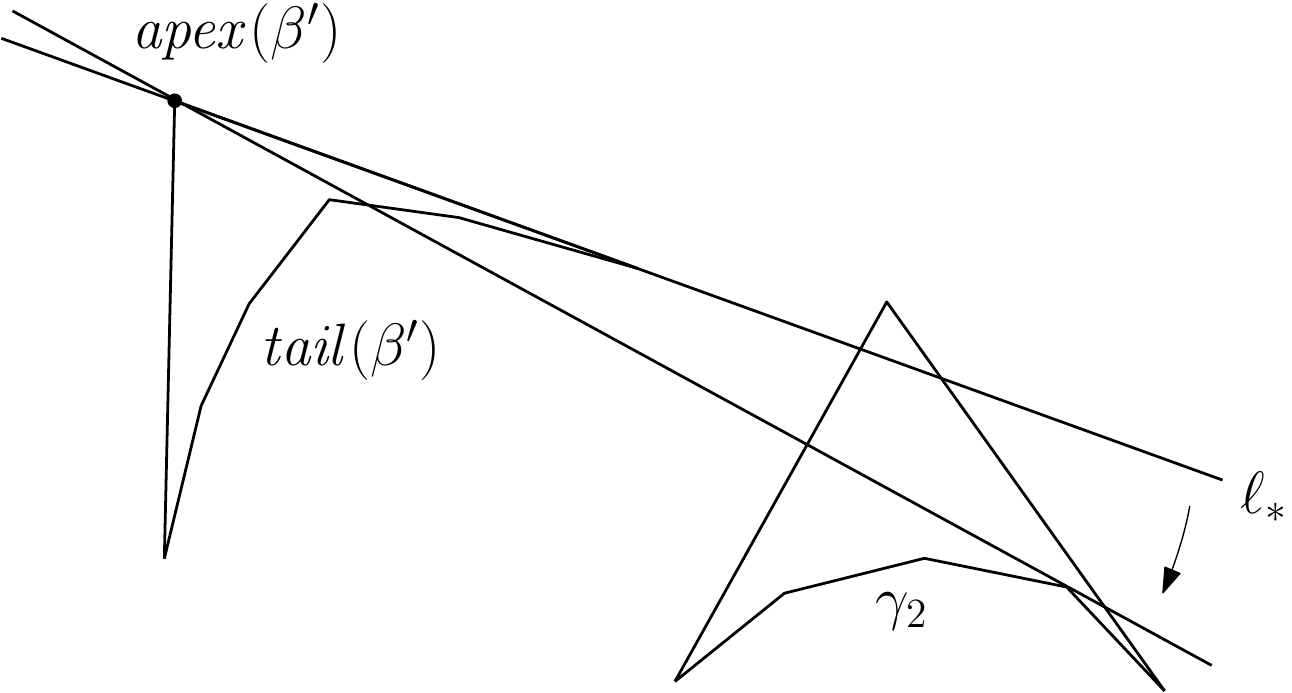} \\ \\
			(c) & \hspace{.15in} & (d) \\
		\end{tabular}
		\caption{(a) \& (b). Rotate $\ell$ around $\tail(\beta')$ so that $\ell$ moves closer to $\tail(\gamma_2)$.  (c) Rotation stops with a line $\ell_*$ tangent to $\tail(\gamma_2)$.  (d) Rotation stops with a line $\ell_*$ containing a head side of $\beta'$.  Then, rotate $\ell_*$ around $\apex(\beta')$ to obtain a tangent to $\tail(\gamma_2)$ that intersects the interior of $\gamma_2$ or contains a head side of $\gamma_2$.}
		\label{fg:case1-2}	
	\end{figure}
	
	Case~1.2.1: $\ell$ is tangent to $\tail(\beta')$ at $x_1$.   Figures~\ref{fg:case1-2}(a) and~(b) show two possible configurations.  We rotate $\ell$ around $\tail(\beta')$ so that $\ell$ moves closer to $\tail(\gamma_2)$ as the rotation proceeds.  Such a rotation is possible because $\beta'$ and $\beta_2 \supseteq \gamma_2$ are disjoint.  The rotated copy of $\ell$ remains tangent to $\tail(\beta')$ throughout the rotation.  We stop the rotation of $\ell$ around $\tail(\beta')$ when the rotated copy $\ell_*$ of $\ell$ is tangent to $\tail(\gamma_2)$ or contains a head side of $\beta'$, whichever happens first.  
	
	If $\ell_*$ is tangent to $\tail(\gamma_2)$, then $\ell_*$ is also tangent to $\tail(\beta_2)$ as $\ell_*$ either intersects the interior of $\gamma_2$ or contains a head side of $\gamma_2$.  Thus, $\ell_* \in \mathit{tangents}[i_2,j_2]$.  But then $\ell_*$ should have caused $\gamma_2$ to be trimmed further by an application of Case~1 or~2 on page~\arabic{pagecount}, a contradiction.  Figure~\ref{fg:case1-2}(c) gives an illustration.
	
	Suppose that $\ell_*$ contains a head side of $\beta'$.  Then, $\ell_*$ passes through $\apex(\beta')$ which is a vertex in $\mathit{pivots}[i_2,j_2] \setminus \tail(\beta_2)$.  We can turn $\ell_*$ around $\apex(\beta')$ to obtain a tangent to $\tail(\gamma_2)$ that intersects the interior of $\gamma_2$ or contains a head side of $\gamma_2$.  This tangent belongs to $\mathit{tangents}[i_2,j_2]$.  This leads to the contradiction that $\gamma_2$ should have been trimmed further.  Figure~\ref{fg:case1-2}(d) gives an illustration.
	
	\vspace{4pt}
	
	\newpage
	
	Case~1.2.2: $\ell$ is not tangent to $\tail(\beta')$ at $x_1$.  Figures~\ref{fg:case1-2-2}(a) and (b) show the two possible configurations depending on whether $x_1x_2$ lies on the same side of $\tail(\beta')$ locally at $x_1$ as $\beta'$.  This case distinction makes sense because the tail endpoints of $\beta'$ belong to $\mathit{pivots}[i_2,j_2]$ by definition, implying that $x_1$ is an interior vertex of $\tail(\beta')$.  Note that $x_1 \in \partial r_{i_2-1}$.  
	
	\begin{figure}
		\centering
		\begin{tabular}{ccc}
			\includegraphics[scale=0.55]{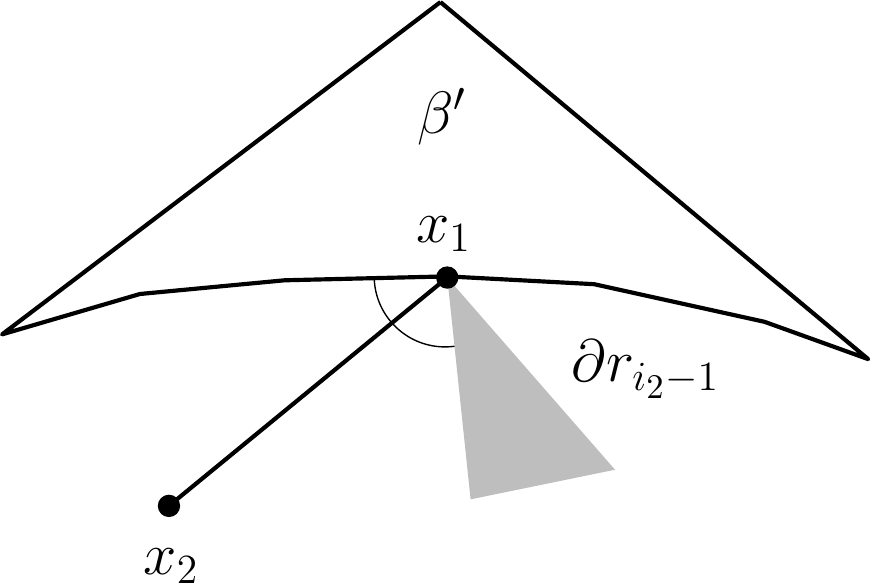} & & 
			\includegraphics[scale=0.55]{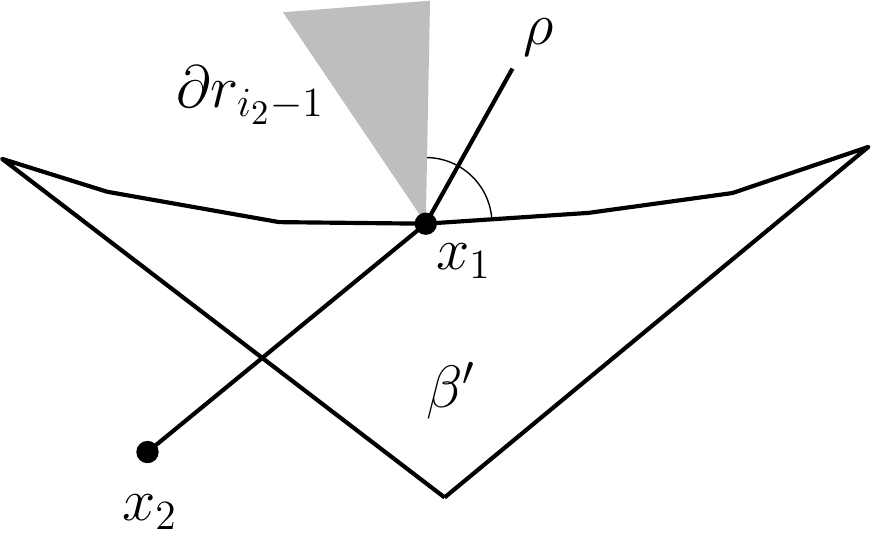} \\ \\
			(a) & \hspace{.15in} & (b) \\ \\
		\end{tabular}
		\caption{The shaded region is part of the exterior of $r_{i_2-1}$.  The angle at $x_1$ stabbed by $x_1x_2$ or $\rho$ is less than $\pi$.  Hence, $x_1$ is a convex vertex of a geodesic triangle selected in buidling $\mathit{seed\_fins}[i_2,j_2]$.  This makes $x_1$ a tail endpoint in $\mathit{seed\_fins}[i_2,j_2]$.}
		\label{fg:case1-2-2}	
	\end{figure}

	In Figure~\ref{fg:case1-2-2}(a), $x_1x_2$ and $\partial r_{i_2-1}$ lie locally at $x_1$ on the same side of $\tail(\beta')$ because $\beta'$ is a subset of $r_{i_2-1}$.  As $\tail(\beta')$ is reflex, the angle inside $\beta'$ at $x_1$ is at least $\pi$.  It follows that the angle at $x_1$ stabbed by $x_1x_2$ is less than $\pi$.  This angle at $x_1$ or a part of it must belong to some geodesic triangle $\tau$ selected in building $r_{i_2}$ because $x_1x_2 \subseteq r_i \subseteq r_{i_2}$.  Observe that $x_1$ is a convex vertex of $\tau$.  
	By Lemma~\ref{lem:seed-fins},
	the convex vertices of $\tau$ are tail endpoints of seed fins.  As a result, $x_1 \in \mathit{pivots}[i_2,j_2] \setminus \tail(\beta_2)$.  (We have shown that $x_1 \not\in \tail(\beta_2)$.)  But then we can turn $\ell$ around $x_1$ to obtain a tangent to $\tail(\gamma_2)$ that trims $\gamma_2$ as in Case~1.1, a contradiction.
	
	In Figure~\ref{fg:case1-2-2}(b), $x_1x_2$ and $\partial r_{i_2-1}$ lie locally at $x_1$ on opposite sides of $\tail(\beta')$.  Recall that $x_1x_2$ is an edge on the shortest path $\rho$ in $r_i$ that connects two vertices of $S_{1,1}$ in $\partial r_i$.  When $\rho$ goes from $x_2$ to $x_1$ and then onwards, $\rho$ cannot bounce back from $\tail(\beta')$ at $x_1$.  Otherwise, $\rho$ can be shortened locally at $x_1$, contradicting that $\rho$ is a shortest path.  Thus, $\rho$ must cross $\tail(\beta')$ and stab an angle at $x_1$ that is less than $\pi$.  This angle at $x_1$ or a part of it must belong to some geodesic triangle $\tau$ selected in building $r_{i_2}$ because $\rho \subseteq r_i \subseteq r_{i_2}$.  So $x_1$ is a convex vertex of $\tau$.  Then, we can analyze as in the previous paragraph to obtain a contradiction.

	\vspace{4pt}
	
	Case 2: $i_1 = i_2$.  It means that $\beta_1$, $\beta_2$, $\gamma_1$, and $\gamma_2$ are created at the same level $i_2$.  So $\gamma_1$ and $\gamma_2$ are disjoint.  Recall that $\gamma_s$ lies on one side of $\ell$ for $s \in [1,2]$.  
	
	\vspace{4pt}
	
	Case 2.1: $\gamma_1$ and $\gamma_2$ lie on the same side of $\ell$.  We have shown earlier that $x_1 \not\in \tail(\beta_2)$.  Then, by Lemma~\ref{lem:suboutreg}(iii), $\ell$ cannot contain a head side of $\gamma_1$ and another head side of $\gamma_2$.  Without loss of generality, assume that $\ell$ does not contain a head size of $\gamma_2$.  Figure~\ref{fg:case2-1} shows two possible configurations.  We translate $\ell$ towards $\tail(\gamma_1)$ and $\tail(\gamma_2)$ until the translated line is tangent to $\tail(\gamma_1)$ or $\tail(\gamma_2)$ whichever happens first.  Refer to Figure~\ref{fg:case2}(a).  Without loss of generality, assume that we obtain a translated copy $\ell_*$ of $\ell$ that is tangent to $\tail(\gamma_1)$.  Observe that $\ell_*$ is also tangent to $\tail(\beta_1)$.   Then, we can rotate $\ell_*$ around $\tail(\beta_1)$ as illustrated in Figures~\ref{fg:case1-2}(c)~and~(d).  (Substitute $\beta'$ with $\beta_1$ in the figures.)  This gives the contradiction that $\gamma_2$ should have been trimmed further.

	\begin{figure}
		\centering
		\begin{tabular}{ccc}
			\includegraphics[scale=0.5]{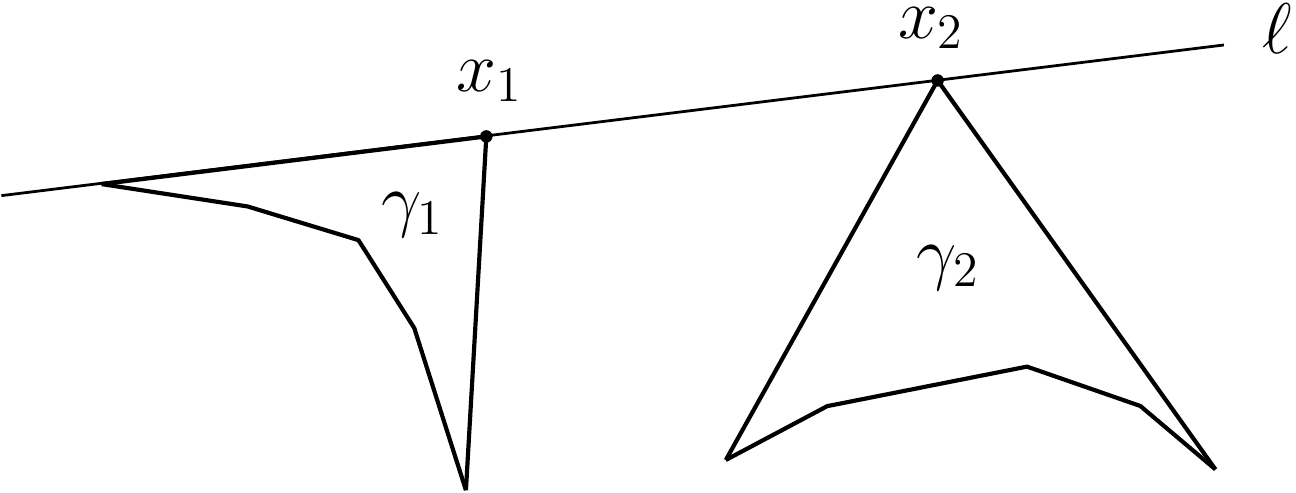} & & 
			\includegraphics[scale=0.5]{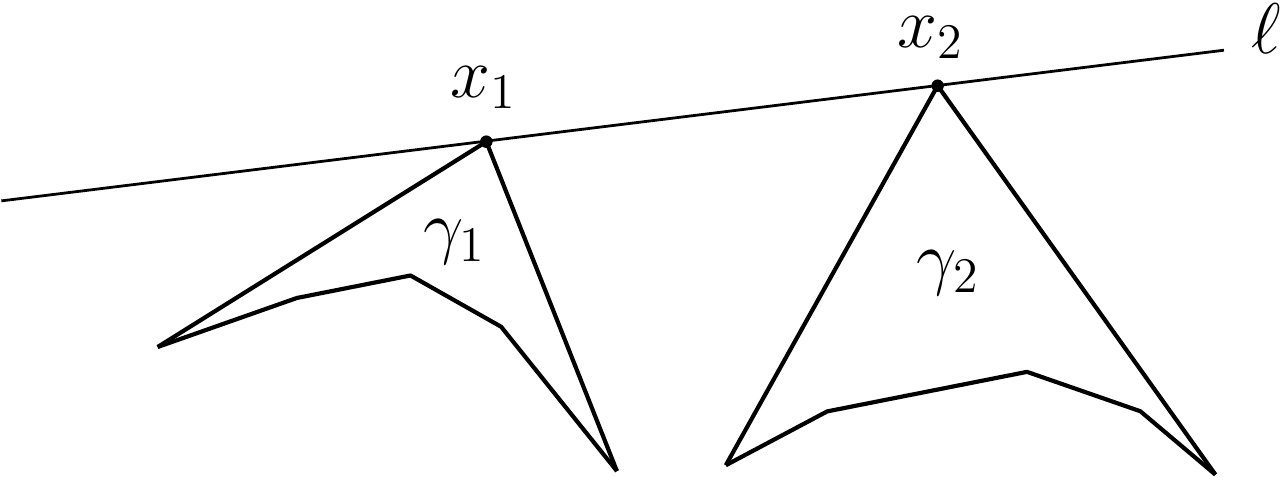} \\ \\
			(a) & \hspace{.2in} & (b) \\
		\end{tabular}
		\caption{The line $\ell$ through $x_1$ and $x_2$ does not contain any head side of $\gamma_2$, and $\gamma_1$ and $\gamma_2$ lie on the same side of $\ell$.}
		\label{fg:case2-1}	
	\end{figure}
	
	\begin{figure}
		\centering
		\begin{tabular}{ccc}
			\includegraphics[scale=0.55]{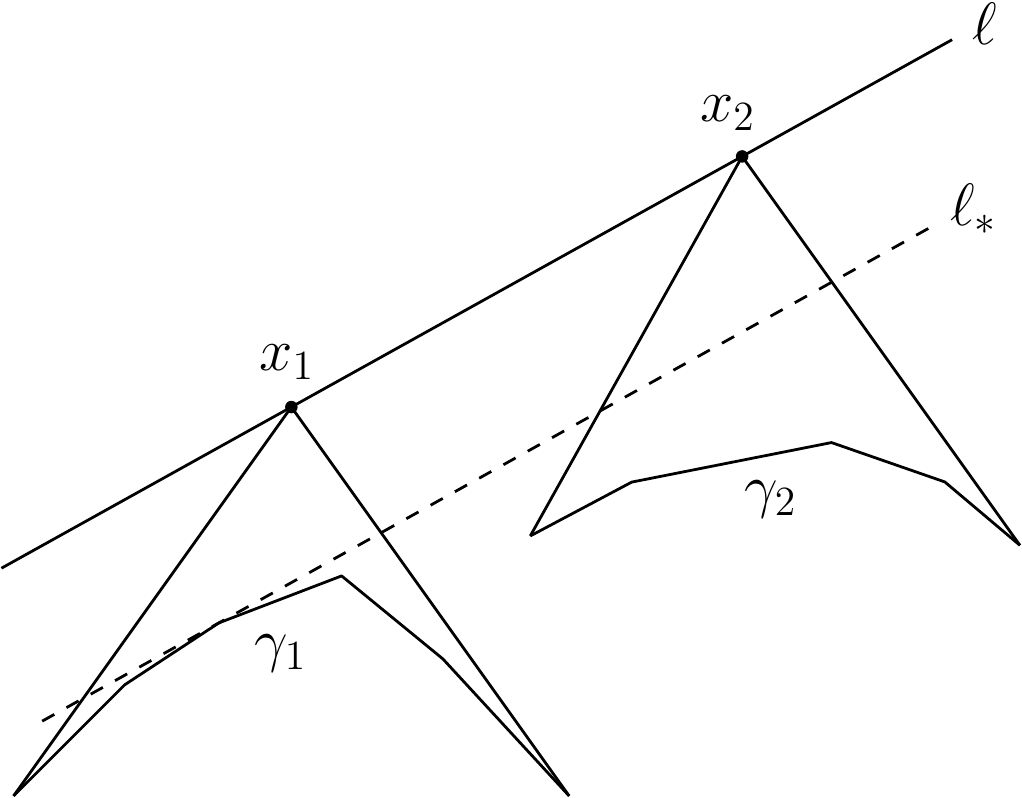} & & 
			\includegraphics[scale=0.6]{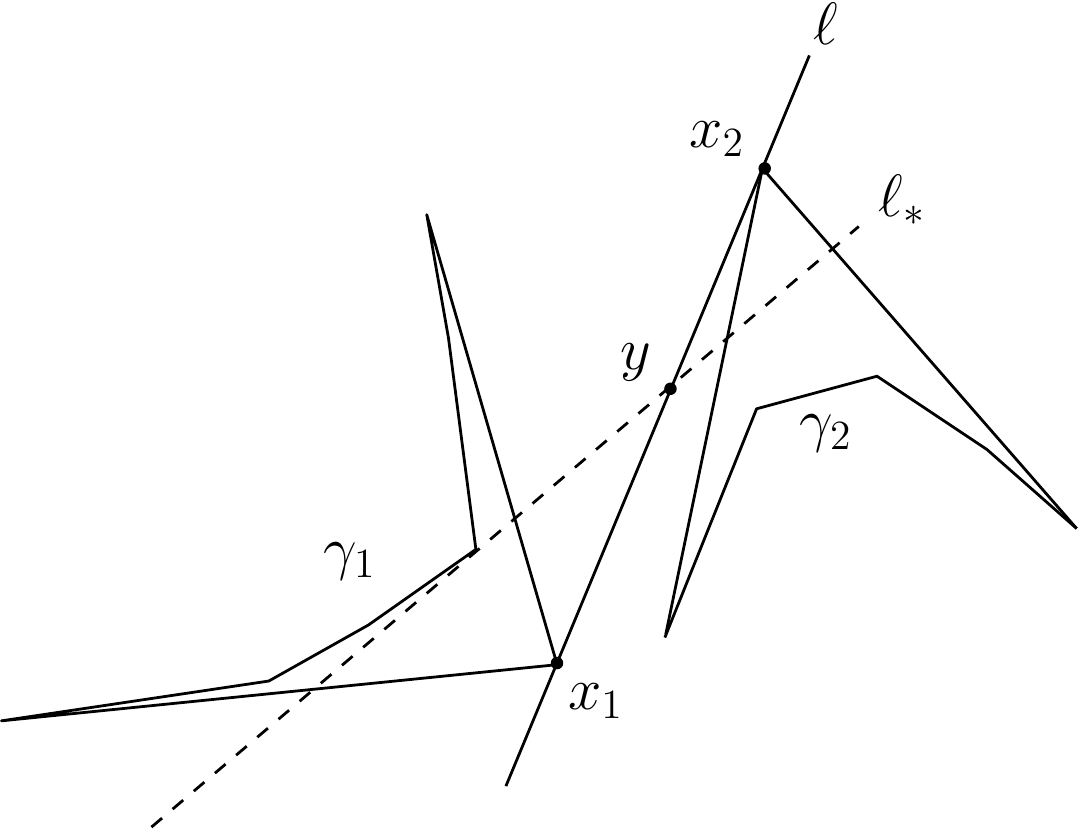} \\ \\
			(a) & \hspace{.15in} & (b) \\ \\
		\end{tabular}
		\caption{(a) Translate $\ell$ to obtain a tangent $\ell_*$ to $\tail(\gamma_1)$.  (b) Rotate $\ell$ around $y$ to obtain a tangent $\ell_*$ to $\tail(\gamma_1)$.}
		\label{fg:case2}	
	\end{figure}
	
	\vspace{4pt}
	
	Case 2.2: $\gamma_1$ and $\gamma_2$ lie on opposite sides of $\ell$.  Since $\gamma_1$ and $\gamma_2$ are disjoint, we can rotate $\ell$ around a point $y \in \ell$ outside $\gamma_1$ and $\gamma_2$ until the rotated copy $\ell_*$ of $\ell$ is tangent to $\tail(\gamma_1)$ and $\tail(\gamma_2)$, whichever happens first, say $\tail(\gamma_1)$.  Refer to Figure~\ref{fg:case2}(b).  Then, we can rotate $\ell_*$ around $\tail(\beta_1)$ in a way similar to those illustrated in Figures~\ref{fg:case1-2}(c)~and~(d).  (Substitute $\beta'$ with $\beta_1$ in the figures.)  This gives the contradiction that $\gamma_2$ should have been trimmed further.
\end{proof}

It is easier to show that Q2 and Q3 are preserved.

\begin{lemma}\label{lem:R1}
	$R_{i,j}$ satisfies Q2 and Q3 in Table~\ref{tb:Q}.
\end{lemma}
\begin{proof}
	By construction, for every pair of regions in $R_{i,j}$ that lie in the same region in $S_{1,1}$, they cannot share any edge because they would have been merged otherwise.  It follows that Q2 is satisfied.  

	Consider Q3.  Assume to the contrary that some region $r'$ in $R_{i,j}$ contains a hole.  Let $r$ be the region in $S_{1,1}$ that contains $r'$.  So the hole in question lies inside $r$.  By Lemma~\ref{lem:pre-edge}, every edge of $R_{i,j}$ is incident to some vertex of $S_{1,1}$.  It means that the hole boundary of $r'$ contains a vertex of $S_{1,1}$ that lies strictly inside $r$, a contradiction to the fact that $r$ is a simple polygon.
\end{proof}

\cancel{

Before analyzing $g(n,i)$ in C4 in the next section, we prove two more properties of $R_{i,j}$ and $S{i,j}$ that will be useful later.

\begin{lemma}
\label{lem:tech1}
	Suppose that $R_{k,l}$ is constructed from $R_{k-1,l'}$.  Every region in $R_{k,l}$ is a subset of some region in $R_{k-1,l'}$.
\end{lemma}
\begin{proof}
	The correctness of (i) follows directly from our construction method.  Consider (ii).  For any $k$, every vertex of $S_{k,*}$ is inherited from $S_{k-1,*}$ or created during the construction of $S_{k,*}$.  Therefore, if $v$ is a vertex of $S_{k',*}$ but $v$ is not a vertex of $S_{k,*}$ for some $k > k'$, then $v$ cannnot be inherited to be a vertex of $S_{k'',*}$ for any $k'' > k$.
\end{proof}

}

\subsection{Meeting the conditions in Table~\ref{tb:C}}
\label{sec:conditions}

We show that C1--C5 in Table~\ref{tb:C} are satisfied by the newly constructed structures at level $i$.  Then, C1--C5 hold inductively, and we can apply Theorem~\ref{thm:iacono}.


Condition C1 is met by the definitions of $c_1 = 31$ and $f(k) = (\log_2 k)^{31}$.  Since $c_2 = 25$, the complexity of $R_{i,j}$ is $O((\log n_{i-1})^{30})$ by Lemma~\ref{lem:Delta}.  This is $O(n_{i})$, where $n_{i} = f(n_{i-1})/\log_2 n_{i-1}$.  So C2 is satisfied.  By Lemma~\ref{lem:Delta}, the processing time for constructing structures of any version at the $i$-th level is $O((\log n_{i-1})^{31}) = O(n_{i} \log n_{i-1})$.  So C3 is met.  Since $g(n,i) = i\log_2^3 n_i$, condition C5 is satisfied because $f(n) = (\log_2 n)^{31} > g(n,1)(\log_2 n)^{27} = (\log_2 n)^{30}$.  

The final task is to prove that C4 in Table~\ref{tb:C} holds for $g(n,i) =  i \log_2^3 n_i$.  We introduce the concept of a segment \emph{slashing} a trimmed subfin: a segment $e$ \emph{slashes} a trimmed subfin $\gamma$ if $e$ intersects both head sides of $\gamma$ and separates $\apex(\gamma)$ from $\tail(\gamma)$.  Figure~\ref{fg:slash} shows some positive and negative examples.  We show that a segment inside a region of $S_{1,1}$ slashes at most one trimmed subfin in any version of structures at any level.

\begin{figure}
	\centerline{\includegraphics[scale=0.6]{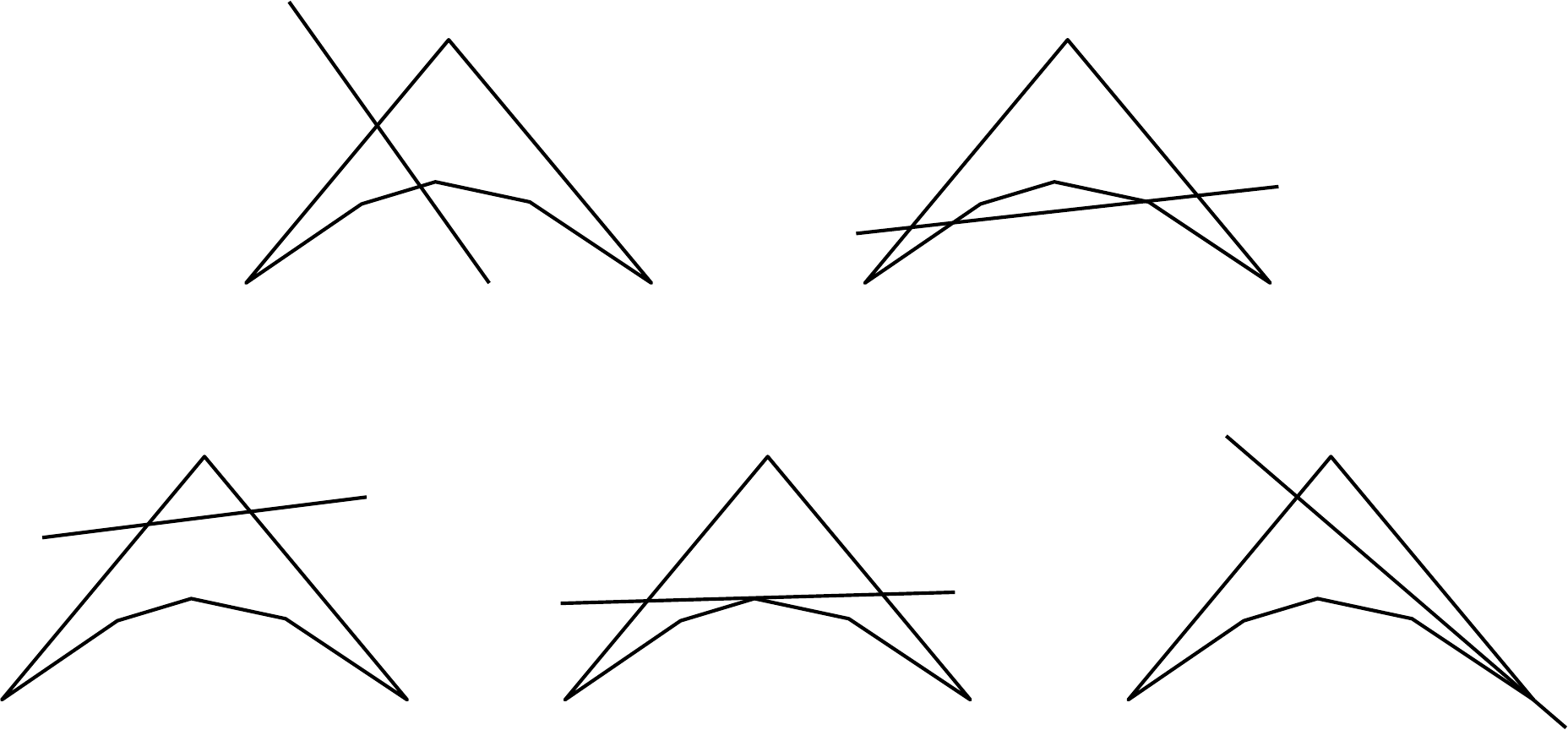}}
	\caption{The upper row shows some trimmed subfins that are not slashed by the segments shown.  The lower row shows some trimmed subfins that are slashed by the segments shown.  In the rightmost figure in the lower row, the segment does not cross the tail because the right head side is tangent to the tail.}
	\label{fg:slash}
\end{figure}

\begin{lemma}
	\label{lem:slash}
	For any level $k$ and any version index $l$, a line segment that lies inside a region of $S_{1,1}$ slashes at most one trimmed subfin in $\bigcup_{\beta \in \mathit{subfins}[k,l]} \mathit{trimmed\_subfins}(\beta)$.
\end{lemma}

\begin{proof}
	Suppose for the sake of contradiction that some segment $e$ slashes two trimmed subfins $\gamma_1$ and $\gamma_2$, where $\gamma_s \in \mathit{trimmed\_subfins}(\beta_s)$ for some $\beta_s \in \mathit{subfins}[k,l]$ for $s \in [1,2]$.  It must be the case that $\beta_1 \not= \beta_2$ because our construction method guarantees that no segment can slash two trimmed subfins in $\mathit{trimmed\_subfins}(\beta)$ for any subfin $\beta$. 
	
	Let $\ell$ denote the support line of $e$.  Since $e$ slashes $\gamma_1$ and $\gamma_2$, $\tail(\gamma_1)$ and $\tail(\gamma_2)$ lie on the same side or opposite sides of $\ell$.  The trimmed subfins $\gamma_1$ and $\gamma_2$ are disjoint because they are created at the same level.  The possible configurations are similar to those in the proof of Lemma~\ref{lem:edgeTouchS}.  
	
	Suppose that $\tail(\gamma_1)$ and $\tail(\gamma_2)$ lie on the same side of $\ell$.  We translate $\ell$ towards $\tail(\gamma_1)$ and $\tail(\gamma_2)$ until the translated line is tangent to $\tail(\gamma_1)$ or $\tail(\gamma_2)$ whichever happens first.  Refer to Figure~\ref{fg:slash-2}(a).  
	We can argue as in Case~2.1 in the proof of Lemma~\ref{lem:edgeTouchS}. This gives the contradiction that $\gamma_2$ should have been trimmed further.

	\begin{figure}
	\centering
	\begin{tabular}{ccc}
		\includegraphics[scale=0.6]{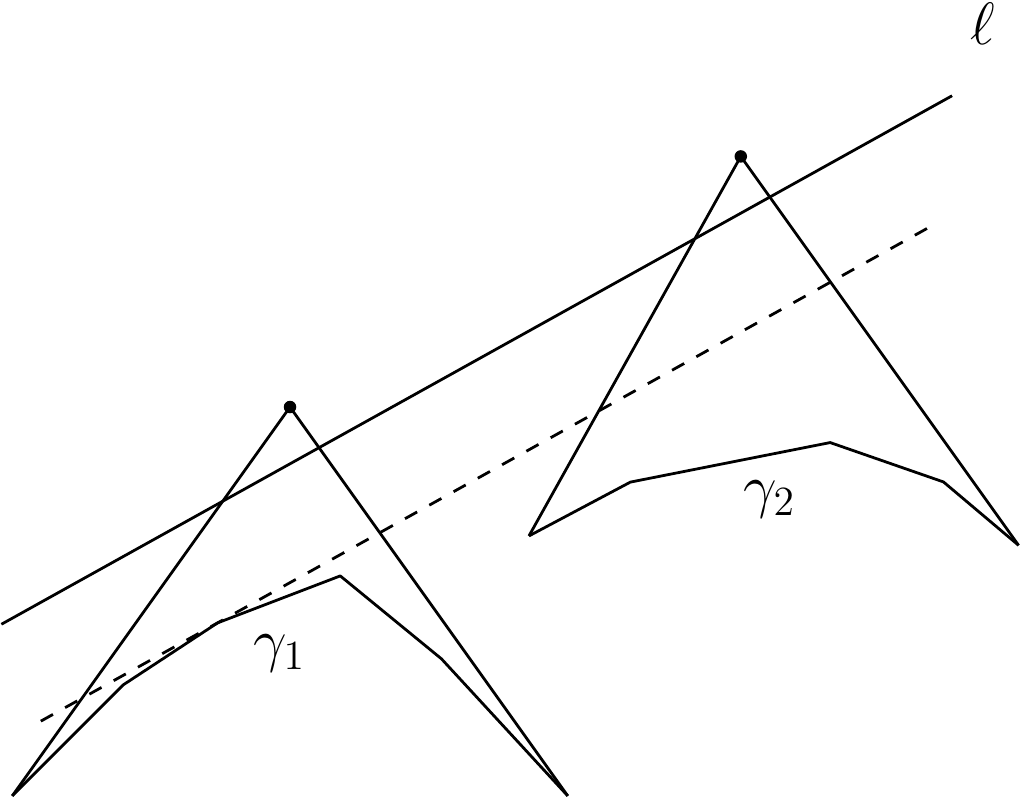} & & 
		\includegraphics[scale=0.6]{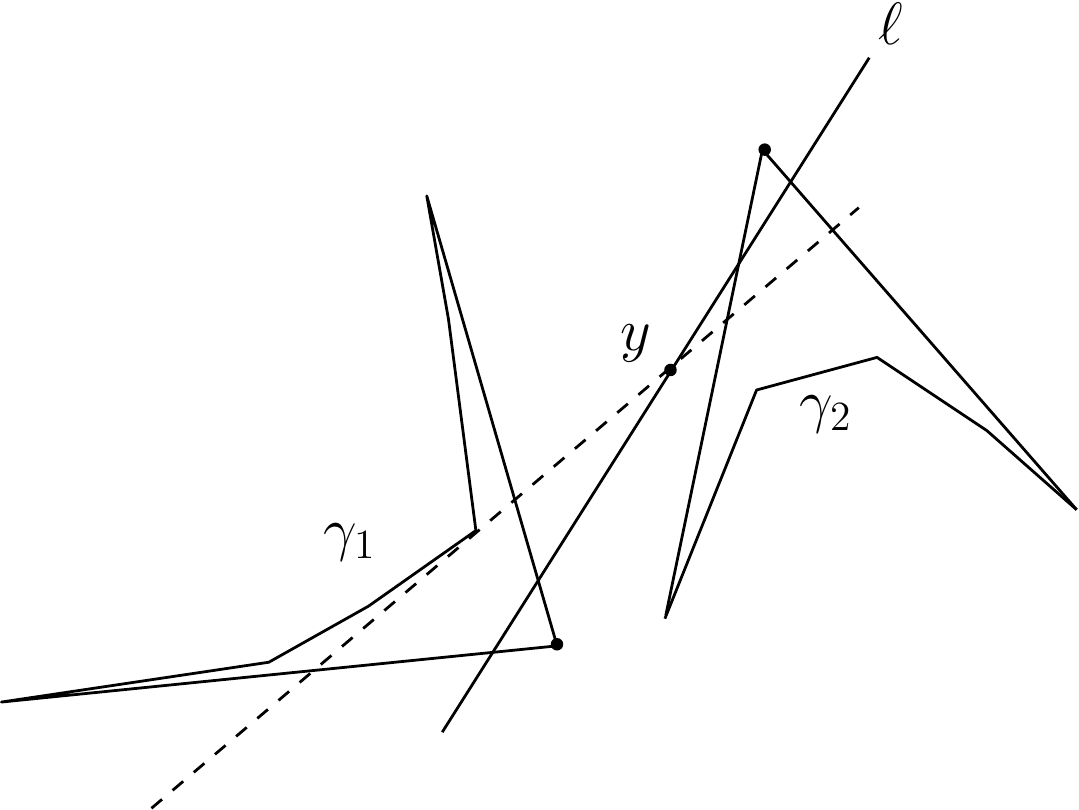} \\ \\
		(a) & \hspace{.15in} & (b) \\
	\end{tabular}
		\caption{Translate or rotate $\ell$ to obtain a tangent to $\tail(\gamma_1)$.}
		\label{fg:slash-2}	
	\end{figure}
	
	Suppose that $\tail(\gamma_1)$ and $\tail(\gamma_2)$ lie on opposite sides of $\ell$.  Since $\gamma_1$ and $\gamma_2$ are disjoint, we can rotate $\ell$ around a point $y \in \ell$ outside $\gamma_1$ and $\gamma_2$ until the rotated copy of $\ell$ is tangent to $\tail(\gamma_1)$ and $\tail(\gamma_2)$, whichever happens first.  Refer to Figure~\ref{fg:slash-2}(b).  We can argue as in Case~2.2 in the proof of Lemma~\ref{lem:edgeTouchS}.  Again, this gives the contradiction that $\gamma_2$ should have been trimmed further.
\end{proof}

Consider a region $r$ in $R_{k,l}$ for some $k$ and $l$.  If a segment $e$ lies inside $r$, we show that the geodesic triangles in $G_r$ intersected by $e$ observe an ancestor-descendant relation in $G_r$.
	
\begin{lemma}
	\label{lem:walk}
	Let $e$ be a segment inside a region $r$ of some $R_{k,l}$.
	Let $(\tau_1, \cdots, \tau_q)$ be the sequence of geodesic triangles in the balanced geodesic triangulation of $r$ that intersect $e$.  Let the sequence be ordered from left to right along $e$.\footnote{It is possible that $e$ intersects a geodesic triangle at one of its convex vertices only.  In this case, we only consider such geodesic triangles that lie above $e$.  If $e$ intersects a common vertex $v$ of a subsequence $\tau_i,\cdots,\tau_j$, the ordering $\tau_i,\cdots,\tau_j$ corresponds to their clockwise order around $v$.  There is no loss of generality in considering only such geodesic triangles that lie above $e$ because we can rotate the plane by $180^{\circ}$.}  There exists $p \in [0,q]$ such that for all $a \in [1,p-1]$, $\tau_a$ is a descendant of $\tau_{a+1}$ in $G_r$, and for all $a \in [p+1,q-1]$, $\tau_a$ is an ancestor of $\tau_{a+1}$ in $G_r$.
\end{lemma}
\begin{proof}
	In the balanced geodesic triangulation of $r$, each geodesic triangle $\tau$ is contained in a kite which we denote by $\mathit{kite}(\tau)$.  Let $\mathit{kite}(\tau)$ be the kite nearest to the root of $G_r$ among all kites intersected by $e$.  The segment $e$ may or may not intersect $\tau$ because $e$ may intersect some fringe of $\mathit{kite}(\tau)$ only.  There are two cases depending on whether $\tau$ belongs to $(\tau_1,\cdots,\tau_q)$.
	
	Case 1: $\tau = \tau_p$ for some $p \in [1,q]$.  We claim that $(\tau_1,\cdots,\tau_p)$ and $(\tau_{p+1},\cdots,\tau_q)$ satisfy the lemma.  We actually prove that for all $a \in [p,q-1]$, $\tau_a$ is an ancestor of $\tau_{a+1}$ in $G_r$.  A symmetric arguemnt then shows that for all $a \in [1,p-1]$, $\tau_a$ is a descendent of $\tau_{a+1}$ in $G_r$.

	The three shortest paths on the boundary of $\mathit{kite}(\tau_p)$ divide $r$ into four parts, including $\tau$ and three other polygonal parts. Walk along $e$ from $\tau_p$ to the right.  If we never exit $\tau_p$, then $p = q$ and there is nothing to prove.  Otherwise, we must cross a reflex boundary chain of $\tau_p$, which is a subsequence of a boundary shortest path $\xi$ of $\mathit{kite}(\tau_p)$.   So we cross $\xi$ too.  Consider the geodesic triangle contained in the kite that shares $\xi$ with $\mathit{kite}(\tau_p)$.  It must be a descendant of $\tau_p$; otherwise, $\mathit{kite}(\tau_p)$ would not be the kite nearest to the root of $G_r$ among those intersected by $e$.  It further implies that the polygonal part separated from $\tau_p$ by $\xi$ bounds geodesic triangles in $G_r$ that are descendants of $\tau_p$.  So $\tau_{p+1}$ is a descendant of $\tau_p$.
	
	There are three shortest paths on the boundary of $\mathit{kite}(\tau_{p+1})$.  Let $\xi'$ be the one that overlaps with $\xi$.  Therefore, we cross $\xi'$ when entering $\tau_{p+1}$ from $\tau_p$.  We also cross a reflex boundary chain $\zeta$ of $\tau_{p+1}$ where $\zeta \subseteq \xi'$.  One boundary shortest path of $\mathit{kite}(\tau_{p+1})$ borders the parent of $\tau_{p+1}$ in $G_r$ (which may or may not be $\tau_p$).  The other two boundary shortest paths of $\mathit{kite}(\tau_{p+1})$ border the two children of $\tau_{p+1}$ and hence other descendants of $\tau_{p+1}$ too.  Since $\xi'$ also borders $\tau_p$ which is an ancestor of $\tau_{p+1}$, $\xi'$ must border the parent of $\tau_{p+1}$.  When we walk along $e$ to exit $\tau_{p+1}$, we cannot cross $\zeta$ again because $\zeta$ is reflex.  So we must exit $\tau_{p+1}$ at another boundary reflex chain, which must be contained in a boundary shortest path, say $\xi''$, of $\mathit{kite}(\tau_{p+1})$ different from $\xi'$.  Again, $\xi''$ separates $\tau_{p+1}$ from a polygonal part of $r$ such that all geodesic triangles in this part are descendants of $\tau_{p+1}$.  So $\tau_{p+2}$ is a descendant of $\tau_{p+1}$.  Repeating the argument shows that for all $a \in [p,q-1]$, $\tau_a$ is an ancestor of $\tau_{a+1}$.
	
	Case 2: $\tau \not\in (\tau_1,\cdots,\tau_q)$.  There exists $p \in [1,q-1]$ such that a boundary shortest path $\xi$ of $\mathit{kite}(\tau)$ contains the edge $\tau_p \cap \tau_{p+1}$, and we cross the edge $\tau_p \cap \tau_{p+1}$ when walking along $e$ from $\tau_p$ to $\tau_{p+1}$. 
	
	The three boundary shortest paths of $\mathit{kite}(\tau)$ divide $r$ into $\tau$ and three polygonal parts.  If $\tau$ is the root of $G_r$, then all three polygonal parts contain geodesic triangles that are descendants of $\tau$.  If $\tau$ is not the root of $G_r$, then two of the polygonal parts contain geodesic triangles that are descendants of $\tau$, and the remaining polygonal part contains the kite $\mathit{kite}(\tau')$ such that $\tau'$ is the parent of $\tau$.  Note that $\mathit{kite}(\tau)$ and $\mathit{kite}(\tau')$ share a common boundary shortest path.  When entering $\tau_{p+1}$ from $\tau_p$, $e$ must enter a polygonal part that contains geodesic triangles descending from $\tau$.  Otherwise, $e$ would intersect $\mathit{kite}(\tau')$, a contradiction to the assumption that $\mathit{kite}(\tau)$ is nearest to the root of $G_r$ among all kites intersected by $e$.  
	
	We walk from $\tau_{p+1}$ to $\tau_{p+2}$ and so on.  The same analysis in case~1 shows that for $a \in [p+1,q-1]$, $\tau_a$ is an ancestor of $\tau_{a+1}$.  A symmetric analysis shows that for $a \in [1,p-1]$, $\tau_a$ is a descendent of $\tau_{a+1}$.
\end{proof}

We are ready to bound the number of intersections between $R_{k,l}$ and any segment that lies strictly inside a region of $S_{1,1}$.  As a shorthand, given a polyline $\xi$ and a set $R$ of regions, we use $\xi \sqcap R$ to denote the set of points in $\xi \cap r$ for all $r \in R$.  Note that $\xi \sqcap R$ may have several connected components, each being a polyline.

\begin{lemma}
	\label{lem:g}
	Let $e$ be a segment inside a region in $S_{1,1}$.
	For any $R_{k,l}$, there are at most $k$ subsegments in $e \sqcap R_{k,l}$.
\end{lemma}
\begin{proof}
	We prove the lemma by induction on the level.  The base case of level 1 is trivial because $R_{1,1}$ is the set of bounded regions in $S_{1,1}$ and $e$ lies inside a region of $S_{1,1}$.  Assume that the lemma is true for level $k-1$ and any version index.  
	
	Consider $R_{k,l}$.  Suppose that $R_{k,l}$ is constructed from $R_{k-1,l'}$.  By induction, $e \sqcap R_{k-1,l'}$ consists of at most $k-1$ subsegments.  Denote them by $I_1, \cdots I_q$, where $q \leq k-1$.  Since each region in $R_{k,l}$ is a subset of a region in $R_{k-1,l'}$, $e \sqcap R_{k,l}$ consists of $I_p \sqcap R_{k,l}$ for all $p \in [1,q]$.  
	
	Take any subsegment $I_p$ such that $I_p \sqcap R_{k,l} \not= \emptyset$.  We claim that if $I_p \sqcap R_{k,l}$ intersects any edge $e'$ of the balanced geodesic triangulations $\tilde{\Delta}_{k-1,l'}$ such that the intersection lies in the interior of $I_p$ and $I_p \sqcap R_{k,l}$ crosses the support line of $e'$, then $e'$ lies in the interior of $R_{k,l}$.  Since $e'$ is intersected in the interior of $I_p$, the intersection is in the interior of $R_{k-1,l'}$.  Therefore, $e'$ is the common edge $\tau_1 \cap \tau_2$ of two geodesic triangles $\tau_1$ and $\tau_2$ in $\tilde{\Delta}_{k-1,l'}$ that lie in the same region in $S_{1,1}$.  Let $I$ be the subsegment in $I_p \sqcap R_{k,l}$ that contains the intersection with $e'$.  Since both endpoints of $I$ are in $R_{k,l}$, they must be contained in two geodesic triangles in $\tilde{\Delta}_{k-1,l'}$ that are selected in forming $R_{k,l}$.  By applying Lemma~\ref{lem:walk} to $I$ and our method of selecting geodesic triangles in $\tilde{\Delta}_{k-1,l'}$ in forming $R_{k,l}$, we conclude that the two geodesic triangles $\tau_1$ and $\tau_2$ must be selected, and for $i \in [1,2]$, the triangular regions in $\tau_i$ incident to $\tau_1 \cap \tau_2$ is part of $R_{k,l}$.  Therefore, by Q2 in Table~\ref{tb:Q}, the edge $e' = \tau_1 \cap \tau_2$ lies in the interior of $R_{k,l}$.  This proves our claim.
	
	Again, take any subsegment $I_p$ such that $I_p \sqcap R_{k,l} \not= \emptyset$.  Recall that only the tail edges of augmented seed fins can contribute to the boundary of $R_{k,l}$.  Lemma~\ref{lem:augmented} implies that a boundary edge of $R_{k,l}$ is either a tail edge in $\hat{\Delta}_{i-1,*}$ (and hence an edge in $\tilde{\Delta}_{i-1,*}$) or a head side of a trimmed subfin. Start from the leftmost point in $I_p \sqcap R_{k,l}$ and walk to the right along $I_p$.  During the walk, if we exit $R_{k,l}$ in the interior of $I_p$, we must exit at an edge in $\tilde{\Delta}_{i-1,*}$ or a head side of a trimmed subfin.  The first case is impossible because such an edge in $\tilde{\Delta}_{i-1,*}$ lies in the interior of $R_{k,l}$ by our claim.  Therefore, we can only exit $R_{k,l}$ in the interior of $I_p$ by entering some trimmed subfin $\gamma$ through a head side of $\gamma$.  If we do not exit $\gamma$, then $I_p \sqcap R_{k,l}$ is a single subsegment.  Suppose that we exit $\gamma$ at a point $x$ in its boundary.  Consider the possibility of $x$ lying on $\tail(\gamma)$.  Recall that the edges of $\tail(\gamma)$ are edges in $\tilde{\Delta}_{i-1,*}$.  In this case, our claim implies that $x$ cannot be in the interior of $I_p$.  On the other hand, if $x$ is an endpoint of $I_p$, then $x$ lies on the boundary of $R_{k-1,l'}$.  Since $\gamma \subseteq R_{k-1,l'}$, we are also exiting $R_{k-1,l'}$ at $x$.   Then, $x \not\in R_{k,l}$ because $R_{k,l} \subseteq R_{k-1,l'}$ and we arrive at $x$ from the outside of $R_{k,l}$.  As a result, $I_p \sqcap R_{k,l}$ is a single subsegment.  
	
	The remaining possibility is that $x$ lies on the head side of $\gamma$ different from the one containing the entry point into $\gamma$.  By our claim, we do not cross $\tail(\gamma)$ before $x$, that is, $e$ slashes~$\gamma$.  By Lemma~\ref{lem:slash}, $e$ slashes at most one trimmed subfin in $\bigcup_{\beta \in \mathit{subfins}[k,l]} \mathit{trimmed\_subfins}(\beta)$.  It means that for every $p' \in [1,q] \setminus \{p\}$, $I_{p'} \sqcap R_{k,l}$ must be empty or a single subsegment.  It also implies that $I_p \sqcap R_{k,l}$ consists of only two subsegments.  As a result, the total number of subsegments in $e \sqcap R_{k,l}$ is at most $k$.
\end{proof}

We can now show that C4 in Table~\ref{tb:C} is satisfied.

\begin{lemma}
	$\Delta_{i,j}$ satisfies C4 in Table~\ref{tb:C}.
\end{lemma}
\begin{proof}
	Let $t$ be any triangle that lies inisde a region of $S$.  In the case that $t$ lies in the exterior region of $S$, if necessary, we can split $t$ into $O(1)$ triangles such that each lies in a region in $S_{1,1}$. Therefore, it suffices to prove that if $t$ lies inside a region of $S_{1,1}$, then $t$ intersects $O(g(n,i))$ triangles in $\Delta_{i,j}$ that are subsets of $R_{i,j}$.  
	
	There is no vertex of $S_{1,1}$ in the interior of $t$.  
	%
	By Lemma~\ref{lem:g}, there are $O(i)$ subsegments in $\partial t \sqcap R_{i,j}$, where $\partial t$ denotes the boundary of $t$.  The triangle $t$ cannot strictly enclose any geodesic triangle in $\tilde{\Delta}_{i,j}$; otherwise, a vertex of $S_{1,1}$ would lie in the interior of $t$ by Q1 in Table~\ref{tb:Q}.  It follows that $t$ intersects a geodesic triangle  $\tau$ in $\tilde{\Delta}_{i,j}$ if and only if $\tau$ intersects a subsegment in $\partial t \sqcap R_{i,j}$.  Each subsegment in $\partial t \sqcap R_{i,j}$ intersects $O(\log n_i)$ geodesic triangles.  
	It follows that $t$ intersects $O(i\log n_i)$ geodesic triangles.  In a geodesic triangle $\tau$, there are $O(\log n_i)$ triangles in $\Delta_{i,j}$ that lie inside $\mathit{star}(\tau)$.  By Q1 in Table~\ref{tb:Q}, every edge in $\tilde{\Delta}_{i,j}$ is incident to a vertex of $S_{1,1}$.  Therefore, for each $\alpha \in \mathit{fins}(\tau)$, $t$ cannot contain two vertices of $\tail(\alpha)$; otherwise, $t$ would not be contained in a region in $S_{1,1}$.  Then, we can analyze as in the convex case that $t$ intersects $O(\log^2 n_i)$ triangles in $\Delta_{i,j}$ that lie inside $\alpha$.  As a result, $t$ intersects $O(i\log^3 n_i)$ triangles in $\Delta_{i,j}$ that are subsets of regions in $R_{i,j}$.
\end{proof}

Since C1--C5 in Table~\ref{tb:C} are satisfied, we can apply Theorem~\ref{thm:generic-3} to obtain the following result for connected subdivision.

\begin{theorem}
Let $S$ be a planar connected subdivision with $n$ vertices.  There is a
point-line comparison based data structure that uses $O(n)$ space and processes
any online sequence $\sigma$ of point location queries in $S$ in
$O(\mathrm{OPT} + n + \vert \sigma \vert \log (\log^* n))$ time, where {\rm OPT} is the minimum
time needed by any point location linear decision tree for $S$ to process
$\sigma$.  The time bound includes the $O(n)$ preprocessing time.
\end{theorem}
\begin{proof}
	By Lemma~\ref{lem:claim}(i), the index $m$ of the highest level in the hierarchy of point location structures is at most $\log_2^* n + 1$.
	
	Suppose that $\log_2 n_m < \log_2^* n + 1$, then $g(n,m) = m \log_2^3 n_m = O(\log^* n) \cdot \log_2^3 n_m = O((\log^* n)^{4})$ and so when we apply Theorem~\ref{thm:generic-3}, the total query time is $O(\mathrm{OPT} + n + |\sigma|\log g(n,m)) = O(\mathrm{OPT} + n+ |\sigma| \log (\log^* n))$.  
	
	Suppose that $\log_2 n_m \geq \log_2^* n + 1$.  Then, $f(n_m) = (\log_2 n_m)^{c_1}
	\geq (\log_2^* n + 1) \cdot (\log_2 n_m)^{c_1-1} \geq m (\log_2 n_m)^{c_2+5} = g(n,m)(\log_2 n_m)^{c_2+2}$ as $c_1 = 31$ and $c_2 = 25$.  By Theorem~\ref{thm:generic-3}, if $n_m \geq (c_1^2+c_1)^{c_1-1}$, the total query time is $O(\mathrm{OPT} + n)$; if $n_m < (c_1^2+c_1)^{c_1-1}$, then $n_m = O(1)$ and $g(n,m) = m\log_2^3 n_m = O(m) = O(\log^* n)$, which implies that the total query time is $O(\mathrm{OPT} + n + |\sigma| \log g(n,m)) = O(\mathrm{OPT} + n + |\sigma|\log(\log^* n))$.  
\end{proof}

\section{Conclusion}

We have developed self-adjusting point location structures for convex and connected subdivisions under the point-line comparision model.  The structure for convex subdivision is optimal within the class of linear decision trees.  For connected subdivisions, there is still a very slowly growing overhead of $O(\log(\log^* n))$ per query.  The first open problem is thus to design a structure for connected subdivisions that is optimal within the class of linear decision trees.  Another research direction is to allow the subdivision to be updated efficiently while still offering an amortized query time that is sensitive to the query distribution.

\cancel{
The performance of our data structure is asymptotically optimal when compared with static point location linear decision trees.  It is an open problem to
obtain optimal performance when compared with linear decision trees that may
reorganize themselves.  This open problem may be difficult as it is related to
the dynamic optimality conjecture by Sleator and Tarjan~\cite{sleator85}, which
conjectures that the performance of a splay tree is no more than $O(n)$ plus a
constant times the time required by any binary search tree algorithm.  The
dynamic optimality conjecture is still open after over thirty years.
}

\end{document}